\documentclass[a4paper]{easychair}

\makeatletter
\let\C\@undefined
\let\G\@undefined
\let\xhookrightarrow\@undefined
\makeatother

\usepackage{sectsty}
\usepackage{lipsum}
\usepackage{epsfig}
\usepackage{complexity}
\usepackage{multirow}
\usepackage{algorithmicx}
\usepackage{algpseudocode}
\usepackage{algorithm}
\usepackage{url}
\usepackage{btran}
\usepackage{pstricks}
\usepackage{pst-tree}
\usepackage{hyperref}

\usepackage{amsmath,amsthm,amssymb,stmaryrd}
\usepackage{booktabs}

\usepackage[T1]{fontenc}
\usepackage[draft]{minted}
\usepackage{upquote}
\AtBeginDocument{%
}

\newcommand\anonornot[2]{#1}

\newtheorem{proposition}{Proposition}
\newtheorem{theorem}{Theorem}
\newtheorem{lemma}{Lemma}
\newtheorem*{fact}{Fact}

\newtheorem*{convention}{Convention}
\newtheorem*{notation}{Notation}

\theoremstyle{remark}
\newtheorem{remark}[theorem]{Remark}

\usepackage{paralist}

\usepackage{mathtools}
\usepackage{amsfonts}
\usepackage{xspace}
\usepackage{pstricks}

\newcommand{\CLASS}{\ensuremath{\mathcal{K}}}

\newcommand\etal{~\textit{et al}.\xspace}

\newcommand\set[1]{\left\{{#1}\right\}}

\newcommand\tuple[1]{\left({#1}\right)}
\newcommand\setcomp[2]{\left\{\ {#1}\ \left|\ {#2}\ \right.\right\}}

\newcommand\sizeof[1]{\left|{#1}\right|}

\newcommand\pds{\mathcal{P}}
\newcommand\controls{\mathcal{Q}}
\newcommand\talphabet{\Sigma}
\newcommand\salphabet{\Gamma}
\newcommand\oalphabet{\Gamma}
\newcommand\pdsrules{\Delta}

\newcommand\sconfig[2]{\left({#1}, {#2}\right)}

\newcommand\vecU{\vec{u}}
\newcommand\vecV{\vec{v}}
\newcommand\vecW{\vec{w}}

\newcommand{\varsetV}{\mathcal{V}}
\newcommand{\varsetW}{\mathcal{W}}

\newcommand{\defn}[1]{\textit{#1}}

\newcommand{\N}{\ensuremath{\mathbb{N}}}

\newcommand{\OMIT}[1]{}

\newcommand{\CM}{\ensuremath{\mathcal{M}}}
\newcommand{\TM}{\ensuremath{\mathcal{M}}}

\newcommand{\trs}{\ensuremath{\mathcal{R}}}
\newcommand{\trsclass}{\ensuremath{\text{\textsc{Trs}}_0}}
\newcommand{\Gtrsclass}{\ensuremath{\text{\textsc{Trs}}}}

\newcommand{\tree}{\ensuremath{T}}
\newcommand{\trees}{\ensuremath{\text{\textsc{Tree}}}}
\newcommand{\treedomain}{\ensuremath{D}}
\newcommand{\treelabeling}{\ensuremath{\lambda}}
\newcommand{\treelabels}{\ensuremath{\Sigma}}
\newcommand{\addchild}[1]{\hbox{\tt AddChild}(#1)}

\newcommand{\addclass}[1]{\hbox{\tt AddClass}(#1)}
\newcommand{\addleftkin}[1]{\hbox{\tt AddLeftKin}(#1)}
\newcommand{\addrightkin}[1]{\hbox{\tt AddRightKin}(#1)}
\newcommand{\removeclass}[1]{\ensuremath{\hbox{\tt RemoveClass}(#1)}}
\newcommand{\removenode}{\ensuremath{\hbox{\tt RemoveNode}}}

\newcommand{\post}{\ensuremath{post}}

\newcommand{\blank}{\ensuremath{\sqcup}}

\newcommand{\trleq}{\preceq}

\newcommand{\Rtrsclass}{\ensuremath{\text{\textsc{Trs}}_0'}}

\newcommand{\Root}{\text{\texttt{root}}}
\newcommand{\Team}{\text{\texttt{team}}}

\newcommand{\PlayerOne}{\text{\texttt{P1}}}
\newcommand{\PlayerTwo}{\text{\texttt{P2}}}
\newcommand{\Success}{\text{\texttt{success}}}
\newcommand{\Pop}{\text{\texttt{pop}}}

\newcommand{\lra}{\leftrightarrow}
\newcommand{\LRA}{\Leftrightarrow}

\newcommand{\Path}{\pi}

\newcommand{\CONF}{\text{\texttt{CONF}}}

\newcommand\dir[1]{{\langle #1 \rangle}}
\newcommand\inc{{\sf inc}}
\newcommand\dec{{\sf dec}}

\newcommand{\myclass}{\underline}

\newcommand{\EnoughTime}[1]{}

\newif\ifdraft\draftfalse
\ifdraft
\newcommand\al[1]{{\color{blue}
[#1 - \textbf{Anthony}]}}
\newcommand\lo[1]{{\color{purple}
[#1 - \textbf{Luke}]}}
\newcommand\mh[1]{{\color{orange} [#1 - \textbf{Matt}]}}
\newcommand\alchanged[1]{{\color{blue}{#1}}}

\newcommand\mhchanged[1]{{\color{red}{#1}}}
\else
\newcommand\todo[1]{}
\newcommand\al[1]{}
\newcommand\lo[1]{}
\newcommand\mh[1]{}
\newcommand\alchanged[1]{#1}

\newcommand\mhchanged[1]{#1}
\fi

\newcommand\shortlong[2]{#2}

\newcommand\acmeasychair[2]{#2}

\shortlong{}{\pagestyle{plain}}

\begin{document}
\pagestyle{empty}

\title{Detecting Redundant CSS Rules in HTML5 Applications: A Tree Rewriting Approach}
\titlerunning{Detecting Redundant CSS Rules in HTML5}
\author{Matthew Hague\inst{1} \and
        Anthony Widjaja Lin\inst{2} \and
        C.-H. Luke Ong\inst{3}}

\institute{
  Royal Holloway, University of London \and
  Yale-NUS College \and
  University of Oxford
}

\authorrunning{M. Hague, A. W. Lin, and C.-H. L. Ong}

\maketitle

HTML5 applications normally have a large set of CSS (Cascading Style Sheets) rules
for data display. Each CSS rule consists of a node selector and a declaration block
(which assigns values to selected nodes' display attributes).
As web applications evolve, maintaining CSS files can easily become
problematic. Some CSS rules will be
replaced by new ones, but
these obsolete (hence redundant) CSS rules often remain in the
applications. Not only does this ``bloat'' the applications
\mhchanged{
    -- increasing the bandwidth requirement --
}
but it also
significantly increases web browsers' processing time.
Most works on detecting redundant CSS rules in HTML5 applications do not
consider the dynamic behaviors of HTML5 (specified in JavaScript); in fact, the
only proposed method
that takes these into account is dynamic analysis, which cannot soundly prove
redundancy of CSS rules. In this paper,
we introduce an abstraction of HTML5 applications based on monotonic tree-rewriting
and study its ``redundancy problem''.
We establish the precise complexity of the problem and various subproblems of
practical importance (ranging from $\P$ to $\EXP$). In particular,
our algorithm relies on an efficient reduction to an analysis of symbolic pushdown
systems (for which highly optimised solvers are available), which yields a
fast method for checking redundancy in practice. We implemented
our algorithm and demonstrated its efficacy in detecting redundant CSS rules in
HTML5 applications.

\section{Introduction}

HTML5 is the latest revision of the HTML standard of the World Wide Web Consortium
(W3C), which has become a standard markup language of the Internet.
HTML5 provides a uniform framework for designing a web application:
(1) data content is given as a standard HTML tree, (2) rules for data display
are given in Cascading Style Sheets (CSS), and (3) dynamic behaviors are
specified through JavaScript.

An HTML5 application normally contains a large set of CSS rules
for data display, each consisting of a \emph{(node) selector} given in an
XPath-like query language and a \emph{declaration block} which assigns values to
selected nodes' display attributes.
\alchanged{However, many of these styling rules are often redundant (in the
sense of unreachable code), which ``bloat'' the application.
As a web application evolves, some rules will be replaced by new rules and
developers often forget to remove obsolete rules. Another cause of redundant
styling rules is the common use of HTML5 boilerplate (e.g. WordPress) since
they include many rules that the application will not need.}
A recent case study \cite{MesbahM12} shows that in several
industrial web applications \alchanged{on average} 60\% of the CSS rules are
redundant.
These bloated applications are not only harder to maintain, but they also
\mhchanged{
    increase the bandwidth requirement of the website and
}
significantly increase web browsers' processing time. In fact,
a recent study \cite{MB10} reports that when web browsers are loading popular pages
around 30\% of the CPU time is spent on CSS selectors (18\%) and parsing (11\%).
[These numbers are calculated \emph{without} even including the
extra 31\% uncategorised operations of the total CPU time, which could include
operations from these two categories.]
This suggests
the importance of detecting and removing redundant CSS rules in an HTML5
application.
Indeed, a sound and automatic redundancy checker would allow
bloated CSS stylesheets to be streamlined during development, and generic
stylesheets to be minimised before deployment.

There has been a lot of work on optimising CSS (e.g.~\cite{MTM14,MB10,MesbahM12,GenevesLQ12,BGL14}), which include merging
``duplicated'' CSS rules, refactoring CSS declaration blocks, and simplifying
CSS selectors, to name a few. However, most of these works analyse the set of CSS
rules \emph{in
isolation}. In fact, the only available methods
\alchanged{(e.g. Cilla~\cite{MesbahM12} and UnCSS~\cite{UnCSS})}
that take into account the dynamic
nature of HTML5 introduced by JavaScript are \alchanged{based on} simple \emph{dynamic
analysis}
(a.k.a. testing), which cannot soundly prove redundancy of CSS rules since such
techniques cannot in general test all possible behaviors of the HTML5
application.
For example, from the benchmarks of Mesbah and Mirshokraie~\cite{MesbahM12}
there are some non-redundant
CSS rules that their tool Cilla falsely identifies as redundant, e.g.,
due to
\mhchanged{%
    the use of JavaScript to compensate for
}%
browser-specific behavior under certain HTML5 tags like \texttt{<input/>}
(see Section \ref{sec:experiments} for more details).  Removing such
rules can distort the presentation of HTML5 applications, which is undesirable.

    \paragraph{Static Analysis of JavaScript}

    A different approach to identifying redundant CSS rules by using \emph{static
    analysis} for HTML5.  Since JavaScript is a Turing-complete programming
    language, the best one can hope for is approximating the behaviors of HTML5
    applications.  Static analysis of JavaScript code is a challenging
    goal, especially in the presence of libraries like jQuery.
    The current state of the art is well surveyed by Andreasen and
    M{\o}ller~\cite{AM14}, with the main tools in the field being
    WALA~\cite{SSDT13,SDCST12} and TAJS~\cite{AM14,JMT09,JMM11}. These tools
    (and others)
    provide traditional static analysis frameworks encompassing features such as
    points-to~\cite{JC09,SDCST12} and determinacy analysis~\cite{AM14,SSDT13}, type
    inference~\cite{jQuery-ecoop13,JMT09} and security
    properties~\cite{Gatekeeper,Actarus}.  The modelling of the HTML DOM is
    generally treated as part of the heap abstraction~\cite{JMM11,Actarus} and
    thus the tree structure is not precisely tracked.

    For the purpose of soundly identifying redundant CSS rules, we need a
    technique for computing a symbolic representation of an
    \emph{overapproximation} of the set of all reachable HTML trees that is
    sufficiently precise for real-world applications.  Currently there is no
    \emph{clean} abstract model that captures \emph{common} dynamics of the HTML
    (DOM) tree caused by the JavaScript component of an HTML5 application and at
    the same time is \emph{amenable to algorithmic analysis}. Such a model is not
    only important from a theoretical viewpoint, but it can also serve as a useful
    \emph{intermediate language} for the analysis of HTML5 applications which
    among others can be used to identify redundant CSS rules.

\paragraph{Tree rewriting as an intermediate language.}
The tree-rewriting paradigm --- which is commonly used in databases
(e.g.  \cite{MBPS05,FGN08,ABM04,GMW10,GMSZ08,ASV09}) and verification
(e.g.  \cite{Hague14,Lin12,Loding06,LS07,rtmc,Lin12b}) --- offers a clean
theoretical framework
for modelling the dynamics of tree updates and usually lends itself to
fully-algorithmic analysis. This makes tree-rewriting a suitable framework in which
to model the dynamics of tree updates commonly performed by HTML5 applications.
Surveying real-world HTML5 applications (including
Nivo-Slider~\cite{nivo-slider} and real-world examples from the benchmarks in
Mesbah and Mirshokraie~\cite{MesbahM12}),
we were surprised
to learn that one-step tree updates used in these applications are extremely
simple, despite the complexity of the JavaScript code from the point of view of
static analysers.  That
said, we found that these updates are
\emph{not} restricted to modifying only
certain regions of the HTML tree. As a result, models such as \emph{ground tree
rewrite systems} \cite{Loding06} and their extensions
\cite{LS07,Lin12,Hague14,Mayr00,GL14} (where \emph{only} the bottom part of the
tree may be modified) are not appropriate.
However, systems with rules that may rewrite nodes in \emph{any} region of
a tree are problematic since they render the simplest problem of reachability
undecidable. Recently, owing to the study of \emph{active XML}, some
restrictions that admit decidability of verification (e.g.
\cite{ABM04,GMSZ08,GMW10,ASV09}) have been obtained. 
However, these models
have very high complexity (ranging from double exponential time to nonelementary),
which makes practical implementation difficult.

\paragraph{Contributions.}
The main contribution of the paper is to give a simple and clean tree-rewriting
model which strikes a good balance between: (1) expressivity
in capturing the dynamics of tree updates commonly performed in HTML5 applications
(esp. insofar as detecting redundant CSS rules is concerned),
and (2) decidability and complexity of rule redundancy analysis
(i.e. whether a given rewrite rule
can ever be fired in a reachable tree).
We show that the complexity of the problem is $\EXP$-complete\footnote{These
complexity classes are defined below and we describe the roles they play in
our investigation.}, though under various
practical restrictions the complexity becomes $\PSPACE$ or even $\P$.
This is substantially better than the complexity of the more powerful tree rewriting
models studied in the context of active XML, which is at least double-exponential
time.
Moreover, our algorithm relies on
an efficient reduction to a reachability analysis in \emph{symbolic pushdown
systems} for which highly optimised solvers (e.g. Bebop \cite{bebop},
Getafix \cite{getafix}, and Moped \cite{moped}) are available.

We have implemented our reduction, together with a proof-of-concept translation
tool from HTML5 to our tree rewriting model.
\mhchanged{
    Our translation by no means captures the full feature-set of
    JavaScript and is simply a means of testing the underlying model and
    analysis we introduce\footnote{
        Handling JavaScript in its full generality is a difficult
        problem~\cite{AM14}, which is beyond the scope of this paper.
    }.%
}
We specifically focus on modelling standard features of \emph{jQuery}
\cite{jQuery} --- a simple JavaScript library that makes HTML
document traversal, manipulation, event handling, and animation easy from a web
application developer's viewpoint.  Since its use is so widespread in HTML5
applications nowadays (e.g., used in more than half of the top hundred thousand
websites \cite{jQuery-widespread}) some authors \cite{jQuery-ecoop13} have
advocated a study of jQuery as a language in its own right.  Our experiments
demonstrate the efficacy of our techniques in detecting redundant CSS rules in
HTML5 applications. Furthermore, unlike dynamic analysis, our techniques will
not falsely report CSS rules that \emph{may} be invoked as redundant (at least
within the fragment of HTML5 applications that our prototypical implementation
can handle).  \alchanged{We demonstrate this on a number of non-trivial examples
(including a specific example from the benchmarks of Mesbah and
Mirshokraie~\cite{MesbahM12} and an HTML application using the image
slider package Nivo-Slider \cite{nivo-slider}).}

\paragraph{Connection with existing works on static analysis of JavaScript.}
    As surveyed by Andreasen and M{\o}ller~\cite{AM14}, the static analysis of
    JavaScript in the presence of the jQuery library --- which is essential
    for analysing HTML5 applications --- is currently a formidable task for
    existing static analysers.  We consider our work to be complementary to
    these works: first, our tree-rewriting model may form part of a static analysis
    abstraction, and second, static analysis will be essential in translating HTML5
    applications into our tree rewriting model by extracting accurate tree update
    operations.  In our implementation, we provide an ad-hoc extraction of tree
    update operations that allows us to demonstrate the applicability of our
    approach.  Creating a robust static analysis that achieves this is an
    interesting and worthwhile research challenge, which might benefit from the
    recent remarkable effort by Bodin\etal~\cite{Maffeis} of capturing the full
    JavaScript semantics and verifying it with Coq.

\paragraph{Organisation.}
We give a quick overview of HTML5 applications via a simple example in Section
\ref{sec:html5-overview}. We then introduce our tree rewriting model in Section
\ref{sec:trs}. Since the
general model is undecidable, we introduce a monotonic abstraction in Section
\ref{sec:monotonic}. We provide an efficient reduction from an analysis of the
monotonic abstraction to symbolic pushdown systems in Section \ref{sec:upper}.
Experiments are reported in Section \ref{sec:experiments}.  We conclude with future
work in Section \ref{sec:conclusion}.
Missing proofs and further technical details can be found in the \shortlong{full version \cite{full}}{appendix}.

\paragraph{Notes on computational complexity.}
In this paper, we study not only decidability but also the \emph{complexity}
of computational problems. We believe that pinpointing the precise complexity of
verification problems is not only of fundamental importance, but also it often
suggests algorithmic techniques that are most suitable for attacking the
problem in practice. In this paper, we deal with the following computational
complexity classes (see \cite{Sipser} for more details): $\P$ (problems solvable in
polynomial-time), $\PSPACE$ (problems solvable in polynomial space), and
$\EXP$ (problems solvable in exponential time).
Verification problems that have
complexity $\PSPACE$ and $\EXP$ or beyond --- see
\cite{schwoon-thesis,MONA} for a few examples --- have substantially benefited from
techniques like symbolic model checking \cite{McMillan}. The redundancy problem
that we study in this paper is another instance of a hard computational problem that
can be efficiently solved by BDD-based symbolic techniques in practice.

\section{HTML5: a quick overview}
\label{sec:html5-overview}

In this section, we provide a brief overview of a simple HTML application.
We assume basic familiarity with the static elements of HTML5, i.e., HTML documents
(a.k.a. HTML DOM document objects) and CSS rules (e.g. see \cite{w3schools}). We
will discuss their formal
models in Section \ref{sec:trs}.
Our example (also available \anonornot{\mhchanged{at the
URL~\cite{simple-html-anon}}}{at the URL \cite{simple-html}}) is a small modification of an example
taken from an online tutorial \cite{sanwebe-addremove}, which is given in Figure
\ref{fig:sanwebe-example}.  To better understand the application, we suggest the
reader open it with a web browser and interact with it.

\begin{figure}
\begin{center}
    \acmeasychair{\input{html5-example-acm}}
                 {\input{html5-example-easychair}}
    \caption{\label{fig:sanwebe-example}A simple HTML5 application (see
    \anonornot{\mhchanged{\cite{simple-html-anon}}}
              {\cite{simple-html}}).}
\end{center}
\end{figure}

\mh{I couldn't get latex escaping working with labels in the source code, so if
the line numbers change we have to update them manually!}

In this example the
page displays a list of input text boxes contained in a \verb+div+ with class
\verb+input_wrap+.  The user can add more input boxes by
clicking the ``add field'' button, and can remove a text box by clicking its
neighbouring ``remove'' button. The script, however, imposes a limit (i.e. 10)
on the number of text boxes that can be added. If the user attempts to add
another text box when this limit is reached, the \verb+div+ with ID \verb+limit+
displays the text ``Limits reached'' in red.

This dynamic behavior is specified within the second
\verb+<script/>+ tag starting at Line 11 (the first simply loads the jQuery library).
To understand the script, we will provide a quick overview of jQuery calls
(see \cite{jQuery} for more detail). A simple jQuery call may take the form
\begin{minted}{javascript}
         $(selector).action(...);
\end{minted}
where `\$' denotes that it is a jQuery call, \verb+selector+ is a CSS selector, and
\verb+action+ is a rule for modifying the subtree rooted at this node.
For example, in Figure \ref{fig:sanwebe-example}, Line 29, we have
\begin{minted}{javascript}
         $('#limit').addClass('warn');
\end{minted}
The CSS selector \verb+#limit+ identifies the unique node in the tree with ID
\verb+limit+, while the \verb+addClass()+ call adds the class \verb+warn+ to this node.
The CSS rule
\begin{minted}{css}
        .warn { color: red }
\end{minted}
appearing in the head of the document at Line 2 will now match the node, and thus its
contents will be displayed in red.

Another simple example of a jQuery call in Figure \ref{fig:sanwebe-example} is
at Line 37
\begin{minted}{javascript}
         $('.warn').removeClass('warn');
\end{minted}
The selector \verb+.warn+ matches \emph{all}\footnote{Unlike node IDs, a single
class might be associated with multiple nodes} nodes in the tree with class
\verb+warn+.   The call to \verb+removeClass()+ removes the class \verb+warn+ from
these nodes. Observe that when this call is invoked, the CSS rule above will no
longer be matched.

Some jQuery calls may contain an event listener. E.g. at Line 16 we have
\begin{minted}{javascript}
        $('.button').click(...);
\end{minted}
in Figure \ref{fig:sanwebe-example}. This specifies that the function in
`\verb+...+' should fire when a node with class \verb+button+ is clicked.
Similarly at Line 33,
\begin{minted}{javascript}
$('.input_wrap').on('click', '.delete', ...);
\end{minted}
adds a click listener to any node within the \verb+input_wrap+ div
that has the class \verb+delete+.

In general, jQuery calls might form chains.  E.g. at Line 34 we have
\begin{minted}{javascript}
         $(this).parent('div').remove();
\end{minted}
In this line, the call \verb+$(this)+ selects the node which has been clicked.
The call to \verb+parent()+ and then \verb+remove()+ moves one step up the tree and
if it finds a \verb+div+ element, removes the entire subtree (which is of the form
\texttt{<div><input/><a></a></div>}) from the document.

In addition to the action \verb+remove()+ which erases an entire subtree from the
document, Figure \ref{fig:sanwebe-example} also contains other actions that
potentially modify the shape of the HTML tree. The first such action is
\verb+append(string1)+ (at Line 19), which simply appends \verb+string1+ at the \emph{end} of the
string inside the selected node tag. Of course, \verb+string1+ might represent
an HTML tree; in our example, it is a tree with three nodes. So, in effect
\verb+append()+ adds this tree as the right-most child of the selected node.
The second such action is \verb+html(string1)+ (e.g. at Line 25), which first erases the string
inside the selected node tag and then appends it with \verb+string1+. In effect,
this erases all the descendants of the selected node and adds a \emph{forest}
represented by \verb+string1+.

\mhchanged{
    \begin{remark}
    \label{rem:redundant-rule}
        An example where the CSS rule in Figure \ref{fig:sanwebe-example}
        becomes redundant is when the limit
        on the number of boxes is removed from the application (in effect, removing \texttt{x}),
        but the CSS is not updated to reflect the change (e.g. see
        \anonornot{\mhchanged{\cite{simple-html-nolimit-anon}}}{\cite{simple-html-nolimit}}).
    \end{remark}
}

In general, CSS selectors are non-trivial.  For example
\begin{minted}{css}
        .a.b .c { color: red }
\end{minted}
matches all nodes with class \verb+c+ and some ancestor containing both
classes \verb+a+ and \verb+b+ (the space indicates \verb+c+ appears on a
descendant).  Thus, detecting redundant CSS rules requires a
good knowledge of the kind of trees constructed by the application.  In
practice, redundant CSS rules easily arise when one modifies a sufficiently
complex HTML5 application (the size of the top 1000 websites has recently
exceeded 1600K Bytes
\cite{average-web-page}).
Some popular web pages are known to have
\mhchanged{an average of}
60\% redundant CSS rules, as
suggested by recent case studies \cite{MesbahM12}.


\newcommand\charfn[1]{{[#1]}}
\newcommand\PROP{\hbox{\sc Prop}_{\CLASS}}

\section{A tree-rewriting approach}
\label{sec:trs}

In this section, we present our tree-rewriting model. Our design philosophy is
to put a special emphasis on model simplicity and fully-algorithmic analysis with
good complexity, while retaining adequate expressivity in modelling common
tree updates in HTML5 applications (insofar as detecting redundant CSS rules is
concerned). We will start by giving an informal description of our approach and
then proceed to our formal model.

\subsection{An informal description of the approach}

\paragraph{Data representation}
The data model of HTML5 applications is the standard HTML (DOM) tree.
    In designing our tree rewriting model, we will adopt a data representation
    consisting of a finite set $\CLASS$ of \emph{classes}, and an unordered,
    unranked tree with the set $2^\CLASS$ of node labels.  An unordered tree
    does not have a sibling ordering, and an unranked tree does not fix the
    number of children of its nodes.
\alchanged{
Since jQuery and CSS selectors may reason about
adjacent node siblings, unordered trees are in general only an
\emph{overapproximation} of HTML trees. As we shall see later, the consequence
of this approximation is that some CSS rules that we identify in our analysis
as non-redundant
might turn out to be redundant when sibling ordering is accounted for, though
all CSS rules that we identify as redundant will \emph{definitely} be redundant
even with the sibling ordering (see Remark \ref{rm:ordered} in Section
\ref{sec:monotonic}).
Although it is possible in theory to extend our techniques to ordered
unranked trees, we choose to use unordered trees in our model for the purpose of
simplicity.
}
Not only do unordered trees give a clean data model, but they turn
out to be sufficient for analysing redundancy of CSS rules in most HTML5
applications.
\mhchanged{
    In the examples we studied, no false positives were reported as a consequence of the unordered approximation.
}%
The choice of tree labels is motivated by CSS and HTML5. Nodes in
an HTML document are tagged by HTML elements (e.g.  \texttt{div} or \texttt{a})
and associated with a set of classes, which can be added/removed by HTML5
scripts.  Node IDs and data attributes are also often assigned to specific
nodes, but they tend to remain unmodified throughout the execution of the
application and so can conveniently be treated as classes.

\paragraph{An ``event-driven'' abstraction}
Our tree-rewriting model is an ``event-driven'' abstraction of the script component
of HTML5 applications.
The abstraction consists of a (finite) set of tree-rewrite rules that can be
fired \emph{any} time in \emph{any} order (so long as they are enabled). In
this abstraction, one can imagine that each rewrite rule is associated with an
external event listener (e.g. listening for a mouse click, hover, etc.). Since
these external events cannot be controlled by the system, it is standard to
treat them (e.g. see \cite{reactive-systems-book}) as \emph{nondeterministic}
components, i.e., that they can occur concurrently and in any order\footnote{
    Note, although, in our model, each rewrite rule is executed atomically, an
    event that leads to two or more tree updates will be modelled by several
    rewrite rules.  Our analysis will be a ``path-insensitive''
    over-approximation in that no ordering or connection is maintained between
    these individual update rules.  Thus, an event leading to several tree
    updates is not assumed to be handled atomically.  Indeed, the order of the
    updates is also not maintained.  In our experiments this over-approximation
    did not lead to false positives in the analysis.}
Incidentally, the case for event-driven abstractions has been made in the
context of transformations of XML data \cite{BPW02}.

A tree-rewrite rule $\sigma$ in our rewrite systems is a tuple $(g,\chi)$ consisting
of a node selector $g$ (a.k.a. \emph{guard}) and a rewrite operation $\chi$.

To get a feel for our approach, we will construct an event-driven abstraction for
the script component of the HTML5 example in Figure \ref{fig:sanwebe-example}.
For simplicity, we will now use jQuery calls as tree-rewrite rules. 
We will formalise them later.

The event-driven abstraction for the example in Figure \ref{fig:sanwebe-example}
contains four rewrite rules as follows:
\begin{minted}{javascript}
(1) $('#limit').addClass('warn');
(2) $('.warn').removeClass('warn');
(3) $('.input_wrap').append('<div>
      <input/><a class="delete"></a></div>');
(4) $('.input_wrap').find('.delete').
      parent('div').remove();
\end{minted}
Note that we removed irrelevant attributes (e.g. \verb+href+) and text contents
since they do not affect our analysis of redundant CSS rules. Rules (1)--(3)
were extracted directly from the script. However, the extraction of Rule (4) is
more involved.  First, the calls to \verb+parent()+ and \verb+remove()+ come directly
from the script.  Second, the other calls ---
which
select all elements with class \verb+delete+ that are descendants of a node with class
\verb+input_wrap+ --- derive from the semantics of \verb+on()+.  The connection of
the two parts arrives because the jQuery selection is passed to the event
handler via the \verb+this+ variable.  This connection may be inferred by
a data-flow analysis that is sensitive to the behaviour of jQuery.

\paragraph{Detecting CSS Redundancy}
It can be shown that the set $S_1$ of all reachable HTML trees in the example
in Figure \ref{fig:sanwebe-example} is a \emph{subset} of the set $S_2$ of all
HTML trees that can be reached by applying Rules (1)--(4) to the
initial HTML document.  We may detect whether
\begin{minted}{css}
    .warn { color: red }
\end{minted}
is redundant by checking whether its selector may match some part of a tree in
$S_2$.  If not, then since $S_1 \subseteq S_2$ we can conclude that the rule is
\emph{definitely} redundant.  In contrast, if the rule can be matched in $S_2$, we
\emph{cannot} conclude that the rule is redundant in the original application.

Let us test our abstraction. First, by applying Rule (1) to the initial HTML
tree, we confirm that \verb+warn+ can appear in a tree in $S_2$ and hence the
CSS rule \emph{may} be fired.  We now revisit the scenario in
\mhchanged{
    Remark~\ref{rem:redundant-rule}
}
in Section \ref{sec:html5-overview} where the limit on the number of boxes is
removed, but the CSS is not updated. In this case, the new event-driven
abstraction for the modified script will not contain Rule (1) and the CSS rule
can be seen to be redundant in $S_2$.  This necessarily implies that the rule is
\emph{definitely} redundant in $S_1$.

Thus, we guarantee that redundancies will not be falsely identified, but may fail
to identify some redundancies in the original application.

\subsection{Notations for trees}
Before defining our formal model, we briefly fix our notations for describing
trees. In this paper we use unordered, unranked trees.
A \defn{tree domain} is a nonempty 
finite subset $\treedomain$ of $\N^*$ 
(i.e. the set of all strings over the alphabet $\N = \{0, 1, \cdots \}$)
satisfying prefix-closure,
i.e., $w \cdot i \in \treedomain$ with
$i \in \N$ implies $w \in \treedomain$. Note that the natural linear order
of $\N$ is immaterial in this definition, i.e., we could use any
countably infinite set in place of $\N$.

A \defn{(labeled) tree} over the nonempty finite set (a.k.a. alphabet)
$\treelabels$ is a tuple $\tree = (\treedomain,\treelabeling)$
where $\treedomain$ is a tree domain and
$\treelabeling$ is a mapping (a.k.a. \defn{node labeling}) from $\treedomain$
to $\treelabels$.
We use standard terminologies for trees, e.g., parents,
children, ancestors, descendants, and siblings. The \defn{level of a node $v
\in\treedomain$} in $\tree$ is $|v|$. Likewise, the \defn{height
of the tree $\tree$} is $\max\{|v| : v \in \treedomain\}$.
Let $\trees(\treelabels)$ denote the set of 
trees over $\treelabels$. For every $k \in \N$, we define
$\trees_k(\treelabels)$ to be the set of trees of height $k$.

If $\tree = (\treedomain,\treelabeling)$ and $v \in \treedomain$,
the \defn{subtree of $\tree$ rooted at $v$} is the tree
$\tree_{\mid_v} = (\treedomain',\treelabeling')$, where
$\treedomain' := \{ w \in \N^* : vw \in \treedomain \}$ and
$\treelabeling'(w) := \treelabeling(vw)$.

\mhchanged{
    We remark, for example, that in our definitions, the trees
    $\tree_1 = \tuple{\set{\varepsilon, 1}, \treelabeling_1}$
    and
    $\tree_2 = \tuple{\set{\varepsilon, 2}, \treelabeling_2}$
    with
    $\treelabeling_1(\varepsilon) = \treelabeling_2(\varepsilon)$
    and
    $\treelabeling_1(1) = \treelabeling_2(2)$
    define distinct trees, although both trees contain a root node with a single
    child with the same labels.  It is easy to see that our guards discussed in
    the following sections cannot distinguish trees up to isomorphism.  In
    Section~\ref{sec:monotonic} we discuss morphisms between trees in the
    context of a monotonicity property.
}

\mh{removed term representation since afaik we don't use it.}

\subsection{The formal model}
We now formally define our tree-rewriting model $\Gtrsclass$ for HTML5 tree updates.
A \defn{rewrite system} $\trs$ in $\Gtrsclass$ is a (finite) set of \defn{rewrite
rules}.  Each rule $\sigma$ is
a tuple $(g,\chi)$ of a guard $g$ and a (rewrite) operation $\chi$. Let us define
the notion of guards and rewrite operations in turn.

Our language for guards is simply modal logic with special types of modalities.
It is a subset of \emph{Tree Temporal Logic}, which is a formal model of the
query language XPath for XML data \cite{Sch92,Marx05,libkin-survey}.
More formally,
a \defn{guard} over the node labeling $\treelabels = 2^{\CLASS}$ with $\CLASS = \{c_1, \cdots, c_n \}$ can be defined by the following grammar:
\[
    g ::= \top\ |\ c\ |\ g \wedge g\ |\ g \vee g\ |\ \neg g\ |\ \langle
        d\rangle g
\]
where
$c$ ranges over $\CLASS$ and $d$ ranges over
$\{\uparrow,\uparrow^\ast,\downarrow,\downarrow^\ast\}$,
standing for parent, ancestor, child, and descendant respectively.
\mhchanged{
    Note, we will also use $\langle
    \downarrow^+ \rangle g$ as shorthand for the formula $\langle \downarrow^\ast
    \rangle\langle \downarrow \rangle g$ and similarly for $\langle \uparrow^+
    \rangle g$.
}%
The guard $g$ is said to be \defn{positive} if there is no occurrence of $\neg$ in
$g$.
Given a tree $\tree = (\treedomain,\treelabeling)$ and a node
$v \in \treedomain$, we define whether $v$ \defn{matches} a guard $g$
(written $v, T \models g$)
\mhchanged{
    below.
}
Intuitively, we interpret $v,T \models c$
(for a class $c \in \CLASS$) as $c \in \treelabeling(v)$, and each modality $\langle
d\rangle$ \mhchanged{
    (where $d \in \{\uparrow,\uparrow^\ast,\downarrow,\downarrow^\ast\}$) in
    accordance with the arrow orientation.
}

\mhchanged{
    Given a tree $\tree = (\treedomain,\treelabeling)$ and a node
    $v \in \treedomain$, we define whether $v$ of $T$ \defn{matches} a guard $g$
    (written $v, T \models g$) by induction over the following rules:
    \begin{itemize}
    \item
        $v,\tree \models \top$.
    \item
        $v,\tree \models c$ if $c \in \treelabeling(v)$.
    \item
        $v,\tree \models g \wedge g'$ if $v,\tree \models g$ and
        $v,\tree \models g'$.
    \item
        $v,\tree \models g \vee g'$ if $v,\tree \models g$ or
        $v,\tree \models g'$.
    \item
        $v,\tree \models \neg g$ if it is not the case that $v,\tree
        \models g$.
    \item
        $v, \tree \models \langle \uparrow \rangle g$ if
        $v = w.i$ (for some $i \in \N$), and $w, \tree \models g$.
    \item
        $v, \tree \models \langle \uparrow^\ast \rangle g$ if
        $v = w.w'$ (for some $w' \in \N^\ast$), and $w, T \models g$.
    \item
        $v, \tree \models \langle \downarrow \rangle g$ if
        there exists a node $v.i \in \treedomain$ (for some $i \in \N$) such
        that $v.i, \tree \models g$.
    \item
        $v, \tree \models \langle \downarrow^\ast \rangle g$ if
        there exists a node $v.w \in \treedomain$ (for some $w \in \N^\ast$) such
        that $v.w, T \models g$.
    \end{itemize}
}
See the section below on encoding jQuery rules for some examples.

We say that
$g$ is \emph{matched} in $T$ if $v, T \models g$ for some node $v$ in $T$.
Likewise, we say that $g$ is \emph{matched} in a set $S$ of $\treelabels$-trees
if it is matched in some $T \in S$.  In the sequel, we sometimes omit mention of
the tree $T$ from $v, T \models g$ whenever there is no possibility of
confusion.

Having defined the notion of guards, we now define our rewrite operations,
which can be one of the following: (1) $\addchild{X}$, (2) $\addclass{X}$,
(3) $\removeclass{X}$, and (4) $\removenode$, where $X \subseteq \CLASS$.
Intuitively, the semantics of Operations (2)--(4) coincides with
the semantics of the jQuery actions \texttt{addClass(.)}, \texttt{removeClass(.)},
and \texttt{remove()}, respectively. Similarly,
the semantics of $\addchild{X}$ coincides
with the semantics of the jQuery action \texttt{append(str)}
in the case when
\texttt{str} represents a single node associated with classes $X$.
By adding extra classes, appending a larger subtree can be easily simulated by
several steps of $\addchild{X}$ operations.  This is demonstrated in the next section.

We now formally define the semantics of
these rewrite operations.
Given two trees $\tree = (\treedomain,\treelabeling)$ and $\tree' =
(\treedomain',\treelabeling')$, we say that $\tree$ \defn{rewrites} to $\tree'$ via
\mhchanged{
    $\sigma = \tuple{g, \chi}$
}
(written $T \to_\sigma T'$) if there exists a node $v \in \treedomain$
such that $v \models g$ and
\begin{itemize}
\item if $\chi = \addclass{X}$ then $\treedomain' = \treedomain$ and $\treelabeling' := \treelabeling[v \mapsto X \cup \treelabeling(v)]$.
\footnote{Given a map $f : A \to B$, $a' \in A$ and $b' \in B$, we write $f[a' \mapsto b']$ to mean the map $(f \setminus \{(a', f(a'))\}) \cup \{(a', b')\}$}

\item if $\chi = \addchild{X}$ then $\treedomain' = \treedomain \cup \{v.i\}$
and $\treelabeling' := \treelabeling[v.i \mapsto X]$
\mhchanged{
    and $v.i \notin \treedomain$
}

\item if $\chi = \removeclass{X}$ then $\treedomain' = \treedomain$ and $\treelabeling' := \treelabeling[v \mapsto \treelabeling(v) \setminus X]$

\item if $\chi = \removenode$ and $v$ is \emph{not} the root node,
    $\treedomain' := \treedomain \setminus \{ v.w : w \in \N^*\}$
    and $\treelabeling'$ is the restriction of $\treelabeling$ to
    $\treedomain'$.
\end{itemize}

\mhchanged{
    Note that the system cannot execute a $\removenode$ operation on the root
    node of a tree.  I.e. there is no transition $T \to_\sigma T'$ if $\sigma$
    would erase the root node of $T$.
}

Given a rewrite system $\trs$ over $\treelabels$-trees, we define
$\to_\trs$ to be the union of $\to_\sigma$, for all $\sigma \in \trs$.
For every $k \in \N$, we define $\to_{\trs,k}$ to be the restriction of
$\to_\trs$ to $\trees_k(\treelabels)$. Given a
set $\mathcal{C}$ of $\treelabels$-trees, we write $\post_{\trs}^*(\mathcal{C})$
(resp. $\post_{\trs,k}^*(\mathcal{C})$)
to be the set of trees $\tree'$
satisfying $\tree \to_\trs^* \tree'$ (resp. $\tree \to_{\trs,k}^* \tree'$)
for some tree $\tree \in \mathcal{C}$.

\paragraph{Encoding jQuery Rewrite Rules}
\label{sec:trs-jquery}
Let us translate the four ``jQuery rewrite rules'' for the application in
Figure \ref{fig:sanwebe-example} into our formalism. The first rule translates
directly to the rule $(\text{\texttt{\#limit}},\addclass{\{\text{\texttt{.warn}}\}})$,
while the second rule translates to $(\text{\texttt{.warn}},
\removeclass{\{\text{\texttt{.warn}}\}})$. The fourth rule identifies \texttt{div} nodes that have some child with class \texttt{delete} that in turn has some ancestor with class \texttt{input\_wrap}.  Thus, it translates to
\[
    (\text{\texttt{div}} \wedge \langle \downarrow \rangle(
    \text{\texttt{.delete}} \wedge {\langle \uparrow^+ \rangle}
    \text{\texttt{.input\_wrap}} ), \removenode).
\]
Finally, the third rule requires the construction of a new sub-tree.  We achieve
this through several rules and a new class $\texttt{tmp}$.  We first add the new
$\texttt{div}$ element as a child node, and use the class $\texttt{tmp}$ to mark
this new node:
\[
    (\texttt{.input\_wrap}, \addchild{\{\texttt{div}, \texttt{tmp}\}}) \ .
\]
Then, we add the children of the \texttt{div} node with the two rules
\[
    \begin{array}{ll}
        & (\texttt{tmp}, \addchild{\{\texttt{input}\}})  \\

        \text{and} \quad & (\texttt{tmp}, \addchild{\{\texttt{a},
        \texttt{.delete}\}}) \ .
    \end{array}
\]
\mhchanged{%
    We show in the
    \shortlong{full version}{appendix}.
    how to encode a large number of jQuery tree traversals into our guard language.
    In general, some of these traversals have to be approximated.
    For example the
}%
\protect\verb|.next()|
\mhchanged{
    operation can be approximated by the modalities
    $\langle \uparrow \rangle \langle \downarrow \rangle$
    which select a sibling of the current node.
}

\paragraph{The redundancy problem}
The \defn{redundancy problem for $\Gtrsclass$} is the problem that, given a
rewrite system $\trs$ over $\treelabels$-trees, a finite nonempty set $S$ of
guards over $\treelabels$, 
and an initial
$\treelabels$-labeled tree $\tree_0$,
compute the subset $S' \subseteq S$ of guards that are not matched in
$\post_{\trs}^*(\tree_0)$.
The decision version of the
redundancy problem for $\Gtrsclass$ is simply to check if the aforementioned set
$S'$ is empty. Similarly, for each $k \in \N$, we define the
\defn{$k$-redundancy problem for $\Gtrsclass$} to be the restriction of the
redundancy problem for $\Gtrsclass$ to trees of height $k$ (i.e. we use
$\post_{\trs,k}^*$ instead of $\post_{\trs}^*$).

The problem of identifying redundant CSS node selectors in a CSS
file can be reduced to the problem of the redundancy problem for $\Gtrsclass$.
This is because CSS node selectors can easily be translated into our guard language
(e.g. using the translation given in \cite{GenevesLQ12}).
    Observe that the converse is false.  E.g., $\langle \downarrow \rangle a
    \land \langle \downarrow \rangle b$ cannot be expressed as a CSS selector
    since $a$ and $b$ may appear on different children of the matched node. The
    guard in the fourth rule of our running example is also not expressible as a
    CSS selector.
    \alchanged{
However, this increased expressivity is required to model tree traversals
in jQuery.
}
Note that the redundancy problem for $\Gtrsclass$ could also have potential
applications
beyond
detecting redundant CSS rules, e.g., detecting redundant jQuery calls in HTML5.

Despite the simplicity of our rewrite rules, it turns out that the
redundancy problem is in general undecidable (even restricted to trees
of height at most two); see the \shortlong{full version}{appendix}.
\begin{proposition}
The 1-redundancy problem for $\trs$ is undecidable.
\label{prop:trs-undec}
\end{proposition}

\newcommand{\subsumedby}{\preceq}
\newcommand{\pred}{\hbox{\rm pred}\,}
\newcommand\sem[1]{{[\![}#1{]\!]}}

\section{A monotonic abstraction}
\label{sec:monotonic}
The undecidability proof of Proposition \ref{prop:trs-undec} in fact
relies fundamentally on the power of negation in the guards. A natural question, therefore, is
what happens in the case of positive guards. Not only is this an
interesting theoretical question, such a restriction often suffices in practice.
This is partly because the use of negations in CSS and jQuery selectors (i.e.
\verb+:not(...)+) is rather limited in practice. In particular, there was no
use of negations in CSS selectors
found in the benchmark in \cite{MesbahM12} containing 15 live web applications.
In practice, negations are almost always limited to negating atomic formulas, i.e.,
$\neg c$ (for a class $c \in \CLASS$) which can be overapproximated by
$\top$ often without losing too much precision.

\mhchanged{
    Note, in general it is not possible to express $\neg c$ via the formula
    $\bigvee_{c' \in \mathcal{K}\setminus\set{c}} c'$ since nodes may be
    labelled by multiple classes.  That is, labelling a node by $c'$ does not
    prevent the node also being labelled by $c$.  However, when $c$ is an HTML
    tag name (e.g. \texttt{div} or \texttt{img}), we can assert $\neg c$ by
    checking whether the node is labelled by some other tag name, since a node
    can only have one tag (e.g. a node cannot be both a \texttt{div} and
    an \texttt{img}).
}

A main result of the paper is that the ``monotonic''
abstraction that is obtained by restricting to positive guards
gives us decidability with a good complexity.
In this section, we prove the
resulting tree rewriting class is
``monotonic'' in a technical sense of the word, and summarise the main
technical results of the paper.

\begin{notation}
Let us denote by $\trsclass$ the set of rewrite systems with
positive guards. The guard databases in the input to redundancy and
$k$-redudancy problems for $\trsclass$ will only contain positive
guards as well. In the sequel, unless otherwise stated, a ``guard'' is
understood to mean a positive guard.
\end{notation}

\subsection{Formalising and proving ``monotonicity''}
Recall that a binary relation $R \subseteq S \times S$ is a \emph{preorder} if
it is transitive, i.e., if $(x,y) \in R$ and $(y,z) \in R$, then $(x,z) \in R$.
We start with a definition of a preorder $\trleq$ over
$\trees(\treelabels)$, where $\treelabels = 2^\CLASS$. Given two
$\treelabels$-trees $\tree = (\treedomain,\treelabeling)$ and
$\tree' = (\treedomain',\treelabeling')$, we write
$\tree \trleq \tree'$ if there exists an \defn{embedding} from $\tree'$
to $\tree$, i.e., a function $f: \treedomain \to
\treedomain'$ such that:
\begin{description}
\item[(H1)] $f(\epsilon) = \epsilon$
\item[(H2)] For each $v \in \treedomain$, $\treelabeling(v) \subseteq
\treelabeling'(f(v))$
\item[(H3)] For each $v.a \in \treedomain$ where $v \in \N^*$ and $a \in \N$,
    we have $f(v.a) = f(v).b$ for some $b \in \N$.
\end{description}
Note that this is equivalent to the standard notion of \emph{homomorphisms} from
database theory (e.g. see \cite{AHV94}) when each class $c \in \CLASS$ is treated
as a unary relation.
The following is a basic property of $\trleq$, whose proof
is easy and is left to the reader.
\begin{fact}
$\trleq$ is a preorder on $\trees(\treelabels)$.
\end{fact}
The following lemma shows that embeddings preserve positive
guards.
\begin{lemma}\label{lem:preserve-vDash}
Given trees $\tree = (\treedomain,\treelabeling)$ and $\tree' = (\treedomain',
\treelabeling') \in \trees(\treelabels)$ satisfying $\tree \trleq \tree'$ with
a witnessing embedding $f: \treedomain \to \treedomain'$, and a
positive guard $g$ over
$\treelabels$, then if $v,\tree \models g$ then $f(v),\tree' \models g$.
\label{lm:trleq}
\end{lemma}
We relegate to the \shortlong{full version}{appendix} the proof of Lemma \ref{lm:trleq}, which is similar to
(part of) the proof of the homomorphism theorem for conjunctive queries (e.g. see
\cite[Theorem 6.2.3]{AHV94}). This lemma yields the following monotonicity property
of $\trsclass$.
\begin{lemma}[Monotonicity]
For each $\sigma \in \trs$, if $\tree_1 \trleq \tree_2$ and $\tree_1 \to_\sigma
\tree'_1$, then either $\tree'_1 \trleq \tree_2$ or $\tree_2 \to_\sigma
\tree_2'$ for some $\tree_2'$ satisfying $\tree'_1 \trleq \tree_2'$.
\label{lem:monotonic}
\end{lemma}
\mhchanged{
    Intuitively, the property states that any rewriting step of the ``smaller''
    tree either does not expand the tree beyond the ``bigger tree'', or, the
    step can be simulated by the ``bigger'' tree while still preserving the
    embedding relation.
}
The proof of the lemma is easy (by considering all four possible rewrite
operations), and is relegated to the \shortlong{full version}{appendix}.

One consequence of this monotonicity property is that, when dealing with the
redundancy problem, we can safely ignore rewrite rules that use one of the rewrite
operations $\removenode$ or $\removeclass{X}$. This is formalised in the following
lemma, whose proof is given in the \shortlong{full version}{appendix}.
\begin{lemma}\label{lem:removenode}
Given a rewrite system $\trs$ over $\treelabels$-trees, a guard database $S$,
and an initial tree $T_0$, let $\trs^{-}$ be the set of $\trs$-rules less
those that use either $\removenode$ or $\removeclass{X}$. Then, for each $g \in S$, $g$
is matched in $\post_{\trs}^*(\tree_0)$
iff $g$ is matched in $\post_{\trs^-}^*(\tree_0)$.
\end{lemma}

\begin{convention}
In the sequel, we assume that there are only two possible
rewrite operations, namely, $\addchild{X}$, and $\addclass{X}$.
\end{convention}

\begin{remark}
    \label{rm:ordered}
    When choosing unordered trees for our CSS redundancy analysis
    in Section \ref{sec:trs} (i.e. instead of ordered trees),
    we remarked that we have added a layer of \emph{sound} approximation
    to our analysis. \alchanged{We will now explain why this is a sound
    approximation.}
    We could consider the extension of our
    guard language with the left-sibling and right-sibling operators $\langle
    \leftarrow \rangle$ and $\langle\rightarrow\rangle$ (which would still
    be contained in Tree Temporal Logic, which as we already mentioned is a
    formal model of XPath \cite{Sch92,Marx05,libkin-survey}). The semantics
    of formulas of the form $\langle \leftarrow \rangle g$ and
    $\langle \rightarrow \rangle g$ (with respect ordered trees) can be defined
    in the same way as our guard language. Given an ordered tree $T$, let
    $T'$ be the unordered version of $T$ obtained by ignoring the sibling
    ordering from $T$. Given a formula $g$ with left/right sibling operators,
    we could define its ``unordered approximation'' $g'$ by replacing every
    occurrence of $\langle \leftarrow \rangle$ and $\langle\rightarrow\rangle$
    by $\langle \uparrow\rangle\langle \downarrow\rangle$. For positive guards
    $g$, it is easy to show by induction on $g$ that $v,T \models g$ implies
    $v,T' \models g'$. By the same token, we could also consider an extension
    of our tree rewriting to ordered trees that allows adding an immediate
    left/right sibling, e.g., by the operators $\addleftkin{X}$ and
    $\addrightkin{X}$ (these are akin to the \texttt{.before()} and
    \texttt{.after()} jQuery methods). Given a tree rewriting $\trs$ in this
    extended rewrite system, we could construct an approximated rewriting
    $\trs'$ in $\trsclass$ as follows: (1) for every rewrite rule of the form
    $(g,\addleftkin{X})$ or
    $(g,\addrightkin{X})$ in this extended tree rewriting, add
    the rewrite rule $(\langle \downarrow \rangle g',
    \addchild{X})$ in $\trs'$, and (2) for every other rewrite rule $(g,\chi)$
    in $\trs$, add the rewrite rule $(g',\chi)$ in $\trs$. A consequence of
    this approximation is that a guard $g$ can be matched in a reachable tree
    $post_{\trs}^*(T_0)$ implies that the unordered approximation $g'$ can
    be matched in a reachable tree $post_{\trs'}^*(T_0')$. That is, if $g'$
    is redundant in $\trs'$, then $g$ is redundant in $\trs$, i.e., that $g$
    can be safely removed.
\end{remark}

\subsection{Summary of technical results}
We have completely identified the computational complexity of the redundancy
and $k$-redundancy problem for $\trsclass$. Our first result is:
\begin{theorem}
The redundancy problem for $\trsclass$ is $\EXP$-complete.
\label{th:exp}
\end{theorem}

Our upper bound was obtained via an efficient reduction to an analysis of
symbolic pushdown systems (see Section \ref{sec:upper}), for which there are
highly optimised tools (e.g. Bebop
\cite{bebop}, Getafix \cite{getafix}, and Moped \cite{moped}). We have
implemented
our reduction and demonstrate its viability in detecting redundant CSS rules in
HTML5 applications (see Section \ref{sec:experiments}). The proof of the lower
bound in Theorem \ref{th:exp} is provided in the \shortlong{full version}{appendix}.
In the case of $k$-redundancy problem, a better complexity can be obtained.
\begin{theorem}
The $k$-redundancy problem $\trsclass$ is:
\begin{enumerate}[(i)]
\item $\PSPACE$-complete if $k$ is part of the input in unary.
\item solvable in $\P$-time --- $n^{O(k)}$ --- for each fixed parameter $k$, but is
    $\W[1]$-hard.
\end{enumerate}
\label{th:pspace}
\end{theorem}
Recall that $\PSPACE \subseteq \EXP$. The second item of Theorem \ref{th:pspace}
contains the complexity class $\W[1]$ from \emph{parameterised complexity
theory}
(e.g. see \cite{grohe-book}), which provides a theory for answering whether
a computational problem with multiple input parameters is \emph{efficiently
solvable} when certain parameters are \emph{fixed}.
\mhchanged{
    In the case of the
    $k$-redundancy problem for $\trsclass$
    a problem instance contains $k \in \N$ and $\trs \in \trsclass$ as the input
    parameters.
    The problem can be solved in time $n^{O(k)}$, where $n$ is the size of $\trs$.
    We would like to know whether the parameter $k$ can be removed
    from the exponent of $n$
    in the time-complexity.  That is,%
}
whether the problem is solvable in time
$f(k) n^c$ for some computable function $f: \N \to \N$ and a constant $c \in
\N$ (a.k.a. \emph{fixed-parameter-tractable (FPT) algorithms}).
    Observe that, asymptotically, $f(k)n^c$ is smaller than $n^{O(k)}$ \
    \alchanged{for
    \emph{every} fixed value of $k$}.
By showing that
the problem is $\W[1]$-hard, we have in effect shown that the parameter $k$
cannot be removed from the exponent of $n$, i.e., that our $n^{O(k)}$-time
algorithm is, in a sense, optimal. For space reasons, we relegate the proofs of
Theorem \ref{th:pspace} to the \shortlong{full version}{appendix}.
%

%
\OMIT{
(i.e. problem instances now are partitioned by parameter $k \in \N$, which in
our case denotes the height of the input tree) that contains problems that
are solvable in time $n^{k^{O(1)}}$, but are in general \emph{unlikely} to be
solvable in by \emph{fixed-parameter-tractable algorithms}, i.e., in time
$f(k) n^c$ for some computable function $f: \N \to \N$ and a constant $c \in
\N$.
In other words, $\W[1]$-hard problems are solvable in polynomial-time
for each \emph{fixed} parameter $k$, though the degree of the polynomial has to
\emph{depend} on the parameter $k$. Problems that are $\W[1]$-hard include

In the context of Theorem \ref{th:pspace},
the $\W[1]$-hardness of $k$-redundancy problem suggests that our
$n^{O(k)}$-time algorithm is essentially optimal for each
fixed height $k$.
}

\begin{remark}
    Decidability for the $k$-redundancy problem for $\trsclass$ is
    immediate from the theory of well-structured transition systems
    (e.g. see \cite{AbdullaCJT96,FinkelS01}).
    We have shown in Lemma \ref{lem:monotonic} that the tree embedding relation
    $\trleq$ is monotonic. It can be shown that $\trleq$ is also
    \emph{well-founded} on trees of height $k$ (e.g. see \cite{GMSZ08}),
    i.e., there is no infinite descending chain $T_1 \succ T_2 \succ \cdots$
    for trees $T_1,\ldots$ of height $k$. The theory of well-structured
    transition systems (e.g. see \cite{AbdullaCJT96,FinkelS01}) would imply
    decidability for the $k$-redundancy problem. Unfortunately, this only gives
    a nonelementary upper bound complexity (i.e. an unbounded tower of
    exponential) for the $k$-redundancy and does yield decidability for the
    general redundancy problem.
\end{remark}

\section{The algorithm}
\label{sec:upper}

In this section, we provide an efficient reduction to an analysis of symbolic
pushdown systems, which will give an algorithm with an exponential-time (resp.
polynomial-space) worst-case upper bound for the redundancy (resp.
$k$-redundancy) problem for $\trsclass$.
To this end, we first provide a preliminary background on symbolic pushdown
systems. We will then provide a roadmap of our reduction to symbolic pushdown
systems, which will consist of a sequence of three polynomial-time reductions
described in the last three subsections.

\subsection{Pushdown systems: a preliminary}
Before describing our reduction, we will first provide a preliminary background on
pushdown systems and their extensions to symbolic pushdown systems.

Pushdown systems are standard (nondeterministic) pushdown automata without
input labels. Input labels are irrelevant since
one mostly asks about their transition graphs (in our case, reachability).
More formally,
a \defn{pushdown system (PDS)} is a tuple
\[
    \pds = (\controls,\salphabet,\pdsrules)
\]
where
\begin{itemize}
\item
    $\controls$ is a finite set of \defn{control states},
\item
    $\salphabet$ is a finite set of \defn{stack symbols}, and
\item
    $\pdsrules$ is a finite subset of $(\controls \times \salphabet) \times
    (\controls \times \salphabet^*)$ such that if $((q,a),(q',w)) \in \pdsrules$
    then $|w| \leq 2$.
\end{itemize}
This PDS generates a transition relation $\to_\pds
\subseteq (\controls \times \salphabet^*) \times (\controls \times \salphabet^*)$
as follows: $(q,v) \to_\pds (q',v')$ if there exists a rule
$((q,a),(q',w)) \in \pdsrules$ such that $v = ua$ and $v' = uw$ for some
word $u \in \salphabet^*$.

Symbolic pushdown systems are pushdown systems that are succinctly represented
by boolean formulas. They are equivalent to \emph{(recursive) boolean programs}.
More precisely,
a \defn{symbolic pushdown system (sPDS)} is a tuple
\[
    (\varsetV,\varsetW,\pdsrules)
\]
where
\begin{itemize}
\item
    $\varsetV = \{x_1,\ldots,x_n\}$ and $\varsetW = \{y_1,\ldots,y_m\}$ are
    two disjoint sets of boolean variables, and
\item
    $\pdsrules$ is a finite set of pairs $(i,\varphi)$ of
    number $i \in \{0,1,2\}$ and boolean formula $\varphi$ over
    the set of variables $\varsetV \cup \varsetW \cup \varsetV' \cup \varsetW'$,
    where
    \begin{itemize}
    \item
        $\varsetV' := \{x_1',\ldots,x_n'\}$, and
    \item
        $\varsetW' := \bigcup_{j=1}^i \varsetW_j$ with $\varsetW_j :=
        \{y_1^j,\ldots,y_m^j\}$.
    \end{itemize}
\end{itemize}
This sPDS generates a (exponentially bigger) pushdown system
$\pds = (\controls,\salphabet,\pdsrules')$, where
$\controls = \{0,1\}^n$, $\salphabet = \{0,1\}^m$, and
$((q,a),(q',w)) \in \pdsrules'$ iff there exists a pair $(i,\varphi) \in
\pdsrules$ satisfying $i = |w|$, and $\varphi$ is satisfied by the assignment
that assigns\footnote{Meaning that if $q = (q_1,\ldots,q_n)$, then $q_i$
is assigned to $x_i$} $q$ to $\varsetV$, $a$ to $\varsetW$, $q'$ to $\varsetV'$,
and $w$ to $\varsetW'$ (i.e. assigning the $j$th letter of $w$ to $\varsetW_j$).

The \defn{bit-toggling problem for sPDS} is a simple reachability problem over
symbolic pushdown systems.
Intuitively, we want to decide if we
can toggle on the variable $y_i$ from the initial configuration.
More precisely, given an sPDS $\pds = (\varsetV,\varsetW,\pdsrules)$ with
$\varsetV = \{x_1,\ldots,x_n\}$ and $\varsetW = \{y_1,\ldots,y_m\}$,
a variable $y_i \in \varsetW$, and an
initial configuration $I_0 = ((b_1,\ldots,b_n),(b_1',\ldots,b_m'))
\in \{0,1\}^n \times \{0,1\}^m$, decide if $I_0 \to_\pds^*
(q,a)$ for some $q \in \{0,1\}^n$ and $a = (b_1'',\ldots,b_m'')
\in \{0,1\}^m$ with $b_i'' = 1$.

The \defn{bounded bit-toggling problem} is the same as the
bit-toggling problem but the stack height of the pushdown system cannot
exceed some given input parameter $h \in \N$ (given in unary).
\begin{proposition}
The bit-toggling (resp. bounded bit-toggling) problem for sPDS is solvable in
$\EXP$ (resp. $\PSPACE$).
\label{prop:spds-complexity}
\end{proposition}
The proof of this is standard (e.g. see \cite{schwoon-thesis}), which for
completeness we provide in the \shortlong{full version}{appendix}.

Despite the relatively high complexity mentioned in Proposition
\ref{prop:spds-complexity}, nowadays there are highly optimised sPDS and
boolean program solvers
(e.g. Moped \cite{schwoon-thesis,moped}, Getafix \cite{getafix}, and Bebop
\cite{bebop})  that can solve sPDS bit-toggling
problem efficiently using BDD (Binary Decision Diagram) representation of
boolean formulas. In fact, the boolean formulas that we produce in our
polynomial-time reductions below have straightforward small representations as
BDDs.

\OMIT{
Most
As
we will show in the \shortlong{full version}{appendix}, using a BDD (Binary
Decision Diagram) representation of boolean formulas is adequate for our purposes.
BDD representations are crucial since they will allow us to tap into
efficient sPDS solvers which exploit BDD representations
}

\subsection{Intuition/Roadmap of the reduction}
The following theorem formalises our reduction claim.
\begin{theorem}
    The redundancy (resp. $k$-redundancy) problem for $\trsclass$ is
    polynomial-time reducible to the bit-toggling (resp. bounded bit-toggling)
    problem for sPDS.
    \label{th:reduction}
\end{theorem}
Together with Proposition \ref{prop:spds-complexity}, Theorem \ref{th:reduction}
implies an $\EXP$ (resp. $\PSPACE$) upper bound for the redundancy (resp.
$k$-redundancy) problem for sPDS.  Moreover, as discussed in the
\shortlong{full version}{appendix}, it is straightforward to construct from our
reduction a counterexample path in the rewrite system witnessing the non-redundancy
of a given guard.

The actual reduction in Theorem \ref{th:reduction} involves several intermediate
polynomial-time reductions. Here is a roadmap.
\begin{itemize}
\item
    We first show that it suffices to consider ``simple'' rewrite systems.
    These systems are simple in the sense that guards only test direct parents or children of the nodes.
    Furthermore, these systems have the property that we only need to check redundancy at the root node.
    This is given in Section~\ref{sec:simplifying-rewrite}.

\item
    \mhchanged{
        We then show that the simplified problem can be solved by a ``saturation'' algorithm that uses a subroutine to check whether a given class can be added to a given node via a sequence of rewrite rules.
        This subroutine solves what we call the ``class-adding problem'':
        a simple reachability problem that checks whether a class can be added to a given node via a series of rewrite rules that do not change any other existing nodes in the tree (but may add new nodes).
        That is, the node is considered as the only node in a single node tree,
        possibly with some contextual (immutable) information about its parent.
        This is given in Section~\ref{sec:reduction-class-adding}.
    }

\item
    Finally, we show that the class-adding problem is efficiently reducible
    to the bit-toggling problem for sPDS.
    This is given in Section~\ref{sec:reduction-bit-toggling}.
\end{itemize}
\mhchanged{
    The case of $k$-redundancy for $\trsclass$ is similar but each intermediate
    problem is relativised to the version with bounded height.
}

The reason we need to simplify the system is because, in the simplified system, we can test redundancy only by inspection of the root node, and all guards only refer to direct neighbours of each node.
The latter simplification makes the reduction to sPDS possible.
\mh{removed sentence about only having access to the labels of $v$, since it's not true...}
\mhchanged{
    As explained in Section~\ref{sec:reduction-bit-toggling}, the constructed sPDS performs a kind of depth-first search over the constructed trees.
}%
Since an sPDS can only see the top of its stack, it is important that it only needs to maintain local information about the node being inspected by the depth-first search.
The former simplification justifies us only maintaining labelling information about the nodes in the original tree (in particular, the root node) without having to (explicitly) add new nodes.

\begin{figure}
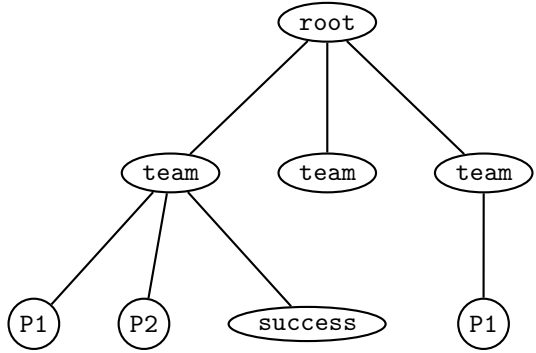

\begin{center}
\pstree{\Toval{$\Root$}}{
    \pstree{\Toval{$\Team$}}{
        \Tcircle{$\PlayerOne$}
        \Tcircle{$\PlayerTwo$}
        \Toval{$\Success$}
    }
    \Toval{$\Team$}
    \pstree{\Toval{$\Team$}}{
        \Tcircle{$\PlayerOne$}
    }
}
\end{center}
\caption{A reachable configuration\label{fig:ex}}
\end{figure}

\paragraph{Running Example.}
We provide a running example for our sequence of reductions. Imagine a
tennis double tournament web page which can be used to keep track of a list of
teams (containing exactly two players).
The page allows a user to create a team, and add players to a team. The
page will also indicate a success next to the team details after both players
have been added. An overapproximation of the behavior of the page could be
abstracted as a rewrite system $\trs$ as follows:
\begin{itemize}
    \item The initial tree is the single-node tree with label $\Root$.
    \item The set of node labels is
        $\{\Root,\Team,\PlayerOne,\PlayerTwo,\Success\}$.
    \item A team can be added to the tournament, i.e., there is
        a rule $(\Root,\addchild{\Team}) \in \trs$.
    \item Player 1 can be added to a team, i.e., there is a rule
        $(\Team,\addchild{\PlayerOne}) \in \trs$
    \item Player 2 can be added to a team, i.e., there is a rule
        $(\Team,\addchild{\PlayerTwo}) \in \trs$
    \item
        Success after
        \acmeasychair{}{both of the required players}
        Player 1 \& Player 2 are added, i.e., there is a rule
        $(\langle \downarrow \rangle \PlayerOne
        \wedge \langle \downarrow \rangle \PlayerTwo,\addchild{\Success})
          \in \trs$
\end{itemize}
The set $S$ of guards that we want to check for redundancy is $\{\Success\}$.
A snapshot of a reachable configuration is provided in Figure \ref{fig:ex}.

\subsection{Simplifying the rewrite system}
\label{sec:simplifying-rewrite}

We will make the following two simplifications: (1) restricting the
problem to only checking redundancy at the \emph{root} node,
(2) restrict the guards to be used.

To achieve simplification (1), one can simply define a new set of guards
from $S$ as $\{ \langle \downarrow^\ast \rangle g: g \in S \}$.
Then, for each tree $\tree = (\treedomain,\treelabeling) \in
\trees(\treelabels)$ and guard $g$, it is the case that $(\exists v\in\treedomain: v,\tree
\models g)$ iff $\epsilon,\tree \models \langle\downarrow^\ast\rangle g$.
\OMIT{
\begin{convention}
In this section, the redundancy and $k$-redundancy problems will be restricted
to decide whether a set $S$ of guards is never matched at the \emph{root}
of reachable trees.
\end{convention}
}

We now proceed to simplification (2).
A guard over the node labeling $\treelabels = 2^\CLASS$
is said to be \defn{simple} if it is of one of the following two forms
\[
    \bigwedge_{i=1}^m c_i
    \text{\qquad or \qquad}
    \langle d \rangle \bigwedge_{i=1}^m c_i
\]
for some $m \in \N$,
where each $c_i$ ranges over $\CLASS$ and $d$ ranges over
$\{\uparrow,\downarrow\}$. [Note: if $m = 0$, then $\bigwedge_{i=1}^m c_i
\equiv \top$.] For notational convenience, if $X =
\{c_1,\ldots,c_m\}$, we shall write $X$ (resp. $\langle d \rangle X$) to mean
$\bigwedge_{i=1}^m c_m$ (resp. $\langle d\rangle \bigwedge_{i=1}^m c_m$).
A rewrite system $\trs \in \trsclass$ is said to be
\defn{simple} if (i) all guards occurring in $\trs$ are simple, and (ii) if
$(\langle d \rangle X,\chi) \in \trs$, then $\chi$ is of the form
$\addclass{Y}$.

We define $\Rtrsclass \subseteq \trsclass$ to be the set of simple rewrite
systems.
The redundancy
(resp. $k$-redundancy) problem for $\Rtrsclass$ is defined in the same way
as for $\trsclass$ \emph{except that} all the guards in the set of guards
are restricted to be a subset of $\CLASS$.

The following lemma shows that the redundancy (resp. $k$-redundancy) problem
for $\trsclass$ can be reduced in polynomial time to the redundancy (resp.
$k$-redundancy) problem for $\Rtrsclass$.
\mhchanged{
    Essentially, the reduction works by introducing new classes representing (non-simple) subformulas of the guards.
    New rewrite rules are introduced that inductively calculate which subformulas are true.
    That is, if a subformula $g$ is true at $v$, then the labelling of $v$ will include a class representing $g$.

    Note, in the lemma below, $S'$ is a set of atomic guards.
    That is, each guard in $S'$ is of the form $c$ for some $c \in \CLASS'$.
}
\begin{lemma}
    Given a $\trs \in \trsclass$ over $\treelabels = 2^\CLASS$ and a set
    $S$ of guards over $\treelabels$, there exists $\trs' \in \Rtrsclass$ over
    $\treelabels' = 2^{\CLASS'}$ (where $\CLASS \subseteq \CLASS'$) and a set
    \mhchanged{%
        $S'$ of atomic%
    }
    guards such that:
\begin{description}
\item[(P1)] For each $k \in \N$, $S$ is $k$-redundant for $\trs$ iff $S'$ is
$k$-redundant for $\trs'$.
\item[(P2)] $S$ is redundant for $\trs$ iff $S'$ is redundant for $\trs'$.
\end{description}
Moreover, we can compute $\trs'$ and $S'$ in polynomial time.
\label{lm:simple}
\end{lemma}
We show how to compute $\trs'$. The set $\CLASS'$ is defined as the union of
$\CLASS$ with the set $G$ of all
\mhchanged{%
    non-atomic%
}
subformulas (i.e. occurring in the parse tree) of guard formulas in $S$ and $\trs$.
In the sequel, to avoid potential confusion, we will often underline members of $G$
in $\CLASS'$, e.g., write $\myclass{\langle \downarrow \rangle c}$ instead of
$\langle \downarrow \rangle c$.
\mhchanged{
    Note that $\myclass{c} = c$ for all $c \in \CLASS$.
}

We now define the simple rewrite system $\trs'$. Initially, we will
define a rewrite system $\trs_1$ that allows the operators $\langle \uparrow^\ast
\rangle$ and $\langle \downarrow^\ast \rangle$;
later we will show how to remove them. We first add the following
``intermediate'' rules to $\trs_1$:
\begin{enumerate}
\item  $(\{\myclass{g},\myclass{g'}\},\addclass{\myclass{g \wedge g'}})$, for each $(g \wedge
        g') \in G$.
\item  $(\myclass{g},\addclass{\myclass{g \vee g'}})$ and $(\myclass{g'},\addclass{\myclass{g
    \vee g'}})$, for each $(g \vee g') \in G$.
\item
    \label{item:modality-addition}
    $(\langle d \rangle \myclass{g},\addclass{\myclass{\langle d \rangle g}})$, for each
$\langle d \rangle g \in G$.
\item
    \label{item:general-rule-addition}
    $(\myclass{g},\chi)$, for each $(g,\chi) \in \trs$.
\end{enumerate}
Note that in Rule~(\ref{item:modality-addition}) the guard $\langle d \rangle \myclass{g}$ is understood
to mean a non-atomic guard over $2^{\CLASS'}$
\mhchanged{%
    -- that is, $\myclass{g}$ is an atomic guard.
}
Finally, we define $S' := \{ \myclass{g} : g \in S \}$. Notice that
each guard in $S'$ is atomic. The aforementioned algorithm computes $\trs_1$ and
$S'$ in linear time.


We now show how to remove the operators $\langle \uparrow^\ast \rangle$ and
$\langle \downarrow^\ast \rangle$ (i.e. rules of type (3)).
\mhchanged{
    To do this we will introduce new rules $\trs_2$ that essentially compute the required transitive closures using the $\uparrow$ and $\downarrow$ operators.
    Our final rewrite system $\trs'$ will be $\trs_1 \cup \trs_2$.
    We define $\trs_2$ to be the set containing the following rules for each rule
    $(\langle d^\ast \rangle \myclass{g}, \chi) \in \trs_1$
    where
    $d \in \{\uparrow,\downarrow\}$:
}
\begin{enumerate}[(a)]
\item $(\myclass{g},
\addclass{\underline{\langle d^\ast\rangle g}})$.
\item $(\langle d \rangle\underline{\langle d^\ast\rangle g},
\addclass{\underline{\langle d^\ast\rangle g}})$.
\end{enumerate}
Note that
\mhchanged{%
    from Rule~(\ref{item:general-rule-addition})
}
$\trs_1$ contains the rule $(\myclass{\langle d^\ast \rangle
g},\chi)$, where $\myclass{\langle d^\ast \rangle g}$ is understood to mean an atomic
guard over $2^{\CLASS'}$.
Intuitively, this simplification can be done because $v,\tree \models \langle d^\ast
\rangle g$ iff
at least one of the following
cases holds: (i) $v,\tree \models g$, (ii)
there exists a node $w$ in $\tree$ such that $w,\tree \models
\langle d^\ast \rangle g$ and $w$ can be reached from $v$ by following the
direction $d$ for one step. The aforementioned computation step again can be done
in linear time.
The proof of correctness (i.e. \textbf{(P1)} and \textbf{(P2)}) is provided in the
\shortlong{full version}{appendix}. In particular, for all $g \in S$, it is the case that $g$ is redundant in
$\trs$ iff $\underline{g}$ is redundant in $\trs'$.

\paragraph{Running Example.}
In our example, initially, $S$ is changed into
$\{\langle \downarrow^\ast \rangle \Success \}$ after the first simplification.
To perform the second, we first obtain the rewrite system $\trs_1$ containing the following rules
(recall we equate $\myclass{c}$ and $c$ for each class):
\begin{itemize}
    \item $(\Success,\addclass{\myclass{\langle \downarrow^\ast \rangle
        \Success}})$
    \item $(\{\langle \downarrow \rangle \PlayerOne,
            \langle \downarrow \rangle \PlayerTwo\},
            \addclass{\myclass{
            \langle \downarrow \rangle \PlayerOne
            \wedge \langle \downarrow \rangle \PlayerTwo
            }})$
        \item $(\langle \downarrow \rangle \myclass{P},
            \addclass{\myclass{\langle \downarrow \rangle P}})$ for each
        $P \in \{\PlayerOne,\PlayerTwo\}$
\end{itemize}
and from Rule~(\ref{item:general-rule-addition}) we also have in $\trs_1$.
\begin{itemize}
\item
    $(\Root,\addchild{\Team})$,
\item
    $(\Team,\addchild{\PlayerOne})$,
\item
    $(\Team,\addchild{\PlayerTwo})$,
\item
    $(\myclass{\langle \downarrow \rangle \PlayerOne
               \wedge
               \langle \downarrow \rangle \PlayerTwo},
      \addchild{\Success})$.
\end{itemize}
Now, $\trs_2$ contains the following rules:
\begin{itemize}
    \item $(\Success,\addclass{ \myclass{\langle \downarrow^\ast \rangle \Success}})$
    \item $(\langle \downarrow \rangle \myclass{\langle \downarrow^\ast \rangle
        \Success},\addclass{\myclass{\langle
        \downarrow^\ast \rangle \Success}})$
\end{itemize}
Finally, we define $S'$ to be
$\set{\myclass{\langle \downarrow^\ast \rangle \Success}}$.

\subsection{Redundancy $\rightarrow$ class-adding}
\label{sec:reduction-class-adding}

We will show that redundancy
for $\Rtrsclass$ can be solved in polynomial
time assuming an oracle to the ``class
adding problem'' for $\Rtrsclass$.  The class-adding problem is a
reachability problem for $\Rtrsclass$
involving only single-node input trees possibly with a parent node that only
provides a ``context'' (i.e. cannot be modified). As we will see in the following
subsection, the class-adding problem for $\Rtrsclass$
lends itself to a fast
reduction to the bit-toggling problem for
sPDS.  Similarly, $k$-redundancy can be solved via the same routine, where intermediate problems are restricted to their bounded height equivalents.
\OMIT{
Furthermore,
since the complexity of these reachability problems
are $\EXP$ and $\PSPACE$ respectively (by reduction to problems for symbolic
pushdown systems),
}

\paragraph{High-level idea.}
We first provide the high-level idea the reduction.
After simplifying the rewrite system in the previous step, we only need to
check redundancy at the root node of the tree.
Our approach is a ``saturation'' algorithm that exploits the monotonicity property: we begin with an initial tree and then repeatedly apply a ``saturation step'' to build the tree where each node is labelled by all classes that may label the node during any execution.
The saturation step examines the tree built so far and the rules of the rewrite system.
If it finds that it is possible to apply a sequence of rewrite rules to the tree to add a class $c$ to some node $v$, it updates the tree by adding $c$ to the label of $v$.
In this way, larger and larger trees are built.
Once it is no longer possible to add any new classes to the existing nodes of the tree, the algorithm terminates.
In particular, we have all classes that could label the root node, and thus we can detect which classes are redundant by inspecting the labelling of the root.

Given a
rewrite system $\trs \in \trsclass$, an initial tree $\tree = (\treedomain,
\treelabeling)
\in
\trees(\treelabels)$ with $\treelabels = 2^\CLASS$, and a set $S$ of guards,
we try to ``saturate'' each node $v \in \treedomain$ with the classes
that may be added to $v$.
Our saturation step is able to reason about the addition of nodes when determining if a class $c$ can be added to $v$, but does not need to remember which new nodes needed to be added to $\tree$ to add $c$ to $v$.
This is due to monotonicity: since each saturation step begins with a larger tree than the previous step, additional nodes can be regenerated on-the-fly if needed.

Each saturation step proceeds as follows.
Let $\tree_1 = (\treedomain,\treelabeling_1)$ be the tree before the saturation step, and $\tree_2= (\treedomain,\treelabeling_2)$ be the tree after applying then saturation step.
The tree domain does not change and there exists a node $v \in \treedomain$ such that $\treelabeling_1(v) \subset \treelabeling_2(v)$ (i.e. some classes are added to $\treelabeling_1(v)$).
In particular, we have $\tree_1 \prec \tree_2$.

There are two saturation rules that are repeatedly applied until we reach a fixpoint.
\begin{itemize}
\item
    The first saturation rule corresponds to the application of a rewrite rule $(g,\addclass{X}) \in
    \trs$ at $v$.
    We simply set $\treelabeling_2(v) = \treelabeling_1(v) \cup X$.
    Note, in this case we do not need to reason about the addition of nodes to the tree.

\item
    \mhchanged{
        The second saturation rule is a call to the class-adding subroutine and asks whether some class $c$ can be added to node $v$.
        This step incorporates the behaviour of rewrite rules of the form $(g, \addchild{X})$.
    }
    In this case we need to reason about whether the node added by these rules could lead to the addition of $c$ to $v$.
    To do this, we construct a pushdown system $\pds$ that explores the possible impact of these new nodes.
    This is discussed in the next section.
    If it is found that $c$ could be added to $v$, we set
    $\treelabeling_2(v) = \treelabeling_1(v) \cup \set{c}$.
\end{itemize}
In sum, our ``reduction'' is in fact an algorithm for the ($k$-) redundancy
problem for $\trsclass$ that runs in polynomial-time with an oracle to the bit-toggling (resp. for bounded stack height) for sPDS.

\paragraph{The formal reduction.}
Before formally defining the class-adding problem, we first need the definition
of an ``assumption function'', which plays the role of the possible
parent context node but \alchanged{is} treated as a \emph{separate} entity
from the input tree.
As we shall see, this leads to a more natural formulation of
the computational problem. More precisely, an \emph{assumption function} $f$ over
the alphabet
$\treelabels = 2^\CLASS$ is a function mapping each element of $\{\Root\} \cup
\CLASS$ to $\{0,1\}$. The boolean value of $f(\Root)$ is used to indicate
whether the input single-node tree is a root node\footnote{Note that the guard
$\langle \uparrow \rangle \top$ evaluates to false on the root node}.
\OMIT{Intuitively, suppose that a tree $\tree \in
\trees(\treelabels)$ is a subtree of a bigger tree $\tree' =
(\treedomain',\treelabeling') \in \trees(\treelabels)$ rooted at a node
$v.i \in \N^*$ of $\tree'$. Then, an assumption function
can be used to record the node label $\treelabeling'(v)$ of the parent of
$v.i$, which can be used to check if the guard of the form $\langle
\uparrow\rangle X$ is satisfied at $v$.}
Given a tree $\tree = (\treedomain,\treelabeling) \in \trees(\treelabels)$, a
node $v \in \treedomain$, and a simple guard $g$ over $\treelabels$, we write
$v,\tree \models_f g$ if one of the following three cases holds: (i) $v \neq
\epsilon$ and $v,\tree \models g$, (ii) $v = \epsilon$, $g$ is not of
the form $\langle \uparrow \rangle X$, and $v,\tree \models g$, and
(iii) $v = \epsilon$, $g = \langle \uparrow \rangle X$, $f(\Root) = 0$,
and $X \subseteq \{ c \in \CLASS : f(c) = 1 \}$.  In other words,
$v,\tree \models_f g$ checks whether $g$ is satisfied at node $v$ assuming
the assumption function $f$ (in particular, if $f(\Root) = 0$, then
any guard referring to the parent of the root node of $\tree$ is checked against
$f$).

Given a simple rewrite system $\trs \in \Rtrsclass$ over $\treelabels$ and an
assumption function $f$,
we may define the rewriting relation $\to_{\trs,f}\ \subseteq
\trees(\treelabels) \times \trees(\treelabels)$ in the same way as we define
$\to_\trs$, except that $\models_f$ is used to check guard satisfaction.
The \defn{class-adding (reachability) problem} is defined as follows:
given a single-node tree $\tree_0 = (\{\epsilon\}, \treelabeling_0) \in
\trees(\treelabels)$ with $\treelabels = 2^\CLASS$, an assumption function
$f: (\{\Root\} \cup \CLASS) \to \{0,1\}$, a class $c \in \CLASS$, and a simple
rewrite system $\trs$, decide if there exists a tree $\tree =
(\treedomain,\treelabeling)$ such that $\tree_0 \to_{\trs,f}^* \tree$ and
$c \in \treelabeling(\epsilon)$.
Similarly, the \defn{$k$-class-adding} problem is defined in the same way as
the class-adding problem except that the
reachable trees are restricted to height $k$ ($k$ is part of the input).
\begin{lemma}
\label{lm:reduce-class-adding} The redundancy (resp. $k$-redundancy) problem for
simple rewrite systems is $\P$-time solvable assuming oracle calls to the
class-adding (resp. $k$-class-adding) problem.
\end{lemma}
\OMIT{
In fact, we shall see later that the class-adding (resp. $k$-class-adding)
problem is solvable in $\EXP$ (resp. $\PSPACE$), which together with the
above lemma immediately imply the desired upper bound.
}
Given a tree is $\tree_0 = (\treedomain_0,\treelabeling_0) \in
\trees(\treelabels)$ with $\treelabels = 2^\CLASS$, a simple rewrite system
$\trs$ over $\treelabels$, and a set $S \subseteq \CLASS$, the task is
to decide whether $S$ is redundant (or $k$-redundant). We shall give the
algorithm for the redundancy problem; the $k$-redundancy problem can be
obtained by simply replacing oracle calls to the class-adding problem
by the $k$-class-adding problem.

The algorithm is a fixpoint computation. Let $\tree =
(\treedomain,\treelabeling) \\
:= \tree_0$. At each step, we can apply
any of the following ``saturation rules'':
\begin{itemize}
\item if $(g,\addclass{B})$ is applicable at a node $v \in \treedomain_0$
    in $\tree$, then $\treelabeling(v) := \treelabeling(v) \cup B$.
\item
    If the class-adding
    problem has a positive answer on input $\langle \tree_v,f,c,\trs\rangle$,
    then $\treelabeling(v) := \treelabeling(v) \cup \{c\}$, where
    $v \in \treedomain$, $c \in \CLASS \setminus \treelabeling(v)$,
    and $\tree_v := (\epsilon,\treelabeling_v)$ with
    $\treelabeling_v(\epsilon) = \treelabeling(v)$, where we define
    $f: (\{\Root\}\cup\CLASS) \to \{0,1\}$ with $f(\Root)=1\LRA v = \epsilon$
    and, if $f(\Root) = 0$ and $u$ is the parent of $v$, then
    $f(a)=1 \LRA a\in \treelabeling(u)$.
\end{itemize}
Observe that saturation rules can be applied at most
$\CLASS \times |\treedomain|$ times. Therefore, when they can be applied no
further, we check whether $S \cap \treelabeling(\epsilon) \neq
\emptyset$ and terminate. Assuming constant-time oracle calls to the
class-adding problem, the algorithm easily runs in polynomial time.
Furthermore, since each saturation rule only adds new classes to a node
label, the correctness of the algorithm can be easily proven using Lemma
\ref{lem:preserve-vDash} and Lemma \ref{lem:monotonic}; see the \shortlong{full
version}{appendix}.

\paragraph{Running Example}
The application of the saturation algorithm to our running example is quite simple.
Since the only node in the initial tree is the root node, we can solve the redundancy problem with a single call to the class-adding problem.
That is, is it possible to add the class
$\myclass{\langle \downarrow^\ast \rangle \Success}$
to the root node?

\subsection{Class-adding $\rightarrow$ bit-toggling}
\label{sec:reduction-bit-toggling}

We now show how to solve the class-adding problem for a given node $v$ and
class $c$.
Let $\treelabeling(v)$ be the labelling of node $v$.
We reduce the class-adding problem to the bit-toggling problem as follows.
The pushdown system performs a kind of ``depth-first search'' of the trees that
could be built from the new node and the rewrite rules.
It starts with an initial stack of height 1 containing the node
\mhchanged{
    being inspected.  More precisely,
}%
the single item on this stack is the set $\treelabeling(v)$ (possibly with some extra ``context'' information).
It then ``simulates'' each possible branch that is spawned from $v$ by pushing
items onto the stack when a new node is created (i.e. at each given moment, the
stack contains a single branch in a reachable configuration).
It pops these nodes from the stack when it wishes to backtrack and search other potential branches of the tree.

The pushdown system is an sPDS that keeps one boolean variable for each class $c \in \CLASS$.
If $\pds$ reaches a single item stack (i.e. corresponding to the node $v$)
where the item contains $c$ then the answer to the class adding problem is positive.
That is, the sPDS has explored the application of the rewrite rules to the possible children of $v$ and determined that the label of $v$ can be expanded to include $c$.

The reason why it suffices to only keep track of the labelling of $v$ (instead
of the entire subtree rooted at $v$ that $\pds$ explored) is monotonicity of
$\trsclass$ (cf. Lemma \ref{lem:monotonic}), i.e., that the labelling of $v$
contains sufficient information to ``regrow'' the destroyed subtree.

\OMIT{
To achieve $\EXP$ (resp. $\PSPACE$) algorithm for the class-adding
(resp. $k$-class-adding) problem, it suffices to give an efficient reduction to
the bit-toggling problem over sPDS.
}
\begin{lemma}
The class-adding (resp. $k$-class-adding) problem for $\Rtrsclass$ is
polynomial-time
reducible to the bit-toggling (resp. bounded bit-toggling) problem for sPDS.
\label{lm:reduce}
\end{lemma}

We prove the lemma above.
Fix a simple rewrite system $\trs$ over the node labeling $\treelabels =
2^\CLASS$, an assumption function $f: (\{\Root\} \cup \CLASS) \to \{0,1\}$,
a single node tree $\tree_0 = (\{\epsilon\},\treelabeling_0) \in
\trees(\treelabels)$, and a class $\alpha \in \CLASS$.

We construct an sPDS $\pds = (\varsetV,\varsetW,\pdsrules)$. Intuitively,
the sPDS $\pds$ will simulate $\trs$ by exploring all branches in all trees
reachable from
$\tree_0$ while accumulating the classes that are satisfied at the root node.
Define $\varsetV := \{ x_c : c \in \CLASS \} \cup \{\Pop\}$, and
$\varsetW := \{ y_c,z_c : c \in \CLASS \} \cup \{\Root\}$. Roughly speaking,
we will use the variable $y_c$ (resp. $z_c$) to remember whether the class
$c$ is satisfied at the current (resp. parent of the current) node being
explored. The variable $x_c$ is needed to remember whether the class $c$ is
satisfied at a child of the current node (i.e. after a pop operation). The
variable $\Root$ signifies whether the current node is a root node, while
the variable $\Pop$ indicates whether the last operation that changed the
stack height is a pop. We next define $\pdsrules$:
\begin{itemize}
\item
\mhchanged{
    For each $(A,\addclass{B}) \in \trs$, add the rule $(1,\varphi)$, where $\varphi$ asserts
    \begin{enumerate}[(a)]
    \item
        $\bigwedge_{a \in A} y_a$ ($A$ satisfied), and
    \item
        $\bigwedge_{b \in B} y_b^1$ ($B$ set to true), and
    \item
        all other variables remain unchanged, that is we assert
        $\Pop \leftrightarrow \Pop'$,
        $\Root \leftrightarrow \Root^1$,
        $\bigwedge_{c \in \CLASS \setminus B} (y_c \lra y_c^1)$,
        $\bigwedge_{c \in \CLASS} (z_c \lra z_c^1)$, and
        $\bigwedge_{c \in \CLASS} (x_c \lra x_c')$.
    \end{enumerate}
}
\item
For each $(\langle \uparrow \rangle A,\addclass{B}) \in \trs$, add
the rule $(1,\varphi)$ where $\varphi$ asserts
\mhchanged{
    \begin{enumerate}[(a)]
    \item
        $\bigwedge_{a \in A} z_a$ ($A$ satisfied in the parent), and
    \item
        $\bigwedge_{b \in B} y_b^1$ ($B$ set to true), and
    \item
        all other variables remain unchanged, that is
        $\Pop \lra \Pop'$,
        $\Root \lra \Root^1$,
        $\bigwedge_{c \in \CLASS} (x_c \lra x_c')$,
        $\bigwedge_{c \in \CLASS} (z_c \lra z_c^1)$, and
        $\bigwedge_{c \in \CLASS\setminus B} (y_c \lra y_c^1)$.
    \end{enumerate}
}
\item
For each $(\langle \downarrow \rangle A,\addclass{B}) \in \trs$, add
the rule $(1,\varphi)$ where $\varphi$ asserts
\mhchanged{
    \begin{enumerate}[(a)]
    \item
        $\Pop$ (we popped from a child), and
    \item
        $\bigwedge_{a \in A} x_a$ ($A$ satisfied in child), and
    \item
        $\bigwedge_{b \in B} y_b^1$ ($B$ set to true), and
    \item
        all other variables remain unchanged, that is
        $\Pop \lra \Pop'$,
        $\Root \lra \Root^1$,
        $\bigwedge_{c \in \CLASS} (x_c \lra x_c')$,
        $\bigwedge_{c \in \CLASS} (z_c \lra z_c^1)$, and
        $\bigwedge_{c \in \CLASS\setminus B} (y_c \lra y_c^1)$.
    \end{enumerate}
}
\item
For each $(A,\addchild{B}) \in \trs$, add the rule $(2,\varphi)$, where
$\varphi$ asserts
\mhchanged{
    \begin{enumerate}[(a)]
    \item
        $\bigwedge_{a \in A} y_a$ ($A$ satisfied), and
    \item
        $\bigwedge_{b \in B} y_b^2$ ($B$ true in new child), and
    \item
        $\bigwedge_{c \in \CLASS\setminus B} \neg y_c^2$ (new child has no other classes), and
    \item
        $\bigwedge_{c \in \CLASS} (y_c \lra z_c^2)$ (new child's parent classes), and
    \item
        $\neg \Root^2$ (new child is not root), and
    \item
        $\neg \Pop'
         \land
         \bigwedge_{c \in \CLASS} \neg x_c'$ (new child does not have a child), and
    \item
        the current node is unchanged, that is
        $\bigwedge_{c \in \CLASS} (y_c \lra y_c^1)$,
        $\bigwedge_{c \in \CLASS} (z_c \lra z_c^1)$, and
        $(\Root \lra \Root')$
    \end{enumerate}
}
\item
Finally, add the rule $(0,\varphi)$, where $\varphi$ asserts
\mhchanged{
    \begin{enumerate}[(a)]
    \item
        $\neg \Root$ (can return to parent), and
    \item
        $\Pop'$ (flag the return), and
    \item
        $\bigwedge_{c \in \CLASS} (y_c \lra x_c')$ (return classes to parent).
    \end{enumerate}
}
\end{itemize}
These boolean formulas can easily be represented as BDDs
\mhchanged{
    of linear size
}
(see \shortlong{full version}{appendix}).

Continuing with our translation, the bit that needs to be toggled on is
$y_\alpha$. We now construct
the initial configuration for our bit-toggling problem.
For each subset $X \subseteq \CLASS$ and a function $q: \varsetV \to
\{0,1\}$, define the function $I_{X,f,q}:
(\varsetV \cup \varsetW) \to \{0,1\}$ as follows:
$I_{X,f,q}(\Root) := f(\Root)$, $I_{X,f,q}(\Pop) := q(\Pop)$, and for each
$c \in \CLASS$:
(i) $I_{X,f,q}(x_c) := q(x_c)$, (ii) $I_{X,f,q}(y_c) = 1$ iff $c \in X$,
and (iii) $I_{X,f,q}(z_c) := f(c)$. We shall write $I_{X,f}$ to mean
$I_{X,f,q}$ with $q(x) = 0$ for each $x \in \varsetV$.
Define the initial configuration $I_0$ as the
function $I_{\treelabeling_0(\epsilon),f}$.

Let us now analyse our translation. The translation is easily seen to run
in polynomial time. In fact, with a more careful analysis,
one can show that the output sPDS is of linear size and that
the translation can be implemented in polynomial time.
Correctness of our translation immediately follows from the
following technical lemma:

\begin{lemma}
For each subset $X \subseteq \CLASS$, the following are equivalent:
\begin{description}
\item[(A1)] There exists a tree $\tree = (\treedomain,\treelabeling)
    \in \trees(\treelabels)$ such that $\tree_0 \to_{\trs,f}^* \tree$
    and $\treelabeling(\epsilon) = X$.
\item[(A2)] There exists $q': \varsetV \to \{0,1\}$ such that
$I_{\treelabeling_0(\epsilon),f} \to_\pds^* I_{X,f,q'}$.
\end{description}
\label{lm:spds-correctness}
\end{lemma}
This lemma intuitively states that the constructed sPDS performs a ``faithful
simulation'' of $\trs$. Moreover, the direction (A2)$\Rightarrow$(A1) gives \emph{soundness}
of our reduction, while (A1)$\Rightarrow$(A2) gives \emph{completeness}
of our reduction.
The proof is very technical,
which we relegate to the \shortlong{full version}{appendix}.

\paragraph{Running Example}

We show how the sPDS constructed can determine that
\acmeasychair{}{the required class}
\alchanged{$\myclass{\langle \downarrow^\ast \rangle \Success}$}
can be added to the root node of the tree, thus implying that there are no redundant selectors.

We write configurations of a symbolic pushdown system using the notation
$\sconfig{X}{Y_1 \ldots Y_m}$ where $X$ is the set of classes $c$ such that the
variable $x_c$ is true, and $Y_1 \ldots Y_m$ is a stack of sets of classes
(with the top on the right) such that each $Y_i$ is a set of classes $c$ such
that $y_c$ is true at stack position $i$.  Note, we do not show the $z_c$
variables' values since they are only needed to evaluate $\langle \uparrow
\rangle$ modalities, which do not appear in our example.

First we apply the rule $(\Root,\addchild{\Team})$ and then $(\Team, \addchild{\PlayerOne})$.  Each step pushes a new item onto the stack.  By executing these steps the sPDS is exploring a branch from $\Root$ to $\PlayerOne$.
\[
    \begin{array}{l}
        \sconfig{\emptyset}{\set{\Root}}
        \rightarrow \\
        \sconfig{\emptyset}{\set{\Root} \set{\Team}}
        \rightarrow \\
        \sconfig{\emptyset}{\set{\Root} \set{\Team} \set{\PlayerOne}}
    \end{array}
\]

At this point there's nothing more we can do with the $\PlayerOne$ node.  Hence,
the sPDS backtracks, remembering in its control state the classes contained in
the child it has returned from.  Once this information is remembered in its
control state, it knows that a child is labelled $\PlayerOne$ and hence $(\langle \downarrow \rangle \PlayerOne, \addclass{\myclass{\langle \downarrow \rangle \PlayerOne}})$ can be applied.
\[
    \begin{array}{l}
        \sconfig{\emptyset}{\set{\Root} \set{\Team} \set{\PlayerOne}}
        \rightarrow \\
        \sconfig{\set{\PlayerOne}}{\set{\Root} \set{\Team}}
        \rightarrow \\
        \sconfig{\set{\PlayerOne}}
                {\set{\Root}
                 \set{\Team,
                      \myclass{\langle \downarrow \rangle \PlayerOne}}}
    \end{array}
\]
Now the sPDS can repeat the analogous sequence of actions to obtain that $\myclass{\langle \downarrow \rangle \PlayerTwo}$ can also label the $\Team$ node.
\[
    \begin{array}{c}
        \sconfig{\set{\PlayerOne}}
                {\set{\Root}
                 \set{\Team,
                      \myclass{\langle \downarrow \rangle \PlayerOne}}}
        \rightarrow \\
        \sconfig{\emptyset}
                {\set{\Root}
                 \set{\Team,
                      \myclass{\langle \downarrow \rangle \PlayerOne}}
                 \set{\PlayerTwo}}
        \rightarrow \\
        \sconfig{\set{\PlayerTwo}}
                {\set{\Root}
                 \set{\Team,
                      \myclass{\langle \downarrow \rangle \PlayerOne}}}
        \rightarrow \\
        \sconfig{\set{\PlayerTwo}}
                {\set{\Root}
                 \set{\Team,
                      \myclass{\langle \downarrow \rangle \PlayerOne},
                      \myclass{\langle \downarrow \rangle \PlayerTwo}}}
    \end{array}
\]
Now we are able to deduce that $\Success$ can also label the $\Team$ node.
\[
    \begin{array}{c}
        \sconfig{\set{\PlayerTwo}}
                {\set{\Root}
                 \set{\Team,
                      \myclass{\langle \downarrow \rangle \PlayerOne},
                      \myclass{\langle \downarrow \rangle \PlayerTwo}}}
        \rightarrow \\
        \sconfig{\set{\PlayerTwo}}
                {\set{\Root}
                 \set{\ldots,
                      \myclass{
                          \langle \downarrow \rangle \PlayerOne
                          \land
                          \langle \downarrow \rangle \PlayerTwo
                      }}}
        \rightarrow \\
        \sconfig{\emptyset}
                {\set{\Root}
                 \set{\ldots}
                 \set{\Success}}
    \end{array}
\]
We are then able to backtrack to the root node, accumulating the information that $\langle \downarrow^\ast \rangle \Success$ holds at the root node, and thus
$(\langle \downarrow^\ast \rangle \Success)$ is not redundant.
\[
    \begin{array}{c}
        \sconfig{\emptyset}
                {\set{\Root}
                 \set{\ldots}
                 \set{\Success}}
        \rightarrow \\
        \sconfig{\emptyset}
                {\set{\Root}
                 \set{\ldots}
                 \set{\ldots, \langle \downarrow^\ast \rangle \Success}}
        \rightarrow \\
        \sconfig{\set{\ldots, \langle \downarrow^\ast \rangle \Success}}
                {\set{\Root}
                 \set{\ldots}}
        \rightarrow \\
        \sconfig{\set{\ldots, \langle \downarrow^\ast \rangle \Success}}
                {\set{\Root}
                 \set{\ldots, \langle \downarrow^\ast \rangle \Success}}
        \rightarrow \\
        \sconfig{\set{\ldots, \langle \downarrow^\ast \rangle \Success}}
                {\set{\Root}}
        \rightarrow \\
        \sconfig{\set{\ldots}}
                {\set{\ldots,
                      \myclass{\langle \downarrow^\ast \rangle \Success}}}
    \end{array}
\]

\section{Experiments}
\label{sec:experiments}

We have implemented our approach in a new tool TreePed which is available
\anonornot{for download~\cite{treeped-anon}}{for download~\cite{treeped}}.  We
tested it on several case studies.  Our implementation contains two main
components: a proof-of-concept translation from HTML5 applications using jQuery
to our model, and a redundancy checker (with non-redundancy witness generation)
for our model.  Both tools were developed in Java.  The redundancy checker uses
jMoped~\cite{SBSE07} to analyse symbolic pushdown systems.  In the following
sections we discuss the redundancy checker, translation from jQuery, and the
results of our case studies.

\subsection{The Redundancy Checker}

The main component of our tool implements the redundancy checking algorithm for
proving Theorem~\ref{th:exp}.  Largely, the algorithm is implemented directly.
The most interesting differences are in the use of jMoped to perform the
analysis of sPDSs to answer class-adding checks.  In the following, assume a
tree $\tree_v$, rewrite rules $\trs$, and a set $\CLASS$ of classes.

\paragraph{Optimising the sPDS}
For each class $c \in \CLASS$ we construct an sPDS.  We can optimise by
restricting the set of rules in $\trs$ used to build the sPDS.  In
particular, we can safely ignore all rules in $\trs$ that cannot appear in a
sequence of rules leading to the addition of $c$.  To do this, we begin with
the set of all rules that either directly add the class $c$ or add a child to
the tree (since these may lead to new nodes matching other rules).  We then add
all rules that directly add a class $c'$ that appears in the guard of any rules
included so far.  This is iterated until a fixed point is reached.  The rules
in the fixed point are the rules used to build the sPDS.

\paragraph{Reducing the number of calls.}
We re-implemented the global backwards reachability analysis (and witness
generation) of Moped~\cite{schwoon-thesis} in jMoped.  This means that a single call
to jMoped can allow us to obtain a BDD representation of all initial
configurations of the sPDSs obtained from class-adding problems $\langle
\tree_v,f,c,\trs\rangle$ that have a positive answer to the class-adding
problem for a given class $c \in \CLASS$.  Thus, we only call jMoped once per
class.

\subsection{Translation from HTML5}

The second component of our tool provides a proof-of-concept prototypical
translation from HTML5 using jQuery to our model. We provide a detailed
description in the Appendix
\mhchanged{
    and provide a small example below.
}
There are three main parts to the translation.
\begin{itemize}
    \item
        The DOM tree of the HTML document is directly translated to a tree in
        our model.  We use classes to encode element types (e.g. \verb+div+ or
        \verb+a+), IDs and CSS classes.

    \item
        We support a subset of CSS covering the most common selectors and all
        selectors in our case studies.  Each selector in the
        CSS stylesheet is translated to a guard to be analysed for redundancy.
        For
        pseudo-selectors such as \texttt{g:hover} and \texttt{g:before} we
        simply check for the redundancy of \texttt{g}.

    \item
        Dynamic rules are extracted from the JavaScript in the document
        by identifying jQuery calls and generating rules as outlined in
        Section~\ref{sec:trs-jquery}. Developing a translation tool that covers
        all covers all aspects of such an extremely rich and complex language as
        JavaScript is a difficult problem~\cite{AM14}. Our proof-of-concept
        prototype covers many, but by no means all, interesting features of the
        language.  The implemented translation is described in more detail in
        the \shortlong{full version}{appendix}.
\end{itemize}
Sites formed of multiple pages with common CSS files are supported by
automatically collating the results of independent page analyses and reporting
site-wide redundancies.

\mh{the paragraph below is new but not in a change block because the verb
command doesn't like it}
    To give a flavour of the translation of a single line we recall one of the
    rules from the example in Figure~\ref{fig:sanwebe-example}, except we adjust
    it to a call \verb+addClass()+ instead of \verb+remove()+ (since calls to
    \verb+remove()+ are ignored by our abstraction).
\begin{minted}{javascript}
    $('.input_wrap').find('.delete')
                    .parent('div')
                    .addClass('deleted');
\end{minted}
    We can translate this rule inductively.  From the initial jQuery call
    \verb+$('.input_wrap')+ we obtain a guard that is simply
    $\text{\texttt{.input\_wrap}}$ that matches any node with the class
    \verb+input_wrap+.  We can then extend this guard to handle the
    \verb+find()+ call.  This looks for any child of the currently matched node
    that has the class \verb+delete+.  Thus, we build up the guard to
    \[
        \text{\texttt{.delete}} \wedge {\langle \uparrow^+ \rangle}
        \text{\texttt{.input\_wrap}}
    \]
    Next, because of the call to \verb+.parent()+ we have to extend the guard
    further to match the parent of the currently selected node, and enforce that
    the parent node is a \verb+div+.
    \[
        \text{\texttt{div}} \wedge \langle \downarrow \rangle(
        \text{\texttt{.delete}} \wedge {\langle \uparrow^+ \rangle}
        \text{\texttt{.input\_wrap}} )
    \]
    Finally, we encounter the call to \verb+addClass()+ which we translate to an
    $\addclass{\set{\text{\texttt{deleted}}}}$ rule.
    \acmeasychair{
        \begin{multline*}
            (\text{\texttt{div}} \wedge \langle \downarrow \rangle(
            \text{\texttt{.delete}} \wedge {\langle \uparrow^+ \rangle}
            \text{\texttt{.input\_wrap}} ), \\
            \addclass{\set{\text{\texttt{deleted}}}})
        \end{multline*}
    }{
        \[
            (\text{\texttt{div}} \wedge \langle \downarrow \rangle(
            \text{\texttt{.delete}} \wedge {\langle \uparrow^+ \rangle}
            \text{\texttt{.input\_wrap}} ),
            \addclass{\set{\text{\texttt{deleted}}}})
        \]
    }
\mh{end of new paragraph}

\subsection{Case Studies}

We performed several case studies.  One is based on the Igloo example
from the benchmark suite of the dynamic CSS analyser Cilla~\cite{MesbahM12},
and is described in detail below.  Another (and the largest) is based
on the Nivo Slider plugin~\cite{nivo-slider} for animating transitions between
a series of images.  The remaining examples are hand built and use jQuery to
make frequent additions to and removals from the DOM tree.  The first
\texttt{bikes.html} allows a user to select different frames, wheels and
groupsets to build a custom bike, \texttt{comments.html} displays a comments
section that is loaded dynamically via an AJAX call, and
\texttt{transactions.html} is a finance page where previous transactions are
loaded via AJAX and new transactions may be added and removed via a form.  The
example in Figure~\ref{fig:sanwebe-example} is \texttt{example.html} and
\texttt{example-up.html} is the version without the limit on the number of
input boxes.  These examples are available in the \texttt{src/examples/html}
directory of the tool distribution~\anonornot{\cite{treeped-anon}}{\cite{treeped}}.

All case studies contained non-trivial CSS selectors whose
redundancy depended on the dynamic behaviour of the system.  In each case our
tool constructed a rewrite system following the process outlined above, and
identified all redundant rules correctly. Below we provide the answers
provided by UnCSS \cite{UnCSS} and Cilla \cite{Cilla} when they are
available\footnote{We did not manage
    to successfully set up Cilla \cite{MesbahM12}
    and so only provide the
    output for the Igloo example that has been provided by the authors
in \cite{Cilla}.}.

The experiments were run on a Dell Latitude e6320 laptop with 4Gb of RAM and
four 2.7GHz Intel i7-2620M cores.  We used OpenJDK 7, using the argument
\hbox{``-Xmx''} to limit RAM usage to 2.5Gb.  The results are shown in
Table~\ref{tbl:case-studies}.  \textit{Ns} is the initial number of elements in
the DOM tree, \textit{Ss} is the number of CSS selectors (with the
number of redundant selectors shown in brackets), \textit{Ls} is the number of
Javascript lines reported by \texttt{cloc}, and \textit{Rs} is the number of
rules in the rewrite system obtained from the
JavaScript\footnote{Not including the number of rules required to
represent the CSS selectors.} after simplification (unsimplified rules may
have arbitrarily complex guards).  The figures for the Igloo example are
reported per file or for the full analysis as appropriate.

\mhchanged{
    We remark that the translation from JavaScript and jQuery is the main
    limitation of the tool in its application to industrial websites, since we
    do not support JavaScript and jQuery in its full generality.  However, we
    also note that the number of rules required to model websites is often
    much smaller than the size of the code.  For example, the Nivo-Slider
    example contains 501 lines of JavaScript, but its abstraction only
    requires 21 rules in our model.  It is the number of these rules and the
    number of CSS selectors (and the number of classes appearing in the guards)
    that will have the main effect on the scalability of the tool.

    Although we report the number of nodes in the webpages of our examples,
    checking CSS matching against these pre-existing nodes is no harder than
    standard CSS matching.  The dynamic addition of nodes to the webpage is
    where symbolic pushdown analysis is required, hence the number of rewrite
    rules gives a better indication of the difficulty of an analysis instance.
}

\begin{table}
    \acmeasychair{\footnotesize}{}
    \begin{center}
    \begin{tabular}{lccccc}
            Case Study                   & Ns  & Ss  & Ls & Rs & Time  \\

        \midrule

            \verb+bikes.html+            &  22 &     18 (0) &  97 &  37 & 3.6s \\
            \verb+comments.html+         &   5 &     13 (1) &  43 &  26 & 2.9s \\
            \verb+example.html+          &  11 & \phantom{0}1 (0) &  28 &   4 & .6s  \\
            \verb+example-up.html+       &   8 & \phantom{0}1 (1) &  15 &   3 & .6s  \\
            \verb+igloo/+                &     &   261 (89) &     &     & 3.4s \\
            \verb+  index.html+          & 145 &            &  24 &   1 &      \\
            \verb+  engineering.html+    & 236 &            &  24 &   1 &      \\
            \verb+Nivo-Slider/+          &     &            &     &     &      \\
            \verb+  demo.html+           &  15 &  172 (131) & 501 &  21 & 6.3s \\
            \verb+transactions.html+     &  19 &      9 (0) &  37 &   6 & 1.6s \\
    \end{tabular}
    \end{center}
    \caption{\label{tbl:case-studies}Case study results.}
\end{table}

\alchanged{
\paragraph{Example from Figure \ref{fig:sanwebe-example}}
TreePed suggests that the selector \texttt{.warn} might be reachable and,
therefore, is not deleted. We have run UnCSS on this example,
which incorrectly identifies \texttt{.warn} as unreachable.
}

\paragraph{The Igloo example}
\alchanged{
The Igloo example is a mock company website with a home page
(\texttt{index.html}) and an engineering services page
(\texttt{engineering.html}). There are a total of 261 CSS selectors, out of
which 89 of them are actually redundant (we have verified this by hand).
UnCSS reports 108 of them are redundant rules.

We mention one interesting selector which both Cilla and UnCSS have
identified as redundant, but is actually reachable. The Igloo main page
includes a
search bar that contains some placeholder text which is present only when the search bar is empty and does not have
focus.  Placeholder text is supported by most modern browsers, but not all
(e.g. IE9), and the
page contains a small amount of JavaScript to simulate this functionality when
it is not provided by the browser. The relevant code is shown in Figure
\ref{fig:igloo}.
}
\begin{figure}
\begin{minted}[fontsize=\footnotesize]{javascript}
(function($) {
    function supportsInputPlaceholder() {
      var el = document.createElement('input');
      return 'placeholder' in el;
    }
    $(function() {
      if (!supportsInputPlaceholder()) {
        var searchInput = $('#searchInput'),
          placeholder = searchInput.attr('placeholder');
        searchInput.val(placeholder).focus(function() {
           var $this = $(this);
           $this.addClass('touched');
             ...
        })
             ...
      });
    }
})(jQuery);
\end{minted}
\caption{JavaScript code snippet from Igloo example\label{fig:igloo}}
\end{figure}
In particular, the CSS class \texttt{touched},
and rule
\begin{minted}{css}
    #search .touched { color: #333; }
\end{minted}
are used for this purpose.  
Both Cilla and UnCSS incorrectly claim that the rule is redundant.
Since we only
identify genuinely redundant rules, our tool correctly does not report the rule
as redundant.

In addition, TreePed identified a further unexpected mistake in
Igloo's CSS.  The rule
\begin{minted}{css}
 h2 a:hover, h2 a:active, h2 a:focus
 h3 a:hover, h3 a:active, h2 a:focus { ... }
\end{minted}
is missing a comma from the end of the first line. This results in the redundant
selector ``\verb+h2 a:focus h3 a:hover+'' rather than two separate selectors as
was intended.  Cilla did not report this redundancy as it appears to ignore all
CSS rules with pseudo-selectors.  Finally, we remark that the second line of the
above rule contains a further error: ``\verb+h2 a:focus+'' should in fact be
``\verb+h3 a:focus+''. This raises the question of selector subsumption, on
which there is already a lot of research work (e.g. see
\cite{BGL14,MesbahM12}).  These algorithms may be incorporated into
our tool to obtain a further size reduction of CSS files.

\alchanged{
\paragraph{The Nivo-Slider example}
Nivo-Slider \cite{nivo-slider} is an easy-to-use image slider JavaScript
package that heavily employs jQuery. In particular, it provides some
beautiful transition effects when displaying a gallery of images. In this
case study, we use a demo file that displays a series of four images. The file
contains a total of 172 CSS selectors, out of which 131 are actually redundant
(we verified this by hand). UnCSS reports 152 redundant CSS selectors (i.e.
about 50\% of false positives).

\begin{samepage}
    We mention the following interesting rule that UnCSS reports as redundant:
    }
    \begin{minted}{css}
        .nivoSlider a.nivo-imageLink {
            ...
            z-index:6;
            ...
        }
    \end{minted}
\end{samepage}
\alchanged{
Recall that the \texttt{z-index} of an HTML element specifies the vertical
stack order; a greater stack order means that the element is \emph{in front}
of an element with a lower stack order. This CSS rule, among others, sets the
\texttt{z-index} of an HTML element that matches the selector to a higher
value. In effect, this allows a hyperlink that overlaps with the second image
in the series to be clicked (which takes the user to a different web page).
Removing this rule disables the hyperlink from the page.
The code snippet in Figure~\ref{fig:nivo-html-snippet} provides the relevant part of the HTML document.
}
\begin{figure}
    \begin{minted}[fontsize=\footnotesize]{html}
    ...
    <div id="slider" class="nivoSlider">
      ...
      <a href="http://dev7studios.com">
         <img src="images/up.jpg"
              data-thumb="images/up.jpg"
              alt=""
              title="This is an example of a caption"/>
      </a>
    ...
    <div>
    ...
    \end{minted}
    \caption{\label{fig:nivo-html-snippet}HTML code snippet for Nivo-Slider demo.}
\end{figure}

\alchanged{
In particular, the class \texttt{nivo-imageLink} will be added by the
JavaScript bit of the page to the the depicted hyperlink element above,
as is shown in the JavaScript code snippet in Figure \ref{fig:nivo}.
}
\begin{figure}
\begin{minted}[fontsize=\footnotesize]{javascript}
    var slider = $('#slider');
    slider.data('nivo:vars', vars)
        .addClass('nivoSlider');

    // Find our slider children
    var kids = slider.children();
    kids.each(function() {
       var child = $(this);
       var link = '';
       if(!child.is('img')){
         if(child.is('a')){
           child.addClass('nivo-imageLink');
           link = child;
         }
         ...
       }
       ...
    }
\end{minted}
\caption{JavaScript code snippet for Nivo-Slider demo.\label{fig:nivo}}
\end{figure}
\alchanged{
In contrast to UnCSS, TreePed correctly identifies that the above CSS rule may
be used.
}

\section{Conclusion and Future Work}
\label{sec:conclusion}

At the moment our translation from HTML5 applications to our
tree-rewriting model is prototypical and does not
incorporate many features that such a rich and complex language as JavaScript has,
especially in the presence of libraries.
Therefore,
an important research direction is to develop a more robust translation, perhaps
by building on top of existing JavaScript static analysers like
WALA~\cite{SSDT13,SDCST12} and TAJS~\cite{AM14,JMT09,JMM11}. As
Andreasen and M{\o}ller \cite{AM14} \mhchanged{describe}, a static analysis of JavaScript
in the presence of jQuery is presently a formidable task for existing static
analysers for JavaScript. For this reason, we do not expect the task of building
a more robust translation to our tree-rewriting model to be easy.

\alchanged{
Our technique
should not be seen as competing with dynamic analysis techniques for identifying
redundant CSS rules (e.g. UnCSS and Cilla). In fact, they can be
combined to obtain a more precise CSS redundancy checker. Our tool TreePed
attempts to output the definitely redundant rules, while Uncss/Cilla attempts to
output those that are definitely
non-redundant. The complement of the union of these sets is the ``don’t know
set'', which can (for example) be checked manually by the developer. For future
work, we would like to combine static and dynamic analysis to build a more
precise and robust CSS redundancy checker.
}

Another direction is to find better ways of overapproximating redundancy
problems for the undecidable class $\Gtrsclass$ of rewrite
systems using our monotonic abstractions (other than replacing
non-positive guards by $\top$). In particular, for a general guard,
can we automatically construct a more precise positive guard that
serves as an overapproximation? A more precise abstraction would be useful in
identifying redundant CSS rules with negations in the node selectors, which
can be found in real-world HTML5 applications (though not as common as
those rules without negations).

\mh{Added new paragraph below, but since it uses verbatim environments it can't
be in a changed block}

We may also consider other ways of improving CSS performance using the results
of our analysis.  For example, a descendant selector
\begin{minted}{css}
    .a .b { ... }
\end{minted}
is less efficient than a direct child selector
\begin{minted}{css}
    .a > .b { ... }
\end{minted}
since the descendant selector must search through all descendants of a given
node.  However, it is very common for the former to be used in place of the
latter.  If we can identify that the guard
$\texttt{a} \land \langle \downarrow \rangle \langle \downarrow^+ \rangle \texttt{b}$
is never matched, while
$\texttt{a} \land \langle \downarrow \rangle \texttt{b}$
is matched, then we can safely replace the descendant selector with the direct
child selector.

\acmeasychair{\acks}{\paragraph{Acknowledgments}}
We sincerely thank Max Schaefer for fruitful discussions.
Hague is supported by
the Engineering and Physical Sciences Research Council [EP/K009907/1].  Lin is
supported by a Yale-NUS Startup Grant. Ong is partially supported by a visiting
professorship, National University of Singapore.


\bibliographystyle{plain}
\bibliography{references}

\begin{thebibliography}{10}

\bibitem{simple-html-nolimit-anon}
\url{https://bitbucket.org/TreePed/treeped/raw/master/src/examples/html/example-up.html},
  March 2015.

\bibitem{simple-html-anon}
\url{https://bitbucket.org/TreePed/treeped/raw/master/src/examples/html/example.html},
  March 2015.

\bibitem{simple-html}
\url{http://treeped.bitbucket.org/example.html}, March 2015.

\bibitem{AbdullaCJT96}
Parosh~Aziz Abdulla, Karlis Cerans, Bengt Jonsson, and Yih-Kuen Tsay.
\newblock General decidability theorems for infinite-state systems.
\newblock In {\em LICS}, pages 313--321, 1996.

\bibitem{rtmc}
Parosh~Aziz Abdulla, Bengt Jonsson, Pritha Mahata, and Julien d'Orso.
\newblock Regular tree model checking.
\newblock In {\em {CAV}}, pages 555--568, 2002.

\bibitem{AHV94}
S.~Abiteboul, R.~Hull, and V.~Vianu.
\newblock {\em Foundations of Databases}.
\newblock Addison-Wesley, 1994.

\bibitem{ABM04}
Serge Abiteboul, Omar Benjelloun, and Tova Milo.
\newblock Positive active {XML}.
\newblock In {\em PODS}, pages 35--45, 2004.

\bibitem{ASV09}
Serge Abiteboul, Luc Segoufin, and Victor Vianu.
\newblock Static analysis of active {XML} systems.
\newblock {\em {ACM} Trans. Database Syst.}, 34(4), 2009.

\bibitem{AM14}
E.~Andreasen and A.~M{\o}ller.
\newblock Determinacy in static analysis for jquery.
\newblock In {\em {OOPSLA}}, pages 17--31, 2014.

\bibitem{average-web-page}
\url{http://www.websiteoptimization.com/speed/tweak/average-web-page/}.

\bibitem{BPW02}
James Bailey, Alexandra Poulovassilis, and Peter~T. Wood.
\newblock An event-condition-action language for {XML}.
\newblock In {\em {WWW}}, pages 486--495, 2002.

\bibitem{bebop}
Thomas Ball and Sriram~K. Rajamani.
\newblock Bebop: {A} symbolic model checker for boolean programs.
\newblock In {\em {SPIN}}, pages 113--130, 2000.

\bibitem{Maffeis}
M.~Bodin, A.~Chargu{\'{e}}raud, D.~Filaretti, P.~Gardner, S.~Maffeis,
  D.~Naudziuniene, A.~Schmitt, and G.~Smith.
\newblock A trusted mechanised javasript specification.
\newblock In {\em {POPL}}, pages 87--100, 2014.

\bibitem{BGL14}
Mart{\'{\i}} Bosch, Pierre Genev{\`{e}}s, and Nabil Laya{\"{\i}}da.
\newblock Automated refactoring for size reduction of {CSS} style sheets.
\newblock In {\em {DocEng}}, pages 13--16, 2014.

\bibitem{Cilla}
\url{https://github.com/saltlab/cilla/}.

\bibitem{FGN08}
Wenfei Fan, Floris Geerts, and Frank Neven.
\newblock Expressiveness and complexity of {XML} publishing transducers.
\newblock {\em {ACM} Trans. Database Syst.}, 33(4), 2008.

\bibitem{FinkelS01}
Alain Finkel and Ph. Schnoebelen.
\newblock Well-structured transition systems everywhere!
\newblock {\em Theor. Comput. Sci.}, 256(1-2):63--92, 2001.

\bibitem{grohe-book}
J.~Flum and M.~Grohe.
\newblock {\em Parameterized Complexity Theory}.
\newblock Springer, 2006.

\bibitem{GMSZ08}
Blaise Genest, Anca Muscholl, Olivier Serre, and Marc Zeitoun.
\newblock Tree pattern rewriting systems.
\newblock In {\em ATVA}, pages 332--346, 2008.

\bibitem{GMW10}
Blaise Genest, Anca Muscholl, and Zhilin Wu.
\newblock Verifying recursive active documents with positive data tree
  rewriting.
\newblock In {\em FSTTCS}, pages 469--480, 2010.

\bibitem{GenevesLQ12}
Pierre Geneves, Nabil Layaida, and Vincent Quint.
\newblock {On the Analysis of Cascading Style Sheets}.
\newblock In {\em WWW}, pages 809--818, 2012.

\bibitem{GL14}
Stefan G{\"{o}}ller and Anthony~Widjaja Lin.
\newblock Refining the process rewrite systems hierarchy via ground tree
  rewrite systems.
\newblock {\em {ACM} Trans. Comput. Log.}, 15(4):26:1--26:28, 2014.

\bibitem{Gatekeeper}
S.~Guarnieri and V.~B. Livshits.
\newblock {GATEKEEPER:} mostly static enforcement of security and reliability
  policies for javascript code.
\newblock In {\em {USENIX}}, pages 151--168, 2009.

\bibitem{Actarus}
S.~Guarnieri, M.~Pistoia, O.~Tripp, J.~Dolby, S.~Teilhet, and R.~Berg.
\newblock Saving the world wide web from vulnerable javascript.
\newblock In {\em {ISSTA}}, pages 177--187, 2011.

\bibitem{Hague14}
Matthew Hague.
\newblock Senescent ground tree rewrite systems.
\newblock In {\em {CSL-LICS}}, page~48, 2014.

\bibitem{JC09}
D.~Jang and K.~{-}Moo Choe.
\newblock Points-to analysis for javascript.
\newblock In {\em {SAC}}, pages 1930--1937, 2009.

\bibitem{JMM11}
S.~H. Jensen, M.~Madsen, and A.~M{\o}ller.
\newblock Modeling the {HTML} {DOM} and browser {API} in static analysis of
  javascript web applications.
\newblock In {\em {SIGSOFT/FSE}}, pages 59--69, 2011.

\bibitem{JMT09}
S.~H. Jensen, A.~M{\o}ller, and P.~Thiemann.
\newblock Type analysis for javascript.
\newblock In {\em {SAS}}, pages 238--255, 2009.

\bibitem{jQuery}
\url{http://jquery.com/}.

\bibitem{jQuery-widespread}
\url{http://trends.builtwith.com/javascript/jQuery}.

\bibitem{MONA}
N.~Klarlund, A.~M{\o}ller, and M.~I. Schwartzbach.
\newblock {MONA} implementation secrets.
\newblock {\em Int. J. Found. Comput. Sci.}, 13(4):571--586, 2002.

\bibitem{getafix}
Salvatore {La Torre}, Parthasarathy Madhusudan, and Gennaro Parlato.
\newblock Analyzing recursive programs using a fixed-point calculus.
\newblock In {\em {PLDI}}, pages 211--222, 2009.

\bibitem{jQuery-ecoop13}
Benjamin~S. Lerner, Liam Elberty, Jincheng Li, and Shriram Krishnamurthi.
\newblock Combining form and function: Static types for jquery programs.
\newblock In {\em {ECOOP}}, pages 79--103, 2013.

\bibitem{libkin-survey}
Leonid Libkin.
\newblock Logics for unranked trees: An overview.
\newblock {\em Logical Methods in Computer Science}, 2(3), 2006.

\bibitem{Lin12b}
Anthony~Widjaja Lin.
\newblock Accelerating tree-automatic relations.
\newblock In {\em FSTTCS}, pages 313--324, 2012.

\bibitem{Lin12}
Anthony~Widjaja Lin.
\newblock Weakly-synchronized ground tree rewriting.
\newblock In {\em MFCS}, pages 630--642, 2012.

\bibitem{Loding06}
Christof L{\"o}ding.
\newblock Reachability problems on regular ground tree rewriting graphs.
\newblock {\em Theory Comput. Syst.}, 39(2):347--383, 2006.

\bibitem{LS07}
Christof L{\"o}ding and Alex Spelten.
\newblock Transition graphs of rewriting systems over unranked trees.
\newblock In {\em MFCS}, pages 67--77, 2007.

\bibitem{MBPS05}
Sebastian Maneth, Alexandru Berlea, Thomas Perst, and Helmut Seidl.
\newblock {XML} type checking with macro tree transducers.
\newblock In {\em {PODS}}, pages 283--294, 2005.

\bibitem{reactive-systems-book}
Zohar Manna and Amir Pnueli.
\newblock {\em Temporal Verification of Reactive Systems: Safety}.
\newblock Springer, 1995.

\bibitem{Marx05}
Maarten Marx.
\newblock Conditional xpath.
\newblock {\em {ACM} Trans. Database Syst.}, 30(4):929--959, 2005.

\bibitem{Mayr00}
Richard Mayr.
\newblock Process rewrite systems.
\newblock {\em Inf. Comput.}, 156(1-2):264--286, 2000.

\bibitem{MTM14}
Davood Mazinanian, Nikolaos Tsantalis, and Ali Mesbah.
\newblock Discovering refactoring opportunities in cascading style sheets.
\newblock In {\em FSE}, pages 496--506, 2014.

\bibitem{McMillan}
K.~L. McMillan.
\newblock {\em Symbolic model checking}.
\newblock Kluwer, 1993.

\bibitem{MesbahM12}
Ali Mesbah and Shabnam Mirshokraie.
\newblock {Automated Analysis of CSS Rules to Support Style Maintenance}.
\newblock In {\em ICSE}, pages 408--418, 2012.

\bibitem{MB10}
Leo~A. Meyerovich and Rastislav Bodik.
\newblock Fast and parallel webpage layout.
\newblock In {\em WWW}, pages 711--720, 2010.

\bibitem{moped}
\url{http://www2.informatik.uni-stuttgart.de/fmi/szs/tools/moped/}.

\bibitem{nivo-slider}
\url{https://github.com/gilbitron/Nivo-Slider}.

\bibitem{sanwebe-addremove}
\url{http://www.sanwebe.com/2013/03/addremove-input-fields-dynamically-with-jquery}.

\bibitem{SSDT13}
M.~Sch{\"{a}}fer, M.~Sridharan, J.~Dolby, and F.~Tip.
\newblock Dynamic determinacy analysis.
\newblock In {\em {PLDI}}, pages 165--174, 2013.

\bibitem{Sch92}
Bernd{-}Holger Schlingloff.
\newblock Expressive completeness of temporal logic of trees.
\newblock {\em Journal of Applied Non-Classical Logics}, 2(2):157--180, 1992.

\bibitem{schwoon-thesis}
Stefan Schwoon.
\newblock {\em Model-Checking Pushdown Systems}.
\newblock PhD thesis, Technischen Universit\"{a}t M\"{u}nchen, 2002.

\bibitem{Sipser}
M.~Sipser.
\newblock {\em Introduction to The Theory of Computation}.
\newblock Cengage Learning, 3 edition, 2012.

\bibitem{SDCST12}
M.~Sridharan, J.~Dolby, S.~Chandra, M.~Sch{\"{a}}fer, and F.~Tip.
\newblock Correlation tracking for points-to analysis of javascript.
\newblock In {\em {ECOOP}}, pages 435--458, 2012.

\bibitem{SBSE07}
Dejvuth Suwimonteerabuth, Felix Berger, Stefan Schwoon, and Javier Esparza.
\newblock j{M}oped: {A} test environment for {J}ava programs.
\newblock In {\em {CAV}}, pages 164--167, 2007.

\bibitem{treeped-anon}
\url{https://bitbucket.org/TreePed/treeped}.

\bibitem{UnCSS}
\url{https://github.com/giakki/uncss}.

\bibitem{w3schools}
\url{http://www.w3schools.com/}.

\end{thebibliography}

\appendix
    \section{Proof of Proposition \ref{prop:trs-undec}}

We prove the undecidability of the $1$-redundancy problem when negation is allowed in the guards.
The proof is by a reduction from the undecidable control-state reachability
problem for
deterministic two counter machines (2-CM): given a 2-CM $\TM =
(\controls,\pdsrules,q_0, q_F)$ with counters $X$ and $Y$, decide whether there
exist a path from
$(q_0,0,0)$ to some configuration in $\{q_F\} \times \N \times \N$.  We will use $X$ and $Y$ to denote the two counters, and $Z$ to range over $\{X, Y\}$.  The two counter tests will be $=$ and $>$ and we use $\varominus$ to range over these tests.
We may assume that
 \begin{enumerate}[(i)]
 \item counter tests and counter increments / decrements never take place in the same transition, i.e., rules are of form $(q, [Z \varominus 0]) \to q'$ or $q \to (q',\inc(Z))$ or $q \to  (q', \dec(Z))$, and
 \item there are no self-loops in the counter machine i.e.~there is no transition of the form
     $(q,\cdot) \to q$.\mh{What does $\cdot$ mean?}
 \end{enumerate}

The initial tree $\tree_0 = (\treedomain_0,\treelabeling_0)$ is defined
to be $\treedomain_0 = \{\epsilon\}$ and $\treelabeling_0(\epsilon) = q_0$. We now
construct a rewrite system $\trs \in \Gtrsclass$.
The state of the machine is always recorded in the root node. The values of the
two counters $X$ and $Y$ will be encoded by the number of nodes at level 1 (i.e.
children of the root node) that are associated with class $X$ and $Y$,
respectively. More precisely, for each
transition rule $\sigma$ of the machine, we introduce a class $\sigma$. In
addition, we use the following four classes: $X$, $Y$, $\mathit{new}$, and
$\mathit{del}$. The main problem that we have to overcome in our reduction is
to simulate one transition of $\TM$ by several simpler steps of $\tree_0$
owing to the simplicity of rewrite operations that we allow in $\Gtrsclass$.
This is in fact the reason for adding the last two classes $\mathit{new}$ and
$\mathit{del}$, which are just ``flags'' to indicate whether a node (at level 1) is
``new'' or is ``scheduled to be deleted''.
Our rewrite system $\trs$ operates over $\treelabels := 2^\CLASS$
with $\CLASS = \controls \cup \pdsrules \cup \{X,Y,\mathit{new},\mathit{del}\}$.
Next we introduce the guards:
(a) $g(Z=0) := \neg \dir{\downarrow} Z$,
(b) $g(Z>0) := \dir{\downarrow} Z$,
(c) $\theta_1 := \neg \langle \downarrow \rangle \mathit{new} \wedge
    \neg \langle \downarrow \rangle \mathit{del}$, and
(d) $\theta_2 := \bigwedge_{\gamma\in\pdsrules} \neg \gamma$.
As we shall see, the guard $\theta_1 \wedge \theta_2$ means that there is no
``pending job''.
The rewrite rules for $\trs$ are as follows:
\begin{itemize}\small
\item If $\sigma = (q, [Z \varominus 0]) \to  q'$, then add the following rules:
\begin{enumerate}[(1)]
\item $(q \wedge g(Z \varominus 0) \wedge \theta_1 \wedge \theta_2,
    \addclass{\{ \sigma, q' \}})$
\item $(\sigma,\removeclass{\{ q, \sigma \}})$
\end{enumerate}

\item If $\sigma = q \to  (q', \inc(Z))$, then add:
\begin{enumerate}[(1)]
\item $(q \wedge \theta_1 \wedge \theta_2,\addclass{\{\sigma, q'\}}$
\item $(\sigma \wedge \theta_1,
    \addchild{\{ Z, \mathit{new} \}})$
\item $(\sigma \wedge \dir{\downarrow} (Z \wedge \mathit{new}),
    \removeclass{\{\sigma,q\}})$.
\item $(\mathit{new} \wedge \dir{\uparrow} \theta_2, \removeclass{ \{ \mathit{new} \} })$
\end{enumerate}

\item If $\sigma = q \to  (q', \dec(Z))$, then add:
\begin{enumerate}[(1)]
\item $(q \wedge \theta_1 \wedge \theta_2, \addclass{\{ \sigma, q' \}})$
\item $(Z \wedge \dir{\uparrow} (\sigma \wedge \neg \dir{\downarrow}
    \mathit{del}), \addclass{\{ \mathit{del} \}})$
\item $(\sigma \wedge \dir{\downarrow} (Z \wedge \mathit{del}),
    \removeclass{\sigma,q})$
\item $(\mathit{del} \wedge \dir{\uparrow} \theta_2, \removeclass{\{Z,
    \mathit{del} \} })$
\end{enumerate}
\end{itemize}
Observe that the
trees reachable from the initial tree is of height at most 1.
Possible labels of the root are only $q \in \controls$ and $\sigma
\in\pdsrules$. Possible labels of a leaf node (necessarily a child of the
root) are $X, Y, \mathit{new}$ and $\mathit{del}$.

It is not difficult to see that $\trs$ when initialised in $\tree_0$ performs
a faithful simulation of the counter machine $\TM$ from the configuration
$(q,0,0)$.
Our rewrite system $\trs$ is defined in such a way that each group of rewrite rules
introduced in $\trs$ to simulate a transition $\sigma$ of $\CM$ has to
be executed before the next transition of $\CM$ can be simulated by $\trs$.
In particular, to simulate a $\CM$ transition
$\sigma$ from a given state $q \in \controls$, our rewrite system $\trs$ needs to
add $\sigma$ to the root node, which can only be done
if the following properties are satisfied:
(1) the label of the root node is $\{q\}$ for some $q \in
\controls$, and (2) no child of the root has $\mathit{del}$ or $\mathit{new}$ in
the label.
This is owing to the conjuncts $\theta_1\wedge\theta_2$ in the guard.
In particular, this means that the root contains \emph{at most} one $\sigma$.
In addition, if $\sigma$ is a class of the form $q \to
(q',\inc(Z))$ or $q \to (q',\dec(Z))$, then $\sigma$ can only be removed
only after proper simulation of counter increments/decrements has been
performed.

In summary, we have that $\TM$ reaches the control state $q_F$ iff $q_F$ is
not redundant in $\trs$. This concludes our reduction.

    \section{Proof of Lemma \ref{lm:trleq}}

We prove that if $\tree \trleq \tree'$, then $\tree'$ satisfies at least the same guards as $\tree$ at each node $v$.
The proof is by induction on $g$. For the base cases, $g = \top$ is trivial, and that $v,\tree \models X$
implies $f(v),\tree' \models X$ holds since, by \textbf{(H2)}, we have $X \subseteq
\treelabeling(v) \subseteq \treelabeling'(f(v))$. We now proceed to the
inductive case. The cases of conjunction and disjunction are obvious.

Assume that $v,\tree \models \langle d \rangle g$ for some $d \in \{\uparrow,
\downarrow\}$. We will only show the case when
$d = \downarrow$; the other case can be
handled in the same way. This means that
$v.i,\tree \models g$ for some $i \in \N$. By induction, we have
$f(v.i),\tree' \models g$.
By \textbf{(H3)}, we have that $f(v.i) = f(v).j$ for some $j \in \N$.
Altogether, this implies that $f(v),\tree' \models \langle \downarrow\rangle
g$ as desired.

Assume that $v,\tree \models \langle d^\ast \rangle g$ for some $d \in \{\uparrow^\ast,
\downarrow^\ast\}$. We will only show the case when $d = \downarrow^\ast$; the other case
can be handled in the same way. The proof is by induction on the length $n$ of path
from $v$ to a descendant $vw$ (with $w \neq \epsilon$) satisfying $vw,\tree
\models g$. If $n = 0$, then by the previous paragraph we have
$f(v),\tree' \models g$ and so
$f(v),\tree' \models \langle \downarrow^\ast \rangle g$ as desired. If $n > 0$
and $w = iu'$ for some $i \in \N$ and $u' \in \N^*$, then $vi,\tree \models
\langle \downarrow^\ast \rangle g$ with a witnessing path of length $n-1$. Therefore,
by induction, we have $f(vi),\tree' \models \langle \downarrow^\ast \rangle g$.
Since $f(vi) = f(v).j$ for some $j \in \N$, it follows that $f(v),\tree'
\models \langle \downarrow \rangle\langle \downarrow^\ast \rangle g$ and so
$f(v),\tree' \models \langle \downarrow^\ast \rangle g$ as desired.

\section{Proof of Lemma \ref{lem:monotonic}}
\mh{I updated this section to reflect the new notation in lem~\ref{lem:monotonic}}
We prove that, whenever $\tree_1 \trleq \tree_2$ and $\tree_1 \to_\sigma \tree'_1$, then either $\tree'_1 \trleq \tree_2$ or $\tree_2 \to_\sigma \tree'_2$ with $\tree'_1 \trleq \tree'_2$.

Suppose $T_1 \to_\sigma T_2$ where $T_1 = (D_1, \treelabeling_1)$, $T'_1 = (D'_1,
\treelabeling'_1)$ and $\sigma = (g, \chi)$; and $T_1 \preceq T_2 = (D_2,
\treelabeling_2)$ is witnessed by an embedding $f : D_1 \to D_2$. Assume $v, T_1
\vDash g$ for some $v \in D_1$. Then $f(v), T_2 \vDash g$ by
Lemma \ref{lm:trleq}.
We consider each of the four cases of $\chi$ in turn.
\begin{asparaitem}
\item[\emph{Case:} $\chi = \addclass{X}$.] Then set $T_2' = (D_2', \treelabeling_2')$ with $D_2' := D_2$ and $\treelabeling_2' := \treelabeling_2[f(v) \mapsto \treelabeling_2(f(v)) \cup X)]$. 
Hence $T_2 \to_\sigma T_2'$. Moreover we have $T_1' \preceq T_2'$, as witnessed by $f : D'_1 \to D_2'$ (note that $D'_1 = D_1$ and $D_2' = D_2$).

\item[\emph{Case}: $\chi = \addchild{X}$.] Assume $D_1' = D_1 \cup \{v.i\}$ and $\treelabeling'_1 = \treelabeling_1[v.i \mapsto X]$ where $i$ is the number of children of $v$ in $T_1$. Let $i'$ be the number of children of $f(v)$ in $T_2$. Then set $T_2' = (D_2', \treelabeling_2')$ with $D_2' := D_2 \cup \{f(v).i'\}$ and $\treelabeling_2' := \treelabeling_2[f(v).i' \mapsto X]$. Then $T_2 \to_\sigma T_2'$ and $T_1' \preceq T_2'$ is witnessed by $f' : D'_1 \to D_2'$ such that 
$f' := f[v.i \mapsto f(v).i']$.

\item[\emph{Case}: $\chi = \removeclass{X}$ or $\chi = \removenode$.]
    In both cases, the function $f$ (restricted to $\treedomain'_1$) is still an
    embedding from
    $\tree'_1$ to $\tree_2$ and so $\tree'_1 \trleq \tree_2$.
\end{asparaitem}

\section{Proof of Lemma \ref{lem:removenode}}
We show that, when guards are monotonic, the $\removeclass{B}$ and $\removenode$ operations do not affect the redundancy of the guards.
The $\Leftarrow$-direction is trivial. For the other direction,
 assume $T_0 \to_{\sigma_0} T_1
\to_{\sigma_1} \cdots \to_{\sigma_{n-1}} T_n \to_{\sigma_n} T_{n+1} = (D,
\treelabeling)$ where each $\sigma_i = (g_i, \chi_i)$, and $v, T_{n+1} \vDash g$
for some $v \in D$. It suffices to construct trees $T_0', \cdots, T_{n+1}'$ with
$T_0' = T_0$ such that $g$ matches $T_{n+1}'$, and for each $i \in [1..n]$,
$T_i \preceq T_i'$ and $T_i' \to_{\sigma_i'} T_{i+1}'$ where each $\sigma_{i}'$
is \emph{not} a $\removenode$ or a $\removeclass{X}$ operation but may be the
identity operation. We shall construct such a sequence of trees by induction on $i$. The base case follows immediately from the assumptions. Let $i \geq 0$. Suppose $T_i'$ has been constructed such that $T_i \preceq T_i'$. There are two cases. If $\chi_i = \removenode$ or $\chi_i = \removeclass{X}$ then set
$T_{i+1}' := T_i'$ and $\sigma_i'$ is the identity operation; otherwise, thanks to Lemma~\ref{lem:monotonic}, set $\sigma_i' := \sigma_i$ and we can construct $T_{i+1}'$ such that $T_i' \to_{\sigma_i'} T_{i+1}'$ and $T_{i+1} \preceq T_{i+1}'$. Finally, since $T_{n+1} \preceq T_{n+1}'$ and $g$ matches $T_{n+1}$, we have $g$ also matches $T_{n+1}'$ by Lemma~\ref{lm:trleq} as required.

    \section{Proof of Proposition \ref{prop:spds-complexity}}
To show $\EXP$ membership for the bit-toggling problem, we simply
compute an exponential-sized PDS from a given sPDS and apply
the standard polynomial-time algorithm for solving control-state
reachability for PDS (this is the same as emptiness for pushdown automata
over a unary input alphabet $\{a\}$).

To show $\PSPACE$ membership for the bounded bit-toggling problem,
suppose that the input to the problem is $\pds = (\varsetV,\varsetW,\pdsrules)$
and $h \in \N$ with $\varsetV = \{x_1,\ldots,x_n\}$ and
$\varsetW = \{y_1,\ldots,y_m\}$. Each configuration is a pair $(q,w)$ of $q
\in \{0,1\}^n$ and $w$ is a word over the stack alphabet
$\salphabet = \{0,1\}^m$. Since the stack height is always bounded by $h$,
each configuration is of size at most $n + (h \times m)$. Also,
checking whether $\sigma = ((q,a),(q',w))$ is a transition rule of $\pds$ ---
where
$q,q' \in \{0,1\}^n$,
$a \in \{0,1\}^m$, and $w$ is a word of length at most 2 over $\salphabet$
--- can be checked in linear time (and therefore polynomial space) by simply
evaluating
each symbolic rule of the form $(|w|,\varphi)$ with respect to the boolean
assignment $\sigma$. [Evaluating BDDs, boolean formulas, and boolean circuits
can all be done in linear time.] So, we can obtain a nondeterministic
polynomial space algorithm (and therefore membership in $\PSPACE$) by
guessing the symbolic rule that needs to be applied at any given moment.

\section{Proof of Correctness of Lemma \ref{lm:simple}}
We show correctness of the reduction to simple rewrite systems.
To show \textbf{(P1)} and \textbf{(P2)}, it suffices to show the following
four statements:
\begin{description}
\item[(P1a)] For each $k \in \N$, $S$ is $k$-redundant for $\trs$ iff $S'$ is
$k$-redundant for $\trs_1$.
\item[(P1b)] For each $k \in \N$, $S$ is $k$-redundant for $\trs_1$ iff $S'$ is
$k$-redundant for $\trs'$.
\item[(P2a)] $S$ is redundant for $\trs$ iff $S'$ is redundant for $\trs_1$.
\item[(P2a)] $S$ is redundant for $\trs$ iff $S_1$ is redundant for $\trs'$.
\end{description}

We first show \textbf{(P1a)} and \textbf{(P2a)}. To this end,
we say that a tree $\tree = (\treedomain,\treelabeling) \in
\trees(\treelabels')$ is \defn{consistent} if, for each $g \in G$
and each node $v \in \treedomain$ with $\underline{g} \in \treelabeling(v)$,
it is the case that $v,\tree \models g$. We say that $\tree$ is
\defn{maximally consistent} if it is consistent and that, for each
$g \in G$ and each node $v \in \treedomain$, if $v,\tree
\models g$ then $\underline{g} \in \treelabeling(v)$.
Observe that
each tree in $\trees(\treelabels)$ (in particular, the initial tree) is
consistent and that
each of the
\mh{removed ``above''} rules (1--3) preserves consistency. Therefore, for all
trees $\tree \in \trees(\treelabels)$ and $\tree' \in \trees(\treelabels')$,
if $\tree \to_{\trs_1}^* \tree'$ then $\tree'$ is consistent.
In addition,
for a tree $\tree \in \trees(\treelabels')$, we write $\tree \cap \treelabels$ to
denote the the tree that is obtained from $\tree$ by restricting to node labels
from $\treelabels$.
Let $\Theta \subseteq \trs_1$ denote the set of intermediate rules added above.
Since the class $G$ is closed under taking subformulas, for each tree $\tree
\in \trees(\treelabels)$,
it is the case that $\tree \to_{\Theta}^* \tree'$ for some maximally consistent
tree $\tree' \in \trees(\treelabels')$ satisfying $\tree' \cap \treelabels =
\tree$.
Now, by a simple induction on the length of paths, it follows that
that $\tree_1 = (\treedomain_1,\treelabeling_1) \to_\trs^* \tree_2 =
(\treedomain_2,\treelabeling_2)$ iff
$\tree_1 \to_{\trs_1}^* \tree_2'$ for some maximally consistent tree $\tree_2'$
such that $\tree_2' \cap \treelabels = \tree_2$.
Since $S'$ contains precisely all $\underline{g}$ for which $g \in S$,
\textbf{(P1a)} and \textbf{(P2a)} immediately follow.

We now show \textbf{(P1b)} and \textbf{(P2b)}.
Observe that both rules (a) and (b) preserve tree-consistency since
$v,\tree \models \langle d^\ast \rangle g$ iff
at least one of the following
cases holds: (i) $v,\tree \models g$, (ii)
there exists a node $w$ in $\tree$ such that $w,\tree \models
\langle d^\ast \rangle g$ and $w$ can be reached from $v$ by following the
direction $d$ for one step.
As before, by a simple induction on the length of paths, for all consistent trees
$\tree_1,\tree_2 \in
\trees(\treelabels')$, it is the case that $\tree_1 \to_{\trs'}^* \tree_2$ iff
$\tree_1 \to_{\trs_1}^* \tree_2$.
%
This implies implies \textbf{(P1b)} and \textbf{(P2b)}, i.e., the correctness of
our entire construction.

Note, in particular, for all $g \in S$, it is the case that $g$ is
redundant in $\trs$ iff $\underline{g}$ is redundant in $\trs'$.

\section{Proof of Algorithm for Lemma~\ref{lm:reduce-class-adding}}

We argue the correctness of the fixpoint algorithm in the proof of
Lemma~\ref{lm:reduce-class-adding}, which shows that the redundancy (resp.
$k$-redundancy) problem is polynomial time solvable assuming oracle calls to the
class-adding (resp. $k$-class-adding) problem.

\paragraph{Soundness}
The proof is by induction over the number of applications of the fixpoint rules.
Fix some sequence of rule applications reaching a fixpoint and let $\tree_i =
(\treedomain_0, \treelabeling_i)$ be the tree obtained after $i$th application.
We prove by induction that we have a tree $\tree$ such that $\tree_0 \to_\trs^*
\tree = (\treedomain', \treelabeling')$ and $\tree_i \trleq \tree$ and thus our
algorithm is sound.

In the base case, when $i = 0$, $\tree = \tree_0$ and the proof is trivial.
Hence, assume by induction we have a tree $\tree$ such that $\tree_0
\to_{\trs,f}^* \tree$ and $\tree_i \trleq \tree$.  There are two cases.
\begin{itemize}
    \item
        When we apply a rule $\sigma = (g,\addclass{B})$ at a node $v \in
        \treedomain_0$, then $\tree_i \to_\sigma \tree_{i+1}$ and by
        Lemma~\ref{lem:monotonic} (Monotonicity) we obtain $\tree'$ satisfying
        the induction hypothesis.

    \item
        When the class-adding problem has a positive answer on input $\langle
        \tree_v,f,c,\trs\rangle$ and we update $\treelabeling(v) :=
        \treelabeling(v) \cup \{c\}$, then there is a sequence of rule
        applications $\sigma_1,\ldots,\sigma_n$ witnessing the positive solution
        to the class-adding problems.  By $n$ applications of
        Lemma~\ref{lem:monotonic} we easily obtain $\tree'$ satisfying the
        induction hypothesis as required.
\end{itemize}

\paragraph{Completeness}
Take any sequence of rules $\sigma_1, \ldots, \sigma_n$ such that
\[
    \tree_0 \to_{\sigma_1} \tree_1 \to_{\sigma_2} \cdots \to_{\sigma_n} \tree_n
\]
where for each $i$, $\tree_i = (\treedomain_i, \treelabeling_i)$.  Note
$\treedomain_0 \subseteq \treedomain_i$ for each $i$.  Let $\tree_F =
(\treedomain_0, \treelabeling)$ be the tree obtained as the conclusion of the
fixpoint calculation.  We show $\treelabeling_i(v) \subseteq \treelabeling(v)$
for all $v \in \treedomain_0$ and hence our algorithm is complete.

We induct over $i$.  In the base case $i = 0$ and the result is trivial.  For
the induction, if $\sigma_i$ was applied to $v \notin \treedomain_0$ or
$\sigma_i$ is an $\addchild{B}$ rule, the result is also trivial.  Else
$\sigma_i$ is of the form $\sigma = (g, \addclass{B})$ and applied to $v \in
\treedomain_0$.  There are several cases.
\begin{itemize}
    \item
        When $g = A$ let $v' = v$ and when $g = \langle\uparrow\rangle A$ let
        $v'$ be the parent of $v$.  Since $v \in \treedomain_0$ we also have $v'
        \in \treedomain_0$ and thus $\treelabeling_i(v') \subseteq
        \treelabeling(v)$ and thus $\sigma$ is applicable at $v'$ and the first
        fixpoint rule applies.  Since $\tree_F$ is a fixpoint, we necessarily
        have $B \subseteq \treelabeling(v)$.

    \item
        When $g = \langle\downarrow\rangle A$ there are two cases.
        \begin{itemize}
            \item
                If the guard was satisfied by matching with $v' \in
                \treedomain_0$, then we have $B \subseteq \treelabeling(v)$
                exactly as above.

            \item
                Otherwise, the guard was satisfied in the run to $\tree'$ by
                matching with $v' \notin \treedomain_0$.  Thus we can pick out
                the sub-sequence of $\sigma_1, \ldots, \sigma_n$ beginning with
                the $\addchild{C}$ rule (applied at $v$) that created $v'$ and
                all rules applied at $v'$ or a descendant.  By
                Lemma~\ref{lem:monotonic} (Monotonicity) and the fact that all
                guards are simple, this sequence is a witness to the positive
                solution of the class-adding problem $\langle \tree_v, f, c,
                \trs \rangle$ for all $c \in B$ and thus $B \subseteq
                \treelabeling(v)$ as required.
\end{itemize}
\end{itemize}

\section{BDD representation in the Proof of Lemma \ref{lm:reduce}}
We first recall the definition of BDDs. A \defn{binary decision diagram (BDD)} is
a symbolic representation of boolean functions. A \defn{BDD over the variables
$\varsetV = \{x_1,\ldots,x_n\}$} is a rooted directed acyclic graph
$G = (V,E)$ together with a mapping $\treelabeling: (V \cup E) \to
(\varsetV \cup \{0,1\})$ such that:
\begin{enumerate}[(i)]
    \item
        if $v \in V$ is an internal node, then $\treelabeling(v) \in \varsetV$
        and it has \emph{precisely two} outgoing edges $e_1, e_2$ labeled by
        $\treelabeling(e_1) = 0$ and $\treelabeling(e_2) = 1$, respectively,
        and
    \item(ii)
        if $v \in V$ is a leaf node, then it is labeled by $\treelabeling(v)
        \in \{0,1\}$.
\end{enumerate}
Given an assignment $\nu: \varsetV \to \{0,1\}$, we can determine the truth
value of $G$ under $\nu$ as follows: start at the root node of $G$; if
the current node is $v$ labeled by $x \in \varsetV$, follow the outgoing
$\nu(x)$-labeled from $v$ to determine the next node; and if
the current node is a leaf, then the truth value is given by the label of
this leaf.

An \defn{ordered BDD (oBDD)} is a BDD together with an ordering of
$\varsetV$ (say $x_1 < \cdots < x_n$) such that if $v$ is a predecessor
of $u$ in $G$ where both $v$ and $u$ are internal nodes, then
$\treelabeling(v) < \treelabeling(u)$. So, no variable $x\in\varsetV$
occurs more than once in each path from the root to a leaf in an oBDD.
Finally, an oBDD is reduced if (i) it has no two isomorphic subgraphs,
and (ii) it has no internal node whose two children are the same node.
In the
sequel, we use the term BDD to mean reduced oBDD.

We now demonstrate how to represent boolean formulas in the reduction as BDDs.
Observe that each boolean formula in our $\pds$-rules is a conjunction of
formulas, in which \emph{no two conjuncts share any common variable}. This
allows an easy conversion to BDD representation. For example, the
BDD $G = (V,E,\treelabeling)$ for the formula $\varphi$ in the first type of
rules in our sPDS construction in Section \ref{sec:upper} can be defined as
follows. Each conjunct of the form $s \lra t$ is turned
into the following BDD
\begin{center}
    \epsfig{file=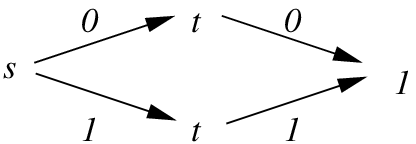,width=3.5cm}
\end{center}
where a missing edge means that it goes to a 0-labeled leaf (also
omitted from diagram). Similarly,
a conjunct of the form $s$ is turned into the following BDD
\begin{center}
    \epsfig{file=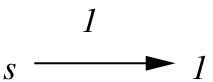,width=2cm}
\end{center}
where, again, a missing edge means that it goes to a 0-labeled leaf. Now,
if $G_1$ and $G_2$ are BDDs for the conjuncts $\varphi_1$ and $\varphi_2$
with different sets of variables, then the BDD for $\varphi_1 \wedge
\varphi_2$ can be constructed
by gluing the 1-labeled leaf in $G_1$ by the root of $G_2$ (using the label
of the root of $G_2$). Since no two conjuncts in $\varphi$ share a common
variable, the BDD for $\varphi$ can be obtained by linking these simpler BDDs
in this way. The resulting BDD is of size linear in the size of $\varphi$.

\section{Proof of Lemma \ref{lm:spds-correctness}}
We first prove the direction
(A2)$\Rightarrow$(A1). We will prove this by induction on the length of
paths. The problem is that our induction hypothesis is not strong enough
to get us off the ground (since the state component of
$I_{\treelabeling_0(\epsilon),f}$ is always $(0,\ldots,0)$). Therefore,
we will first prove this by induction with a stronger induction
hypothesis:
\begin{lemma}
Given $\pds$-configurations $(q,\alpha)$ and $(q',\alpha')$ with
$q,q': \varsetV \to \{0,1\}$ and $\alpha,\alpha': \varsetW \to
\{0,1\}$ (i.e. stack of height 1), suppose that $(q,\alpha) \to_\pds^*
(q',\alpha')$.
Define $\tree =
(\treedomain,\treelabeling) \in \trees(\treelabels)$ as
(1) $D = \{\epsilon\} \cup \{ 0 : q(\Pop) = 1 \}$,
(2) $\treelabeling(\epsilon) =
\{ c \in \CLASS : \alpha(y_c) = 1 \}$ and
$\treelabeling(0) = \{ c \in \CLASS : q(x_c) = 1 \}$.
Define the
assumption function $f_\alpha: (\CLASS\cup\{\Root\}) \to \{0,1\}$ with
$f_\alpha(c) = \alpha(z_c)$ for $c \in \CLASS$ and $f_\alpha(\Root) = \alpha(\Root)$.
Then, there exists $\tree' = (\treedomain',\treelabeling')\in\trees(
\treelabels)$ such that
$\treelabeling'(\epsilon) =
\{ c \in \CLASS : \alpha'(y_c) = 1\}$
and $\tree \to_{\trs,f_\alpha}^* \tree'$.
\label{lm:sound}
\end{lemma}
Observe that the direction (A2)$\Rightarrow$(A1) is immediate from this
Lemma: simply apply it with $(q,\alpha) = I_{\treelabeling_0(\epsilon),f}$ and
$(q',\alpha') = I_{X,f,q'}$.
\begin{proof}
The proof is by induction on the length $k$ of the witnessing path $\Path:
(q,\alpha) \to^* (q',\alpha')$. When $k = 0$, we have that $(q,\alpha) =
(q',\alpha')$. So, we may set $\tree' = \tree$ and the statement is satisfied.

Let us now consider the case when $k > 0$. There are two cases to consider
depending on the second configuration in the path $\Path$:
\begin{enumerate}[(i)]
    \item $\Path = (q,\alpha) \to (q_1,\alpha_1) \to^* (q',\alpha')$,
       for some $q_1: \varsetV \to \{0,1\}$ and $\alpha_1: \varsetW \to \{0,1\}$.
    \item $\Path = (q,\alpha) \to (q_1,\alpha_1\beta_1) \to^* (q',\alpha')$, where
        $q_1: \varsetV \to \{0,1\}$ and $\alpha_1,\beta_1: \varsetW \to \{0,1\}$.
\end{enumerate}
Note that the case when the first action is a pop (i.e. $(q,\alpha) \to
(q_1,\epsilon)$) is not possible since $(q_1,\epsilon)$ is a deadend.
\OMIT{
Suppose that $\Path = (q,\alpha) \to
(q_1,\alpha_1) \to^* (q',\alpha')$,
where $q_1: \varsetV \to \{0,1\}$ and $\alpha_1: \varsetW \to \{0,1\}$.
}

Let us consider Case (i).
The configuration $(q_1,\alpha_1)$ is obtained by applying a
rule $\sigma = (1,\varphi) \in \pdsrules$, which was generated by a rule
$\gamma = (g,\addclass{B}) \in \trs$. Suppose that $g = A \subseteq \CLASS$;
the other two cases (i.e. when $g$ is of the form $(\langle d\rangle
A,\addclass{B})$) can be handled in the same way.
Since $\varphi$ has a conjunct of the form $y_a$ for each $a \in A$,
it is the case that $A \subseteq \treelabeling(\epsilon)$.
Since $y_b^1$ is a conjunct of $\varphi$ for each $b \in B$,
The rule $\sigma$ toggles on the variables $y_b$ for each
$b \in B$ while preserving the value of all other variables
in $\varsetV \cup \varsetW$. So, if $T_1 := (\treedomain,\treelabeling_1)$
with $\treelabeling_1(0) := \treelabeling(0)$ and
$\treelabeling_1(\epsilon) := \treelabeling(\epsilon) \cup B$, then
$T_0 \to_{\trs,f_{\alpha}} T_1$. We may now apply the induction hypothesis on
the subpath $(q_1,\alpha_1) \to^* (q',\alpha')$ of $\Path$ of shorter
length and obtain $T_1 \to_{\trs,f_{\alpha}} T'$. Concatenating paths,
we obtain a path $T \to_{\trs,f_{\alpha}} T'$ as desired.

Let us now consider Case (ii).
\OMIT{
The last case to consider is when $\Path = (q,\alpha) \to
(q_1,\alpha_1\beta_1) \to^* (q',\alpha')$, where
$q_1: \varsetV \to \{0,1\}$ and $\alpha_1,\beta_1: \varsetW \to \{0,1\}$.
}
The configuration $(q_1,\alpha_1\beta_1)$ is obtained by applying a rule
$\sigma = (2,\varphi) \in \pdsrules$, which was generated by a
rewrite rule $(A,\addchild{B}) \in \trs$. Therefore, it must be the case that
$(q_1,\beta_1) \to_\pds^* (q_2,\beta_2) \to_{\sigma_2} (q_3,\epsilon)$ and
$(q_3,\alpha_1) \to_\pds^* (q',\alpha')$
for some configurations $(q_2,\beta_2)$ and $(q_3,\epsilon)$ with stack
height 1 and 0, respectively, and some $\pds$-rule $(0,\varphi')$. As in the
previous case, we have
$A \subseteq \treelabeling(\epsilon)$ and so, applying $(A,\addchild{B})$
at $\epsilon$, we obtain $\tree_1 = (\treedomain_1,\treelabeling_1)$ with
$\treedomain_1 := \treedomain \cup \{1\}$ and $\treelabeling_1$ is an
extension of $\treelabeling$ to $\treedomain_1$ with $\treelabeling_1(1) :=
B$. Let $\tree_1' = (\{\epsilon\},\treelabeling_1')$ be the subtree of
$\tree_1$ rooted at node 1. Applying induction on the path
$(q_1,\beta_1) \to_\pds^* (q_2,\beta_2)$, we obtain a tree
$\tree_2' = (\treedomain_2',\treelabeling_2')$ such that:
(1) $\nu: \tree_1' \to_{\trs,f}^* \tree_2'$ with assumption function
$f := f_{\beta_1}$, and (2) $\treelabeling_2'(\epsilon) = \{ c \in
\CLASS : \beta_2(y_c) = 1 \}$. Similarly, if we let $\tree_3 =
(\{\epsilon,0\},\treelabeling_3)$ be the tree with $\treelabeling_3(\epsilon)
= \treelabeling(\epsilon)$ and $\treelabeling_3(0) =
\treelabeling_2'(\epsilon)$, we apply induction hypothesis on
the path $(q_3,\alpha_1) \to_\pds^* (q',\alpha')$ and obtain a path
$\mu: \tree_3 \to_{\trs,f_\alpha}^* \tree''$ (here $f_\alpha = f_{\alpha_1}$) for
some
tree $\tree'' = (\treedomain'',\treelabeling'')$ with $\treelabeling''(\epsilon)
= \{ c \in \CLASS : \alpha'(y_c) = 1\}$. Now let $\tree'$ be the tree
obtained from $\tree''$ by: (1) removing the node 0 of $\tree''$ and replacing
it by the tree $\tree_2'$ (note: $\tree''(0) = \tree_2'(\epsilon)$), and
(2) adding a node labeled by $\treelabeling(0)$ as a
child of the root node of $\tree''$. We have $\tree'(\epsilon) =
\tree''(\epsilon) = \{c \in \CLASS : \alpha'(y_c) = 1 \}$. Finally,
we see that $\tree \to_{\trs,f_\alpha} \tree_1 \to_{\trs,f_\alpha}^* \tree'$,
where the path $\tree_1 \to_{\trs,f_\alpha}^* \tree'$ is obtained by
first applying the path $\nu$ on the subtree of $\tree_1$ rooted at 1 and
then applying the path $\mu$ on the resulting tree (in the second step, we
keep the subtree of $\tree_1$ rooted at the node 0 unchanged).
\end{proof}

We now prove the direction (A1)$\Rightarrow$(A2) of the above claim. The
proof is by induction on the length $k$ of the path $\Path: \tree_0
\to_{\trs,f}^* \tree$. The base case is when $k = 0$, in which case
we set $q'(x) := 0$ for all $x \in \varsetV$ and (A2) is
satisfied since $I_{\treelabeling_0(\epsilon),f} = I_{X,f,q'}$.

We proceed to the induction case, i.e., when $k > 0$.
We may assume
that $\treelabeling(\epsilon)$ is a \emph{strict} superset of
$\treelabeling_0(\epsilon)$; otherwise, we may replace $\Path$ with an
empty path, which reduces us to the base case. Therefore,
some rewrite rule $(g,\addclass{B})$ with $B \cap
\treelabeling_0(\epsilon) \neq \emptyset$
is applied at $\epsilon$ at some point in $\Path$. We have two cases
depending on the first application of such a rewrite rule $\sigma =
(g,\addclass{B})$ at $\epsilon$ in $\Path$:

\textbf{(Case 1)} We have $g = A$ or $g = \langle \uparrow \rangle A$ for some
$A \subseteq \CLASS$. Since $\treelabeling_0(\epsilon) =
\treelabeling'(\epsilon)$, the rule $\sigma$ can already be applied at $\epsilon$
in $\tree_0$ yielding a tree $\tree_1 = (\{\epsilon\},\treelabeling_1)$
with $\treelabeling_1(\epsilon) = \treelabeling_0(\epsilon) \cup B$. Owing to
Lemma \ref{lem:monotonic}, we may apply the same sequence of rewrite rules
witnessing $\Path$ but with $\tree_1$ as an initial configuration. Furthermore,
we may remove any application of $\sigma$ at $\epsilon$ in this
sequence (since $B$ is always a subset of the label of $\epsilon$), which
gives a path $\Path': \tree_1 \to_{\trs,f} \tree$ of length $\leq k-1$.
Altogether, we obtain another path $\tree_0 \to_{\sigma,f} \tree_1
\to_{\trs,f}^* \tree$ of length $\leq k$.
In addition, we have $I_{\treelabeling_0(\epsilon),f} \to_\pds
I_{\treelabeling_1(\epsilon),f}$ by applying the rule
$(1,\varphi) \in \pdsrules$ generated from $\sigma$ in the construction of
$\pds$. Now we may apply the induction hypothesis to the path
$\Path': \tree_1 \to_{\trs,f}^* \tree$ (of length $< k$) and obtain
$q': \varsetV
\to \{0,1\}$ such that $I_{\treelabeling_1(\epsilon),f} \to_\pds^*
I_{\treelabeling(\epsilon),f,q'}$, which gives us a path from
$I_{\treelabeling_0(\epsilon),f}$ to $I_{\treelabeling(\epsilon),f,q'}$,
as desired.

\textbf{(Case 2)} We have $g = \langle \downarrow \rangle A$ for some
$A \subseteq \CLASS$. In this case, the first rule applied in $\Path$ must be
of the form $\gamma = (C,\addchild{Y})$ for some $C,Y \subseteq \CLASS$.
That is, we have $\Path: \tree_0
\to_{\gamma,f} \tree_2 \to_{\trs,f}^* \tree_3 \to_{\sigma,f} \tree_4
\to_{\trs,f}^* \tree$ for some
trees $\tree_i = (\treedomain_i,\treelabeling_i)$ ($i \in \{2,3,4\}$) such that:
(1) $\treelabeling_0(\epsilon) = \treelabeling_i(\epsilon)$ for all
$i \in \{2,3\}$, (2)
$\treedomain_2 = \{\epsilon,0\}$, (3) At $\tree_3$, the rule $\sigma$ is applied at
$\epsilon$, and so $\treedomain_4 = \treedomain_3$ and
$\treelabeling_4(v) = \treelabeling_3(v)$ for $v \neq \epsilon$ and
$\treelabeling_4(\epsilon) = \treelabeling_3(\epsilon) \cup B\supset
\treelabeling_3(\epsilon)$. So, this means that some child $i \in \N$ of the
root node in $\tree_3$ satisfies $\treelabeling_3(i) \supseteq A$, which enables
$\sigma$ to be executed at the root node.
Now recall that simple guards allow a node to inspect only its immediate neighbours
(i.e. parent or children). For this reason, we may assume without a loss
of generality that $i = 0$, i.e., that it was the child of $\epsilon$ that was
spawned by the first transition $\gamma$. In fact, we may go one step further to
assume that 0 is the only child of the root node in the subpath $\tree_2
\to_{\trs,f}^* \tree_3$! This is because that without modifying the label of
$\epsilon$, simple guards do not allow the node 0 to ``detect'' the presence of the
other children.

To finish off the proof, we will apply induction hypotheses twice on two different
paths. Firstly, we apply the induction hypothesis on the shorter
path $(\tree_2)_{|0} \to_{\trs,f_\epsilon}^* (\tree_3)_{|0}$, where
$f_\epsilon$ is an assumption function that captures the node label
$\treelabeling_3(\epsilon)$, i.e.,
$f_{\epsilon}(\Root) = 0$ and $f_{\epsilon}(c) = 1$ iff $c \in
\treelabeling_3(\epsilon)$; note that $\treelabeling_0(\epsilon) =
\treelabeling_2(\epsilon) = \treelabeling_3(\epsilon)$. [Recall that
$(\tree_2)_{|0}$ means the subtree of $\tree_2$ rooted at the node 0.] This
gives us the path $I_{\treelabeling_2(0),f_\epsilon} \to_\pds^*
I_{\treelabeling_3(0),f_\epsilon,q''}$ for some $q'': \mathcal{V} \to \{0,1\}$.
Therefore, we have the path $\nu_1: I_{\treelabeling_0(\epsilon),f} \to_\pds
(\bar 0,\mu_\epsilon\mu_0) \to_\pds^* (q'',\mu_\epsilon\mu_0') \to_\pds
(q',\mu_\epsilon) \to (q',\mu_\epsilon')$, where
\begin{itemize}
    \item $\mu_\epsilon(r) = 1$ iff $I_{\treelabeling_0(\epsilon),f}(r) = 1$
        for each $r \in \varsetW = \{ y_c, z_c : c \in \CLASS \} \cup \{\Root\}$.
    \item $\mu_0(\Root) = 0$; $\mu_0(z_c) := \mu_\epsilon(y_c)$;
        and $\mu_0(y_c)= 1$ iff $c \in B$, for each $c \in \CLASS$.
    \item $\mu_0'(r) = I_{\treelabeling_3(0),f_\epsilon,q''}(r)$ for each
        $r \in \varsetW$.
    \item $q'(\Pop) = 1$ and $q'(x_c) = \mu_0'(y_c)$ for each
        $c \in \CLASS$.
    \item $\mu_\epsilon'(r) = 1$ iff $\mu_\epsilon(r) = 1$ or $r \in B$.
\end{itemize}
Note that $(q',\mu_\epsilon')$ is the same as $I_{\treelabeling_4(\epsilon),f,q'}$.
We are now going to apply the induction hypothesis the second time. To this end,
we let $\tree_4' = (\{\epsilon\},\treelabeling_4')$ be the single-node tree
defined by restricting $\tree_4$ to the root, i.e., $\treelabeling_4'(\epsilon)
= \treelabeling_4(\epsilon) \supset \treelabeling_0(\epsilon)$. Since $\tree_0
\trleq \tree_4$, by the proof of Lemma \ref{lem:monotonic}, we have
$\tree_4' \to_{\trs,f}^* \tree$ by following the actions in the path $\Path$, except
for the action $\sigma$. This gives us a path $\Path'$ whose length is
$|\Path| - 1$. So, by induction, we have a path $\nu_2:
I_{\treelabeling_4(\epsilon),f}
\to_\pds^* I_{\treelabeling(\epsilon),f,p}$ for some $p: \varsetV \to \{0,1\}$.
In order to connect $\nu_1$ and $\nu_2$, we need to use the following lemma, which
completes the proof of Lemma \ref{lm:spds-correctness}.

\begin{lemma}
Suppose $p,q: \varsetV \to \{0,1\}$ such that $p(x) \leq q(x)$ for each
$x \in \varsetV$. Then, for all stack content $v = \alpha_1\cdots \alpha_m$
where each $\alpha_i: \varsetW \to \{0,1\}$, if $(p,v) \to_{\sigma} (p',w)$ with
$\sigma \in \pdsrules$, then $(q,v) \to_{\sigma} (q',w)$ for some $q':
\varsetV \to \{0,1\}$ such that $p'(x) \leq q'(x)$ for each $x \in \varsetV$.
\label{lm:mon-sPDS}
\end{lemma}
\begin{proof}
Easy to verify by checking each $\pds$-rule.
\end{proof}

\OMIT{
Therefore, we may assume that $g = \langle \downarrow \rangle A$
for some $A \subseteq \CLASS$. In this case,
}

\section{A polynomial-time fixed point algorithm for bounded height}
\label{sec:fixed-point}
We now provide a simple fixpoint algorithm for the $k$-redundancy problem that runs
in time
$n^{O(k)}$. We assume that we have simplified our rewrite systems using the
simplification given in Section \ref{sec:upper}.
The input to the algorithm is a simple rewrite system $\trs$ over
$\trees(\treelabels)$ with $\treelabels = 2^\CLASS$, a guard database
$S \subseteq \CLASS$, and an initial $\treelabels$-labeled tree $\tree_0 =
(\treedomain_0,\treelabeling_0) \in \trees(\treelabels)$. The number $k \in \N$
is fixed (i.e. not part of the input). Our algorithm will
compute a tree $\tree_F = (\treedomain_F,\treelabeling_F)$ such that
$\tree_F \in \post_\trs^*(\tree_0)$ and, for \emph{all} trees $\tree \in
\post_\trs^*(\tree_0)$, we have $\tree \trleq \tree_F$. By Lemma \ref{lm:trleq},
if a guard $g \in S$ is matched in some $\tree \in \post_\trs^*(\tree_0)$,
then $g$ is matched in $\tree_F$. The converse also holds since
$\tree_F \in \post_\trs^*(\tree_0)$. Furthermore, we shall see that
$\tree_F$ is of height at most $k$, each of whose nodes has at most
$\sizeof{\tree_0}+|\trs|$ children (and so of size at most
$O(\sizeof{\tree_0}+|\trs|)^{O(k)}$), and is computed in time
$(\sizeof{\tree_0}+|\trs|)^{O(k)}$. Computing the rules
in $S$ that are matched in $\tree_F$ can be easily done in time linear in
$\sizeof{S} \times \sizeof{\tree_F}$.

\newcommand\dom{\mathit{dom}}
\newcommand\calF{{\cal F}}
\newcommand\makeset[1]{\{#1\}}

The algorithm for computing $\tree_F$ works as follows.
\begin{description}
\item[(1)] Set $\tree = (\treedomain,\treelabeling) := \tree_0$;
\item[(2)] Set $W := \emptyset$;
\item[(3)] Add each pair $(v,\sigma)$ to $W$, where $v \in \treedomain$ and
    $\sigma \in \trs$ is a rule that is applicable at $v$,
    i.e.~$v, T \models g$ where $\sigma = (g, \chi)$ for some $\chi$; all
    pairs in $W$ are initially ``unmarked'';
\item[(4)] Repeat the following for each unmarked pair $(v,\sigma) \in W$
    with $\sigma = (g,\chi)$:
    \begin{description}
    \item[(a)] if $\chi = \addclass{X}$, then $\treelabeling(v) :=
        \treelabeling(v) \cup X$;
    \item[(b)] if $\chi = \addchild{X}$, $|v| \leq k - 1$, and
        $X \not\subseteq \treelabeling(v.i)$ for \emph{each} $v.i \in
        \treedomain$ child of $v$,
        then
        set $\treedomain := \treedomain \cup \{v.(i+1)\}$  where
        $i := \max\{ j \in \N : v.j \in \treedomain\} + 1$ (by convention,
        $\max\emptyset = 0$), and $\treelabeling(v.i) := X$;
    \item[(c)] Mark $(v,\sigma)$;
    \item[(d)] For each $(w,\gamma) \in (\treedomain \times
        \trs) \setminus W$ that is applicable, add $(w,\gamma)$ to
        $W$ unmarked;
    \end{description}
\item[(5)] Output $\tree_F := \tree$;
\end{description}

\noindent
\textbf{Proof of Correctness.} We will show that $\tree \trleq \tree_F$
for each $\tree \in \post_\trs^*(\tree_0)$. The proof is by induction on
the length $m$ of path from $\tree_0$ to $\tree$. We write $\tree_F =
(\treedomain_F,\treelabeling_F)$. The base case is when $m=0$,
in which case $\tree = \tree_0 = (\treedomain_0,\treelabeling_0)$. Since
our rewrite rules only add classes to nodes and add new nodes, it is
immediate that $\treedomain_0 \subseteq \treedomain_F$ and that
$\treelabeling_0(v) \subseteq \treelabeling_F(v)$ for each $v \in
\treedomain_0$. That is, $\tree_0 \trleq \tree_F$. We now proceed to the
inductive step and assume that $\tree = (\treedomain,\treelabeling) \trleq
\tree_F$ with a witnessing embedding $f: \treedomain \to \treedomain_F$.
Let $v \in \treedomain$ and $\sigma = (g,\chi)$ be a $\trs$-rule that is
applicable on the node $v$ of $\tree$. By Lemma \ref{lm:trleq}, it follows that
$f(v), \tree_F \models g$. Let us suppose that after applying $\sigma$ on
$v$, we obtain $\tree' := (\treedomain',\treelabeling')$.
We have two cases now:
\begin{enumerate}
\item $\chi = \addclass{X}$. In this case, $\treedomain' = \treedomain$,
$\treelabeling'(u) = \treelabeling(u)$ whenever $u \neq v$, and
$\treelabeling'(v) = \treelabeling(v) \cup X$. Since $\tree \trleq
\tree_F$, it follows that $\treelabeling(v) \subseteq \treelabeling_F(f(v))$.
Since $f(v),\tree_F \models g$, it follows that $(f(v),(g,\addclass{X}))$ is marked
in $W$ (by the end of the computation of $\tree_F$).
This means that $X\subseteq \treelabeling_F(f(v))$. Altogether,
$\tree' \trleq \tree_F$.
\item $\chi = \addchild{X}$ and $|v| \leq k-1$. In this case, $\treedomain' =
\treedomain
\cup \{ v.i \}$ for some $i \in \N$ with $v.i \notin \treedomain$,
$\treelabeling'(u) = \treelabeling(u)$ for each $u \in \treedomain$,
and $\treelabeling'(v.i) = X$. Since $\tree \trleq \tree_F$, the properties
(H1) and (H3) of embeddings imply that
$v$ and $f(v)$ have the same levels in $\tree$ and $\tree_F$ respectively.
Also, since $v,\tree_F \models g$, it follows that $(f(v),(g,\addchild{X}))$ is
marked in $W$ (by the end of the computation of $\tree_F$). This means that
there exists a child $f(v).j$ in $\treedomain_F$ with $X \subseteq
\treelabeling_F(f(v).j)$. 
Hence, the extension of $f$ with $f(v.i) := f(v).j$
is a witness for $\tree' \trleq \tree_F$, as desired.
\end{enumerate}

\noindent
\textbf{Running time analysis.}
We first analyse the time complexity of each individual step in the
above algorithm. Checking whether a simple guard $g$
is satisfied at a node $v$ can be done by inspecting the node labels of
$v$ and its neighbours (i.e. parent and children), which can be performed in time
$O(|\CLASS| \times
(\text{number of neighbours of $v$}))$. So, since each node in $\tree_0$ has
most $|\tree_0|$ children, then step (3) can be achieved in time
$O(|\CLASS| \times |\tree_0|^2)$. Similarly,
modifying a node label can be done in time $O(|\CLASS|)$, which is the
running time of Step (4a).  Since no rewrite rule in $\trs$ can be
applied twice at a node, each node in $\tree_F$ has
at most $m := |\tree_0| + |\trs|$ children, i.e., either that child was already
in $\tree_0$ or that it was generated by a rewrite rule in $\trs$.
Therefore, Step (4b) takes time $O(|\CLASS| \times m)$.
Since Step (4d) needs to check only $(w,\gamma)$ where $w$ is a neighbor
(i.e. children or parent) of $v$, then it can be done in time
$O(|\CLASS| \times m)$.
Since the tree $\tree_F$ has height at most $k$ and each node has at most
$m$ children, the tree $\tree_F$ has at most $O(m^k)$ nodes. So, the
entire Step (4) runs in time $O(|\CLASS| \times |\tree_0|^{k+1} \times
|\trs|^{k+1})$, which is also a time bound for the running time of the
entire algorithm.

    \section{$\W[1]$-hardness for $k$-redundancy problem}
Before proving our $\W[1]$-hardness, we first briefly review some standard
concepts from parameterised complexity theory \cite{grohe-book}.
A \defn{parameterized
problem} $\mathcal{C}$ is a computational problem with input of the form
$\langle v;k\rangle$,
where $v \in \{0,1\}^*$ and $k \in \N$ (represented in unary) is called the
\emph{parameter}.
An \defn{fpt-algorithm} for $\mathcal{C}$ runs in time $O(f(k) . n^c)$ for
some constant $c$, and some computable function
$f: \N \to \N$ (both independent of input). Here, $n$ is the size of the input.
So, an algorithm that
runs in time $O(n^k)$ is not an fpt-algorithm. The problem $\mathcal{C}$
is said to be \defn{fixed-parameter tractable} if it can be solved by an
fpt-algorithm. There are classes of problems in parameterized complexity
that are widely believed not to be solvable by fpt-algorithms, e.g.,
$\W[1]$ (and others in the $\W$-hierarchy). To show that a problem
$\mathcal{C}$ is $\W[1]$-hard, it suffices to give an fpt-reduction
to $\mathcal{C}$ from the \defn{short acceptance problem of nondeterministic
Turing machines (NTM)}: given a tuple $\langle (\TM,w);k \rangle$,
where $\TM$ is an NTM, $w \in \{0,1\}^*$ is an input word, and $k$ is a positive
integer represented in unary, decide whether $\TM$ accepts $w$ in $k$
steps. Among others, standard reductions in complexity theory that satisfy the
following two conditions are fpt-reductions: (1) fpt-algorithms, (2) if
$k$ is the input parameter, then the output parameter is $O(k)$.
\OMIT{
The set of $k$-fold exponential functions is defined inductively.
A $0$-fold exponential function is any polynomial function $f: \N \to \N$; for
each $k \in \N_{> 0}$, a $k$-fold exponential function is any function of the
form $2^{f(n)}$, where $f$ is any $(k-1)$-fold exponential function.
We use $k$-$\EXPSPACE$ to mean the set of problems solvable by Turing
machines equipped with $k$-fold exponential bounded space.
}

We now show that $k$-redundancy is $\W[1]$-hard. To this end, we reduce from the
short acceptance problem of nondeterministic Turing machines.
The input to the problem is $\langle \TM,w;k\rangle$, where $\TM$ is an NTM,
$w$ is an input word for $\TM$, and $k \in \N$ given in unary.
Suppose that $\TM =
(\oalphabet,\talphabet,\blank,\controls,\pdsrules,q_0,q_{acc})$,
where $\oalphabet = \{0,1\}$ is the input alphabet, $\talphabet =
\{0,1,\blank\}$ is the tape alphabet, $\blank$ is the blank symbol, $\controls$
is the set of control states, $\pdsrules \subseteq (\controls \times \talphabet)
\times (\controls \times \talphabet \times \{L,R\})$ is the transition relation
($L$/$R$ signify go to left/right), $q_0 \in \controls$ is the initial
state, and $q_{acc}$ is the accepting state. Each configuration of $\TM$ on $w$
of length $k$ can be represented as a word
\[
    (0,a_0,c_0)(1,a_1,c_1)\cdots (k,a_k,c_k)
\]
where $a_i \in \talphabet$ for each $i \in [0,k]$, and for some $j \in [0,k]$
we have $c_j \in \controls$ and for all $r \neq j$ it is the case that
$c_j = \star$. Note that this is a word (of a certain kind) over $\Omega := \Omega_k := [0,k] \times \talphabet
\times (\controls \cup \{\star\})$. In the sequel, we denote by
$\CONF_k$ the set of all configurations of $\TM$ of length $k$.
Suppose that $w = a_1\cdots a_n$. If
$n < k$, we append some blank symbols at the end of $w$ and assume that
$|w| \geq k$. The initial configuration will be
$I_0(w) := (0,\blank,q_0)(1,a_1,\star)\cdots (k,a_k,\star)$\mh{Why start with a blank square, especially when you add another blank cell below.}. [Notice that
if $n > k$, then some part of the inputs will be irrelevant.]
In the following, we will append the symbol $(-1,\blank,\star)$ (resp.
$(k+1,\blank,\star)$) before the start (resp. after the end) of $\TM$
configurations to aid our definitions. Therefore, define
$\Omega_e := \Omega_{k,e} := [-1,k+1] \times \talphabet
\times (\controls \cup \{\star\})$.


We now construct a tree-rewrite system $\pds$ to simulate $\TM$.
First off,
for each $\pdsrules$-rule $\sigma$,
define a relation $R_\sigma \subseteq \Omega_e^4$ such that
$(\vecV_1,\vecV_2,\vecV_3, \vecW) \in R$ describes the update by
rule $\sigma$ of cell $\vecV_2$ to $\vecW$ given the neighbouring cells
$\vecV_1$ and $\vecV_3$.  That is, $(\vecV_1,\vecV_2,\vecV_3, \vecW) \in R$
iff a configuration $C_2 = \vecU_{-1} \cdots
\vecU_{k+1}$ of $\TM$, where $\vecU_{i+1} = \vecW$ for some $i$, can be
reached in one step using $\sigma$ from another
configuration
$C_1 = \vecU_{-1}' \cdots \vecU_{k+1}'$, where $\vecU_i'\vecU_{i+1}'\vecU_{i+2}' =
\vecV_1\vecV_2\vecV_3$. Let $R = \bigcup_{\sigma\in\pdsrules} R_\sigma$.
In this case, we say that rule $\sigma$ is a witness
for the tuple $(\vecV_1,\vecV_2,\vecV_3,\vecW)$. Notice that, given a rule
$\sigma$ of $\pds$ and $\vecV_1,\vecV_2,\vecV_3$, there exists
\emph{at most one} $\vecW$ such that $(\vecV_1,\vecV_2,\vecV_3,\vecW)
\in R$ and is witnessed by $\sigma$. In addition,
it is easy to see that checking whether
$(\vecV_1,\vecV_2,\vecV_3,\vecW) \in R$ can be done in polynomial time
(independent of $k$) since this is a simple local check of $\TM$
rules.

The rewrite system $\pds$ works over trees of labels $\Xi = \Omega_e
\cup \pdsrules$.
The initial tree is a tree with a single node labeled by
$I_0(w)$. We add the following rules to $\pds$:
\begin{enumerate}
\item For each $\sigma \in \pdsrules$, we add $(\top,\addchild{\sigma})$
    to $\pds$.
\item For each $(\vecV_{i-1},\vecV_i,\vecV_{i+1},\vecW) \in R$ witnessed by
    $\sigma \in \pdsrules$, for $i \in [0,k]$,
    we add
    \[
        (\sigma \wedge \langle \uparrow
        \rangle\{\vecV_{i-1},\vecV_i,\vecV_{i+1}\}, \addclass{\vecW})
    \]
    to $\pds$.
\item For each $\sigma \in \pdsrules$, we add
    $(\sigma,\addclass{(-1,\blank,\star)})$ and
    $(\sigma,\addclass{(k+1,\blank,\star)})$ to
    $\pds$.
\end{enumerate}
Rule 1 first guesses the rule $\sigma$ to be applied to obtain the next
configuration $C_2$ of the NTM $\TM$. This rule is added as a class in the
child $v$ of the current node $u$ containing the current configuration $C_1$.
Rule 3 is then applied to obtain the left and right delimiter. After that,
configuration $C_2$ is written down in node $v$ by applying Rule 2 for $k$
times. In particular, the system will look up the rule $\sigma$ that was
guessed before and the content of the parent node $u$ (i.e. $C_1$). If we
let $S = [0,k] \times \talphabet \times \{q_{acc}\}$, then $S$ is not
$k$-redundant iff $\TM$ accepts $w$ within $k$ steps. Notice that this is
an fpt-preserving reduction. In particular, it runs in time $O(k^c n^d)$
for some constant $c$ and $d$.


    \section{Proof of $\PSPACE$-hardness in Theorem \ref{th:pspace}}
We reduce to the $k$-redundancy problem, when $k$ forms part of the
input, from membership of a linear bounded Turing machine. Note, since
the runs of a linear bounded Turing machine may be of exponential length, the
encoding in the previous section no longer works.

Suppose the input TM is $\TM =
(\oalphabet,\talphabet,\blank,\controls,\pdsrules,q_0,q_{acc})$,
where $\oalphabet = \{a,b\}$ is an input alphabet, $\talphabet =
\{a,b,\blank\}$ is the tape alphabet, $\blank$ is the blank symbol, $\controls$
is the set of control states, $\pdsrules \subseteq (\controls \times \talphabet)
\times (\controls \times \talphabet \times \{L,R\})$ is the transition function
($L$/$R$ signify go to left/right), $q_0 \in \controls$ is the initial
state, and $q_{acc}$ is the accepting state. We are only interested in runs
of $\TM$ that use $N := dn$ space for some constant $d$. Each configuration of
$\TM$ on $w$ can be represented as a word of length $N+1$
in the alphabet $\Omega := \Omega_N$, which was defined in the proof of
$\W[1]$-hardness. The initial configuration $I_0(w)$ is defined to be
\[
    (0,a_0,c_0)(1,a_1,c_1)\cdots (N,a_N,c_N)
\]
where $a_0a_1\cdots a_N = \blank w \blank^{(d-1)n}$, $c_0 = q_0$, and
$c_1 = \cdots = c_N = \star$. In the following, we will also use extension
$\Omega_e := \Omega_{N,e}$ of $\Omega$ and
the logspace-computability of the 4-ary relation $\Delta \subset \Omega_e^4$
defined in that proof.

We now construct our rewrite system $\trs$, which works over the
alphabet $\Xi := \{R,0,1\} \cup \Omega_e$. Let $M :=
2(N+\log|\controls|+\log N)$.
Firstly, it suffices to consider runs
of $\TM$ of at most length $3^N \times |\controls| \times N <
2^M$ since exceeding this length there must
be a repeat of configurations by the pigeonhole principle. Our initial
configuration is $\tree_0 = (\treedomain_0,\treelabeling_0)$, where
$\treedomain = \{\epsilon\}$ and $\treelabeling_0(\epsilon) = R$.

We start
by adding the following rules to $\trs$: $(g,\addchild{0})$ and
$(g,\addchild{1})$, where $g = (\langle \uparrow\rangle)^i R$ and $i \in
[0,M-1]$. These rules simply build trees of height $M$, into which the binary
tree of height $M$ can be embedded. The label in any path from the root to a
leaf node gives us a number $i$ in $[0,2^M]$ written in binary, which will in
turn have a child representing the $i$th configuration of $\TM$ from $I_0(w)$
(which is unique since $\TM$ is deterministic).

We now add the rule \mh{also removed reference to T2} $(g,\addchild{\{(-1,\blank,\star),(N+1,\blank,\star)\}})$, where
$g = (\langle \uparrow \rangle)^M R$. This rule initialises the content
of each configuration (of a given branch).

\begin{notation}
    For a guard $g$, a class $c$, and direction $d \in \{\uparrow,\downarrow\}$, we
    write $c\langle d\rangle g$ to mean $c \wedge \langle d \rangle g$.
\end{notation}
The following set of rules defines the first configuration $I_0(w)$ in the
tree. For each $i \in [0,N]$, we add the rule $(g,\addclass{(i,a_i,c_i)})$
$g = \langle \uparrow \rangle (0
\langle \uparrow \rangle)^{M-1} R$.

We now add a set of rules for defining subsequent configurations in the
tree. The guard will inspect a previous configuration of the current
configuration. Based on this, it will add an appropriate content of each
tape cell. Each of these rules will be of the form $(g,\addchild{\vecV})$,
where $\vecV \in \Omega$. For each $(\vecV_1,\vecV_2,\vecV_3,\vecW)$ and
$i \in [0,M-1]$, we add a rule $(g,\addchild{\vecW})$ where
$g := \top \langle \uparrow \rangle \sigma_i \langle
\downarrow \rangle \{\vecV_1,\vecV_2,\vecV_3\}$ where
$\sigma_i = (0\langle \uparrow \rangle)^i 1 \langle \uparrow \rangle \top
\langle \downarrow \rangle 0 (\langle \downarrow \rangle 1)^i$. The guard
$g$ selects precisely the nodes containing the $j$th configuration of $\TM$
(there can be at most one $j$th configuration of $\TM$ since $\TM$ is
deterministic), where $j$ is a number whose binary representation ends with
$10^i$. The guard traverses the tree upward trying to find the first bit
that is turned on (i.e. of form $10^i$) and then looks at the complement of this
bit pattern (i.e. of form $01^i$).

Finally, if we set $S = \{ (i,a,q_F) : i \in [0,N], a \in \talphabet \}$,
then $S$ is not $(M+1)$-redundant iff $\TM$ accepts $w$. The reduction
is easily seen to run in polynomial time. This completes our reduction.

    \section{Proof of $\EXP$-hardness in Theorem \ref{th:exp}}
We show that the redundancy problem for $\trs$ is $\EXP$-hard. To this end,
we reduce membership for alternating linear bounded Turing machines, which
is $\EXP$-complete. An alternating Turing machine is a tuple $\TM = 
(\oalphabet,\talphabet,\blank,\controls,\pdsrules,q_0,q_{acc})$,
where $\oalphabet = \{0,1\}$ is the input alphabet, $\talphabet = 
\{0,1,\blank\}$ is the tape alphabet, $\blank$ is the blank symbol, $\controls$
is the set of control states, $\pdsrules \subseteq (\controls \times \talphabet)
\times (\controls \times \talphabet \times \{L,R\})^2 \times \{\wedge,\vee\}$
is the transition relation. Given an $\TM$-configuration $C$, a transition 
$((q,a),(q_L,a_L,d_L),(q_R,a_R,d_R),
s)$ spawns two new $\TM$-configurations $C_L$ and $C_R$ such that 
$C_L$ (resp. $C_R$) is obtained from $C$ by applying a normal NTM
transition $((q,a),(q_L,a_L,d_L))$ (resp. $((q,a),(q_R,a_R,d_R))$). In this
way, a run $\Path$ of $\TM$ is a binary tree, whose nodes are labeled by
$\TM$-configurations and additionally each internal (i.e.
non-leaf) node is labeled by a transition from $\pdsrules$. To determine
whether the run $\Path$ is accepting, we construct a boolean circuit from
$\Path$ as follows: (1) each internal node labeled by a transition 
$((q,a),(q_L,a_L,d_L),(q_R,a_R,d_R), s)$ is replaced by a boolean operator
$s$, and (2) assign 1 (resp. 0) to a leaf node
if the configuration is (resp. is not) in the state $q_{acc}$. The run
$\Path$ is accepting iff the output of this circuit is 1.

We are only interested in tape configurations 
of size $d.n$ for a constant integer $d$, where $n$ is the length of the
input word $w$. Therefore, each configuration of $\TM$ on input $w$ is a word in
$\CONF_N$ defined in the proof of $\W[1]$-hardness. As before, the initial
configuration $I_0(w)$ is
\[
    (0,a_0,c_0)(1,a_1,c_1)\cdots (N,a_N,c_N)
\]
where $a_0a_1\cdots a_N = \blank w \blank^{(d-1)n}$, $c_0 = q_0$, and
$c_1 = \cdots = c_N = \star$. As before, we will also use extension 
$\Omega_e := \Omega_{N,e}$ of $\Omega$ and put a delimiter at the left/right
end of each configuration. 
For each $\sigma = ((q,a),(q_L,a_L,d_L),(q_R,a_R,d_R),s) \in \pdsrules$
and $t \in \{L,R\}$, we define a relation $R_{\sigma,t} \subseteq
\Omega_e^4$ as the log-computable relation $R_{\sigma'}$ defined in the 
previous 
$\W[1]$-hardness proof, where $\sigma'$ is the NTM transition rule
$((q,a),(q_t,a_t,d_t))$.

Our rewrite system $\trs$ will work over the set 
$\CLASS := \{\Root,\Success,L,R\} \cup \Omega_e \cup \pdsrules$ of classes. 
The initial configuration is
$\tree_0 = (\treedomain_0,\treelabeling_0)$ with $\treedomain_0 = \{\epsilon\}$
and $\treelabeling_0(\epsilon) = \{\Root,(-1,\blank,\star),(N+1,\blank,\star)\} 
\cup \{ (i,a_i,c_i) : i \in [0,N]\}$. Define $g := \{\Root,\Success\}$,
which is the guard we want to check for redundancy. Roughly speaking, our 
rewrite system $\trs$ will guess an accepting run $\Path$ of $\TM$ on the
input word $w$. Each accepting configuration (at the leaf) will then be
labeled by $\Success$. This label $\Success$ will be propagated to the root
of the tree according to the $\TM$-transition that produce a $\TM$-configuration
in the tree. More precisely, we define $\trs$ by adding the following
rewrite rules:
\begin{enumerate}
\item
For each $t \in \{L,R\}$, 
\[
    \sigma = ((q,a),(q_L,a_L,d_L),(q_R,a_R,d_R),s)
    \in \pdsrules 
\] 
and $i \in [0,N]$,
add the rule 
\[
    ((i,a,q),\addchild{\{\sigma,t,(-1,\blank,\star),
    (N+1,\blank,\star)\}}
\] 
to $\trs$.
\item
For each $\sigma \in \pdsrules$, $t \in \{L,R\}$, 
and $(\vecV_1,\vecV_2,\vecV_3,\vecW) \in R_{\sigma,t}$,
add the rule
\[
    (\{t,\sigma\}\langle\uparrow\rangle\{\vecV_1,\vecV_2,\vecV_3\},\addclass{\vecW})
\]
to $\trs$.
\item
For each $i \in [0,N]$ and $a \in \salphabet$, add the rule
\[
    (\{(i,a,q_{acc}),\addclass{\Success}) 
\] 
to $\trs$.
\item
For each 
\[
    \sigma = ((q,a),(q_L,a_L,d_L),(q_R,a_R,d_R),s)
    \in \pdsrules 
\] 
if $s = \vee$, then add the rule 
\[
    (\langle \downarrow\rangle\{\sigma,\Success\},
    \addclass{\Success}) \ .
\]
If $s = \wedge$, then add the rule 
\[
    (g,\addclass{\Success})
\]
where
\[
    g = \langle\downarrow\rangle\{\sigma,\Success,L\}\wedge
    \langle\downarrow\rangle\{\sigma,\Success,R\}
\]

\end{enumerate}
The first rewrite rule spawns a child labeled by an $\TM$-transition 
$\sigma$ and a symbol $t\in\{L,R\}$ to indicate whether its a left/right
child in the $\TM$-path to be guessed. The second rewrite rule
determines the configuration of the current node (because an $\TM$-rule
$\sigma$ has been guessed). The third rewrite rule adds a label $\Success$
to accepting configurations. The fourth rewrite rule propagates the
label $\Success$ upwards.

Our translation runs in polynomial-time. Furthermore,
it is easy to see that $\TM$ accepts $w$ iff $g$ is not redundant with respect
to the rewrite system $\trs$ and initial tree $\tree_0$. This completes our
reduction.

\section{Implementation Optimisations}

\subsection{Limiting the Number of Rules}

For each class $c \in \CLASS$ we construct a symbolic pushdown system that is
used to determine whether the class $c$ can be added at a certain node.  We
improve the performance of the pushdown analysis by first removing all rules of
the rewrite system with rules $\trs$ that cannot contribute towards the addition
of class $c$ at a node.  We then perform analysis with the reduced set of rules
$\trs'$.  Since a rule is added for each CSS selector, with intermediate rules
being introduced to simplify the guards on the rules, the removal of irrelevant
rules to a particular class can give significant gains.

We identify rules which may contribute to the addition of class $c$ by a
backwards fixed point algorithm.  We begin by including in $\trs'$ all rules
that either directly add the class $c$, or add a child to the tree.  The reason
we add all $\addchild{B}$ rules is because new nodes may lead to new guards
being satisfied, even if the label of the new child does not appear in the guard
(e.g. $\langle\uparrow\rangle a$ cannot be satisfied without a node labelled
with $a$ having a child).

We then search for any rules $\sigma = (g, \addclass{B}) \in \trs$ such that
there exists $c' \in B$ with $c'$ appearing in the guard of some rule in
$\trs'$.  We add all such $\sigma$ to $\trs'$ and iterate until a fixed point is
reached.  The full algorithm is given in Algorithm~\ref{alg:rules-restrict}

\begin{algorithm}
    \caption{\label{alg:rules-restrict}Restricting $\trs$}
    \begin{algorithmic}
        \Require {A class $c \in \CLASS$ and a set $\trs$ of rewrite rules.}
        \Ensure {$\trs'$ contains all rules relevant to the addition of $c$.}

        \State {$\trs' =
                    \begin{array}{l}
                        \setcomp{\sigma \in \trs}
                                {\sigma = (g, \addclass{B}) \land c \in B}\ \cup \\
                        \setcomp{\sigma \in \trs}{\sigma = (g, \addchild{B})}
                    \end{array}$}

        \Repeat
            \ForAll {$\sigma \in \trs'$ and classes $c'$ appearing in $\sigma$}
                \State {$\trs' = \trs' \cup \setcomp{(g, \addclass{B}) \in \trs}{c' \in B}$}
            \EndFor
        \Until {$\trs'$ reaches a fixed point}
    \end{algorithmic}
\end{algorithm}

\begin{proposition}
    Given a single-node tree $\tree_0$, assumption function $f$ and a class $c
    \in \CLASS$, the rewrite system $\trs$ has a positive solution to the
    class-adding problem $\langle \tree_0, f, c, \trs \rangle$ has a positive
    answer iff the class-adding problem $\langle \tree_0, f, c, \trs' \rangle$
    has a positive answer.
\end{proposition}
\begin{proof}
    The ``if'' direction is direct since $\trs$ contains all rules in $\trs'$.
    The ``only-if'' direction can be shown by induction over the length of a
    sequence of rule applications $\sigma_1, \ldots, \sigma_n$ where $\sigma_n$
    adds class $c$ to the tree.  Note, we relax the proposition for the
    induction hypothesis, allowing $\tree_0$ to be any arbitrary tree.

    In the base case $n = 0$ and the proof is immediate.  Hence, assume for the
    tree $\tree_1$ obtained by applying $\sigma_1$ to $\tree_0$ we have, by
    induction, a sequence of rule applications $\sigma'_1, \ldots, \sigma'_m$
    which is a subsequence of $\sigma_1,\ldots,\sigma_m$,
    resulting in the addition of class $c$.  We need to show there exists such a
    sequence beginning at $\tree_0$.  There are several cases.

    When $\sigma_1 = (g, \addchild{B})$ then $\sigma_1 \in \trs'$ and the
    sequence $\sigma_1, \sigma'_1, \ldots, \sigma'_m$ suffices.

    When $\sigma_1 = (g, \addclass{B})$ there are three cases.  If $c \in B$
    then $\sigma_1 \in \trs'$ and the sequence $\sigma_1$ suffices.  If there is
    some $c' \in B$ used in the matching of some guard $g'$ of a rule in
    $\sigma'_1, \ldots, \sigma'_m$, then, since $c'$ appears in a rule
    in $\trs'$ we have also $\sigma_1 \in \trs'$ and the sequence $\sigma_1,
    \sigma'_1, \ldots, \sigma'_m$ suffices.  Otherwise, no $c' \in B$ is used in
    the matching of some guard $g'$ in $\sigma'_1, \ldots, \sigma'_m$ and the
    sequence $\sigma'_1, \ldots, \sigma'_m$ suffices.
\end{proof}

\subsection{Limiting the Number of jMoped Calls}

By performing global backwards reachability analyses with jMoped we in fact only
need to perform one pushdown analysis per class $c \in \CLASS$.  Given a set of
target configurations, global backwards reachability analysis constructs a
representation of the set of configurations from which a target configuration
may be reached.

Thus, we begin with the set of target configurations $(q, a)$ where $a =
(b_1,\ldots,b_m) \in \{0,1\}^m$ with $b_i = 1$ and $b_i = 1$ indicates the
addition of class $c$ to the root node.  The global backwards reachability
analysis of this target set then gives us a representation of all configurations
that can eventually reach one of the targets, that is, add the class $c$.  In
particular, we obtain from jMoped a BDD representing all $\langle
\tree_v,f,c,\trs\rangle$ that have a positive answer to the class-adding
problem.  We can then use this BDD to determine whether a given class-adding
problem for $c$ represents a positive instance without having to use jMoped
again.

\subsection{Matching Guards Anywhere in the Tree}
We can check whether a guard $g$ is not matched anywhere in the tree by checking
whether $\langle \downarrow^\ast \rangle g$ is redundant.  However, during
the simplification step this results in the addition of a new class for each
guard that has to be propagated up the tree.

We avoid this cost by checking directly whether a guard may be matched anywhere
in the tree.  After simplification this amounts to checking whether a given
class may be matched anywhere in the tree.  Hence, after we have saturated
$\tree_v$, we do a final check to detect whether each class $c$ that does not
appear in the saturated tree might still be able to appear in a node created by
a sequence of $\addchild{B}$ rules.  To do this, we perform a similar class
adding check as above, except we test reachability of a configuration of the
form $(q, w)$ where the top character of the stack $w$ is $a = (b_1,\ldots,b_m)
\in \{0,1\}^m$ with $b_i = 1$.  That is, the node where $c$ was added does not
have to be the root node.  We obtain a single BDD per class as above, and check
it against all assumption functions $f$ that can represent nodes in the
saturated tree.

\section{HTML5 Translation}
\label{sec:html5-translation}

We informally describe here our process for extracting rules from HTML5
documents that use jQuery.  Note, this translation is for experimental purposes
only and does not claim to be a rigorous, systematic, or sound approach.  It is
merely to demonstrate the potential for extracting interesting rewrite rules
from real HTML5 and to obtain some representative benchmarks for our
implementation.  A robust dataflow analysis could in principle be built, and
remains an interesting avenue of future work.

We begin by describing how to translate complete
jQuery calls that do not depend on their surrounding code.  E.g.
\begin{minted}{javascript}
  $('.a').next().append('<div class="b"/>');
\end{minted}
does not depend on any other code in the document, whereas
\begin{minted}{javascript}
  $(x).next().append('<div class="b"/>');
\end{minted}
depends on the value of the \verb+x+ variable.  We will describe in a later
section how we track variables, although we will make mention in when
translating jQuery calls section when a particular variable receives a given
value.  First we describe the translation of CSS selectors and jQuery calls.

\subsection{CSS Selectors}

We translate the following subset of CSS selectors.
\begin{itemize}
    \item
        Each class, element type, or ID is directly translated to a class in our
        model with the same name.

    \item
        Conjunctions of guards are translated to sets of classes.  E.g.
        \verb+div.selected.item+ becomes $\{div, selected, item\}$.

    \item
        Descendant relations are encoded using $\langle\uparrow^+\rangle$.
        That is \verb+x y+, where \verb+x+ and \verb+y+ are selectors with
        translations into guards $g$ and $g'$ respectively becomes
        $\langle\uparrow^+\rangle g \wedge g'$.

    \item
        Child relations are encoded using $\langle\uparrow\rangle$.
        That is \verb+x > y+, where \verb+x+ and \verb+y+ are selectors with
        translations into guards $g$ and $g'$ respectively becomes
        $\langle\uparrow\rangle g \wedge g'$.
\end{itemize}

\subsection{jQuery Calls}

Given a jQuery call of the form
\begin{minted}{javascript}
  $(s).f(x)....g(y).h(z)
\end{minted}
where $s$ is a CSS selector string and \verb+f(x)....g(y).h(z)+ is a sequence of
jQuery calls, we can extract new rules as follows.  We recursively translate the
sequence of calls \verb+$(s).f(x)....g(y)+ where \verb+$(s)+ is the base case,
and obtain from the recursion a guard reflecting the nodes matched by the calls.
Then we generate a rule depending on the effect of \verb+h(z)+.

Note that the call to \verb+$(s)+ returns a jQuery object which contains within
it a stack of sets of nodes matched.  When \verb+s+ is just a CSS selector this
stack will contain a single element that is the set of nodes matched by
\verb+s+.  Each function in the sequence \verb+f(x)....g(y).h(z)+ receives as
its \verb+this+ object the jQuery object and obtains a new set of nodes derived
from the nodes on the top of the stack.  This new set is then pushed onto the
stack and passed to the next call in the sequence.  There is a jQuery function
\verb+end()+ that simply pops the top element from the stack and passes the
remaining stack to the next call.  We model this stack as a stack of guards
during our translation.

We start with the base case \verb+$(s)+.  In this case we translate the CSS
selector \verb+s+ into a guard $g$ and return the single-element stack $g$.
Note: this assumes \verb+s+ is a string constant.  Otherwise we set $g$ to be
true (i.e.  matches any node).

In the recursive case we have a call \verb+l.f(x)+ where \verb+l+ is a sequence
of jQuery calls and \verb+f(x)+ is a jQuery call.  We obtain a stack $\sigma$ by
recursive translation of \verb+l+ and then do a case split on \verb+f+.  In each
case we obtain a new stack $\rho$ which is returned by the recursive call.  In
some cases we also add new rules to our translated model.  We describe the
currently supported functions below.  In all cases, let $\sigma = \sigma' g$.
Note, our implementation directly supports both $\langle \downarrow^\ast \rangle
\varphi$ and $\langle \downarrow^+ \rangle \varphi$, and similarly for
$\langle \uparrow^\ast \rangle \varphi$ and $\langle \uparrow^+ \rangle \varphi$.

\begin{itemize}
    \item
        \verb+animate(...)+: we return $\rho = \sigma$.

    \item
        \verb+addBack()+: let $\sigma = \sigma'' g' g$, we set $\rho = \sigma''
        (g \vee g')$.

    \item
        \verb+addClass(+$c$\verb+)+: set $\rho = \sigma$ and add the rule $(g,
        \addclass{\{c\}})$ when $c$ is a string constant (otherwise add all
        classes).

    \item
        \verb+ajax(...)+: we attempt to resolve all functions $f$ in the
        argument to the call to constant functions (usually this coded directly
        in place by the programmer).  If we are able to do this, we translate
        $f$ with the \verb+this+ variable set to $\sigma$.  We return $\rho =
        \sigma$.  If we cannot resolve $f$ we ignore it (the body of $f$ will be
        translated elsewhere with \verb+this+ matching any node).

    \item
        \verb+append(+$s$\verb+)+: we attempt to obtain a constant string from
        $s$.  If $s$ is a concatenation of strings and variables, we attempt to
        resolve all variables to strings.  If we fail, we make the simplifying
        assumption that variables do not contain HTML code (this can be
        controlled via the command line) and omit them from the concatenation.
        If we are able to obtain a constant HTML string in this way we use a
        sequence of new rules and classes to build that tree at the nodes
        selected by $g$.  E.g. $(g, \addchild{\{x, a\}})$, $(\{x\},
        \addchild{\{b\}})$, $(\{x\}, \addchild{\{c\}})$ where $x$ is a fresh
        class adds a sub-tree containing a node with class ``a'' with two
        children with classes ``b'' and ``c'' respectively.  In the case where
        we cannot find a constant string (e.g. if we assume variables may
        contain HTML strings), we simply add rules to construct all possible
        sub-trees.

    \item
        \verb+bind(..., f)+: we attempt to resolve $f$ to a constant function
        (usually this is provided directly by the programmer).  If we are able
        to do this, we translate $f$ with the \verb+this+ variable set to
        $\sigma$.  We return $\rho = \sigma$.  If we cannot resolve $f$ we
        ignore it (the body of $f$ will be translated elsewhere with \verb+this+
        matching any node)

    \item
        \verb+blur(f)+: we treat this like a call to \verb+bind(..., f)+.

    \item
        \verb+children(+$s$\verb+)+: we first obtain from the CSS selector $s$ a
        guard $g'$ (or true if $s$ is not constant), then we set $\rho = \sigma
        (\langle \uparrow \rangle g \wedge g')$.

    \item
        \verb+click(+$f$\verb+)+: same as \verb+blur()+.

    \item
        \verb+closest(+$s$\verb+)+:  we first obtain from the CSS selector $s$ a
        guard $g'$ (or true if $s$ is not constant), then we set $\rho = \sigma
        (\langle \downarrow^\ast \rangle g \wedge g')$.

    \item
        \verb+data(...)+: we return $\rho = \sigma$.

    \item
        \verb+each(+$f$\verb+)+: we attempt to resolve $f$ to a constant, as
        with \verb+click+.  The function $f$ may take 0 or 2 arguments (the
        second being the matched element).  In both cases we translate $f$ with
        \verb+this+ set to $\sigma$.  If it has arguments, then the second
        argument is also set to $\sigma$.  We return $\rho = \sigma$.

    \item
        \verb+eq(...)+: we return $\rho = \sigma g$.

    \item
        \verb+fadeIn(...)+: we return $\rho = \sigma$.

    \item
        \verb+fadeOut(...)+: we return $\rho = \sigma$.

    \item
        \verb+focus(f)+: same as \verb+blur()+.

    \item
        \verb+end()+: we return $\rho = \sigma'$.

    \item
        \verb+find(+$s$\verb+)+: we obtain from the CSS selector $s$ a
        guard $g'$ (or true if $s$ is not constant), then we set $\rho = \sigma
        (\langle\uparrow^+\rangle g \wedge g')$.

    \item
        \verb+filter(...)+: we simply return $\rho = \sigma g$.

    \item
        \verb+first(...)+: we simply return $\rho = \sigma g$.

    \item
        \verb+has(+$s$\verb+)+: we obtain from the CSS selector $s$ a
        guard $g'$ (or true if $s$ is not constant), then we set $\rho = \sigma
        (\langle\downarrow^+\rangle g' \wedge g)$.

    \item
        \verb+hover(+$f_1$\verb+,+$f_2$\verb+)+: we attempt to resolve $f_1$ and
        $f_2$ to constant functions (else we ignore them to be translated
        elsewhere with less precise guards).  Both $f_1$ and $f_2$ take a single
        argument, which is the element being hovered over.  We translate both
        functions with the argument and the \verb+this+ variable set to
        $\sigma$.

    \item
        \verb+html(+$s$\verb+)+: this is handled like \verb+append+.

    \item
        \verb+next(+$s$\verb+)+: we obtain from the CSS selector $s$ a guard
        $g'$ (or true if $s$ is not constant or not supplied), then we set $\rho
        = \sigma (\langle\uparrow\rangle\langle\downarrow\rangle g \wedge g')$.

    \item
        \verb+nextAll(+$s$\verb+)+: we obtain from the CSS selector $s$ a guard
        $g'$ (or true if $s$ is not constant or not supplied), then we set $\rho
        = \sigma (\langle\uparrow\rangle\langle\downarrow\rangle g \wedge g')$.

    \item
        \verb+not(...)+: we ignore the negation and simply return $\rho = \sigma
        g$.

    \item
        \verb+on(+$e$\verb+,+$s$\verb+,+$f$\verb+)+: we obtain from the CSS
        selector $s$ a guard $g'$ (or true if $s$ is not constant or not
        supplied), and we attempt to resolve $f$ to a constant function (like
        \verb+click+).  We then translate $f$ with \verb+this+ set to $\sigma
        (\langle\uparrow^+\rangle g \land g')$.  We return with $\rho = \sigma$.

    \item
        \verb+parent(+$s$\verb+)+: we obtain from the CSS selector $s$ a guard
        $g'$ (or true if $s$ is not constant or not supplied), then we set $\rho
        = \sigma (\langle\downarrow\rangle g \wedge g')$.

    \item
        \verb+parents(+$s$\verb+)+: we obtain from the CSS selector $s$ a guard
        $g'$ (or true if $s$ is not constant or not supplied), then we set $\rho
        = \sigma (\langle\downarrow^+\rangle g \wedge g')$.

    \item
        \verb+prepend(+$s$\verb+)+: this is handled like \verb+append+.

    \item
        \verb+prev(+$s$\verb+)+: we obtain from the CSS selector $s$ a guard
        $g'$ (or true if $s$ is not constant or not supplied), then we set $\rho
        = \sigma (\langle\uparrow\rangle\langle\downarrow\rangle g \wedge g')$.

    \item
        \verb+prevAll(+$s$\verb+)+: we obtain from the CSS selector $s$ a guard
        $g'$ (or true if $s$ is not constant or not supplied), then we set $\rho
        = \sigma (\langle\uparrow\rangle\langle\downarrow\rangle g \wedge g')$.

    \item
        \verb+remove(+$s$\verb+)+: this function is ignored and we return $\rho
        = \sigma$.

    \item
        \verb+removeClass(+$s$\verb+)+: this function is ignored and we return
        $\rho = \sigma$.

    \item
        \verb+resize(f)+: same as \verb+blur()+.

    \item
        \verb+show(...)+: we return $\rho = \sigma$.

    \item
        \verb+slideUp(...)+: we return $\rho = \sigma$.

    \item
        \verb+slideDown(...)+: we return $\rho = \sigma$.

    \item
        \verb+stop(...)+: we return $\rho = \sigma$.

    \item
        \verb+val(...)+: we return $\rho = \sigma g$.
\end{itemize}

We also ignore the following functions that do not return a jQuery object or
manipulate the DOM tree: \texttt{extend()}, \texttt{hasClass()},
\texttt{height()}, \texttt{is()}, \texttt{width()}.

\subsection{More Complex JavaScript}

In reality, it is rare for jQuery statements to be completely independent of the
rest of the code.  We extend our translation by identifying common jQuery
programming patterns to enable the generation of more precise rules.  This is,
in effect, a kind of ad-hoc dataflow analysis over the syntax tree of the
JavaScript.

We remark that current JavaScript static analysers do not track the kind of
information needed to build meaningful guards from the jQuery calls.  Indeed,
implementing a systematic dataflow analysis for propagating jQuery guards
remains an interesting avenue of future research.

Often a call will take the form
\begin{minted}{javascript}
      $(x).f(y)
\end{minted}
for some variable \verb+x+.  In the worst case one can assume any value for
guards derived from \verb+x+ by using the guard $\top$.  In other cases it is
possible to be more precise.

When translating a call such as
\begin{minted}{javascript}
    $('.a').each(function () {
        $(this).addClass('b');
    ));
\end{minted}
for example, is it easy to infer that the selection returned by the \verb+this+
variable is all elements with the class ``a''.  Hence, when we are translating
constant functions such as these, we track the guards in the \verb+this+
variable as described in the previous section.  Similarly, if we wrote
\begin{minted}{javascript}
    $('.a').each(function (i, e) {
        $(e).addClass('b');
    ));
\end{minted}
we pass the stack of guards from the \verb+$('a')+ selector to the \verb+e+
variable, and are thus able to interpret \verb+$(e)+.

Note, we are also able to translate \verb+$(+$s$\verb+, x)+ where \verb+x+ is
\verb+this+ or some variable known to hold a jQuery object with stack $\sigma g$
by using the guard
\[
    \langle \uparrow^+ \rangle g \wedge g'
\]
where $g'$ is the guard obtained from the selector $s$.  Similarly
\verb+$(document)+ can be translated to the root node of the HTML.

We can further improve the analysis by tracking variables declared in the
JavaScript.  To this end, we identify all variables that are assigned outside of
a loop or conditional variable (i.e. assigned only once) and remember their
value for use in later translations.  This captures common cases such as
\begin{minted}{javascript}
    var eles = $('myelements');
    ...
    eles.addClass('b');
    eles.append('<p>Some test.</p>');
\end{minted}

Finally, we also provide some support for user-defined jQuery functions.  In a
preprocessing step, we look for lines of the form
\begin{minted}{javascript}
    $.fn.f = function (args) { ... };
\end{minted}
and keep a map from \verb+f+ to the function definition.  Thus, when we find
elsewhere in the code a jQuery call of the form
\begin{minted}{javascript}
    ...f(params)...
\end{minted}
we can translate the function body of \verb+f+ with the \verb+this+ variable set
to the current jQuery stack, any parameters can be passed to the arguments of
the function.  Note, if we discover a line of \verb+f+ which overwrites a
parameter variable, we will forget the value passed from the call and use a top
value instead.

\end{document}